\documentclass[12pt,a4paper]{article}

\usepackage{amsfonts}
\usepackage{amsmath}
\usepackage{amsbsy}
\usepackage{theorem}
\usepackage{graphicx}
\usepackage{amssymb}
\usepackage{epstopdf}

\newtheorem{theorem}{Theorem}

\newtheorem{remark}{Remark}

\begin{document}

\title{\bfseries \normalsize A SPATIAL FUNCTIONAL COUNT MODEL FOR HETEROGENEITY ANALYSIS IN TIME}

\date{}
\author{\normalsize  A. Torres--Signes, M.P. Fr\'{\i}as,  J. Mateu and  M.D. Ruiz--Medina}
% \maketitle
\maketitle

\begin{abstract}      A spatial curve dynamical model framework is adopted for functional prediction of counts in a spatiotemporal log--Gaussian Cox process model. Our spatial functional estimation approach  handles both wavelet--based heterogeneity analysis in time, and spectral analysis in space. Specifically, model fitting is achieved by minimising the information divergence or relative entropy between the  multiscale model underlying  the data and the corresponding candidates  in the spatial spectral domain.
  A simulation study is carried out within the family of log--Gaussian Spatial Autoregressive $\ell^{2}$--valued  processes (SAR$\ell^{2}$  processes) to illustrate the  asymptotic properties of the proposed spatial functional estimators.  We apply our modelling strategy to spatiotemporal prediction  of respiratory  disease mortality.

\medskip

 \noindent Cox processes in Hilbert spaces;   Spatial functional estimation;  Spectral wavelet--based analysis
% \PACS{PACS code1 \and PACS code2 \and more}
\medskip

\noindent MSC code1 60G25;  60G60;  62J05;  MSC code2 62J10

\end{abstract}

\section{Introduction}
\label{Intro}
 Count and aggregated data can be generally found in problems of disease incidence, mortality,   population dynamics, or wildfire occurrences that span the scientific fields of Environmental Health, Ecology, Epidemiology, and Atmospheric Environment, to mention just a few. In such cases, stochastic modelling of counts allows for a deeper understanding and accurate predictions for risk assessment and management (see \cite{Choi03};  \cite{Christakos92}; \cite{Christakos98}; \cite{Christakos00}; \cite{Christakos05}; \cite{Christakos17};
\cite{Daley8808};  \cite{Diggleb}; \cite{He20}; \cite{Illian08}, and the references therein).

In most of these cases, the term aggregated point process data (or aggregated data, for short) is used to refer to  discretely observed data which in reality most likely arose from an
underlying spatially- or spatiotemporally-continuous point process (see \cite{Diggleetal10a} and \cite{Taylor18}). These later authors argue that it is possible to fit a discrete model and obtain spatially- or spatiotemporally-continuous inference via spatial prediction.
We refer  the reader to  \cite{MollerWaa04} for background material on spatial point processes and the corresponding theoretical details.

In particular, the family of spatial Cox processes (see \cite{Cox55},  \cite{Grandell76})  has been extensively considered in point pattern analysis.
 The log--normal intensity model adopted here provides a  flexible modelling framework  (see \cite{Diggle}; \cite{Gonzalez16}, and \cite{Moller98}, among others). Its complete characterisation by the intensity and pair correlation functions  makes possible its application to different environmental fields (see, e.g., \cite{Rathbun94} in pine forest; \cite{Serra14} in wildfire occurrences; \cite{Waller97}; \cite{WMZ} in epidemic dynamics modelling, or \cite{Li12} in disease mapping). Extended models can be found, for instance,  in  \cite{Moller14};     \cite{Simpson16}; and \cite{Waagepetersen16}.
 It is well-known that log--Gaussian Cox processes allow the application of  parametric (likelihood, pseudo-li\-ke\-li\-hood,  composite likelihood),   semi-parametric, and  classical and Bayesian estimation methodologies,  avoiding  biased estimations,  as observed in kernel estimators (see \cite{Baddeleyeta06}; \cite{Diggleetal10b}; \cite{GB18};  \cite{Guan06};    \cite{Jalilian19}, to mention a few).

The distribution of the hidden environmental  fields driving the counts usually displays significant variability and uncertainties across
space and time. The characterisation of these fields    depends on the spatial scale
at which the phenomenon is considered, that could be different from the measurement
scale. The effect of    heterogeneities  at different
geographical scales on the spatial distribution of counts has been already examined in \cite{Christakos01};  \cite{Congdon17}; and \cite{Li08}.
 Another issue to be addressed, when inference comes to play, is the size and resolution of the temporal window, quantifying  temporal rate fluctuations   at the spatial regions  (see, e.g.,  \cite{Banks07}; \cite{Kennedy95}; \cite{SalapAyca18}). The approach presented in this paper addresses this problem in a Functional Data Analysis (FDA) framework, incorporating spatial correlations between curve rate parameters, at  the considered regions. The resulting functional predictions reflect spatial point pattern evolution at any time.
 Note that  FDA techniques are well suited to estimate summary statistics, which are  functional in nature. In particular, point process data classification, based on second--order statistics, can be performed applying  FDA methodologies   (see, e.g., pp. 135--150 in \cite{Baddeleyeta06}, and \cite{Illian08}). However, FDA is a relatively new branch in point pattern analysis.  We note the contributions of \cite{WMZ}, where a functional statistical approach is adopted in the approximation of the  distribution  of the random event times observed over a fixed time interval, and the recent one by \cite{Cronie20}, where a new framework to handle functional marked point processes is derived.

   One of the most important  challenges in  point pattern analysis from a FDA framework is the suitable definition of the process that generates the points. An $\ell^{2}$-valued homogeneous Poisson process is introduced in \cite{Bosq14}, where its functional parameter estimation and prediction are addressed from both, a componentwise Bayesian and classical frameworks. The asymptotic efficiency and equivalence of both estimation approaches are also shown.  In \cite{Torres16}, sufficient conditions are derived  for the existence and proper definition of an $\ell^{2}$-valued temporal log-Gaussian Cox process, with infinite-dimensional log--intensity given by a Hilbert-valued Ornstein--Uhlenbeck  process. Its estimation is achieved using a discrete ARH(1) approximation of such process in time.

 The present paper establishes sufficient conditions to introduce a new class of spatial $\ell^{2}$-valued log-Gaussian Cox processes.   These conditions entail the corresponding random intensity process to live in a real separable Hilbert space. Note that, recently, in \cite{Frias2020}, under mild conditions, a new class of spatial Cox processes has been introduced, driven by  a  log-intensity process lying  in a real separable Hilbert space. However, its intensity process does not necessarily belongs to such a space.  This paper attempts to cover this gap. The derived conditions allow to perform a multiscale analysis of the functional variance of the random intensity process. The  range of
 temporal  fluctuations  is then analysed through different scales. In our case, we choose   a compactly supported wavelet basis. A more accurate fitting of  the local variability displayed by  curve data is obtained with  this  multiscale analysis. Note that,  B--splines bases have been widely used in Functional Data Analysis (FDA) preprocessing leading, in some cases, to an over--smoothing of the analysed curve data.

The present paper also proposes  an alternative  spectral--based  multiscale    spatial functional estimation methodology, in contrast with the Whittle-based parametric one adopted  in  \cite{Frias2020}. Indeed, this methodology involves the relative entropy minimization criterion, to obtain the optimal multiscale model, underlying  the data, in the spatial spectral domain, from the computation of the periodogram operator at different temporal resolution levels.  The properties of the derived multiscale estimators are analysed in the simulation study.  The validation results obtained in the real--data application illustrate the good properties of the estimation approach presented in the reconstruction of the log--intensity field at different temporal scales.

Summarising,  the main ingredients used in the introduction of a new class of multiscale spatial
log--Gaussian Cox processes in $\ell^{2}$ spaces can be found in Section \ref{s3b}.  The theoretical results for a  multiscale analysis of the  functional variance are provided in Section \ref{SSAS}.
In Section \ref{sec3vf}, a temporal multiresolution estimation  approach  is adopted in the spatial  spectral domain. The class of  log--Gaussian SAR$\ell^{2}$(1) intensity processes is considered in the implementation of this estimation framework.
  The multiscale analysis, and the asymptotic properties  of the proposed componentwise estimators, in the spectral domain, are illustrated  in  the simulation study carried out in Section \ref{s5}.
   The  introduced spatial functional estimation methodology is then implemented  for  prediction of respiratory  disease mortality,  in a real-data application in Section \ref{s6}.

\section{Spatial  log-Gaussian Cox processes in infinite dimensions}
\label{s3b}
Let  $(\Omega,\mathcal{A},P)$ be  the basic probability space, where all the random variables  are subsequently defined on. Denote by $\mathcal{H}$  an arbitrary  real separable Hilbert space of functions,  with the inner product $\left\langle \cdot,\cdot\right\rangle_{\mathcal{H}},$ and the associated norm $\|\cdot\|_{\mathcal{H}}.$
Let  $\mathbf{X}=\{X_{\mathbf{z}},\ \mathbf{z}\in \mathbb{R}^{d}\}$ be a spatial stationary zero--mean Gaussian random field,  with values in $\mathcal{H}.$
Hence,  $ \sigma^{2}=E\|X_{\mathbf{z}}\|_{\mathcal{H}}^{2}<\infty,$ and
$P\left[X_{\mathbf{z}}\in \mathcal{H}\right]=1,$
 for each $\mathbf{z}\in  \mathbb{R}^{d}.$ That is,  $X_{\mathbf{z}}$ defines a random element in $\mathcal{H},$ for every $\mathbf{z}\in \mathbb{R}^{d}.$

 The  nuclear cross--covariance  operator
 \begin{eqnarray}
\mathcal{R}_{\mathbf{z}-\mathbf{y}}^{\mathbf{X}}(f)(g)&=&
E\left(X_{\mathbf{z}}\otimes X_{\mathbf{y}}\right)(f)(g)=\left\langle E\left( X_{\mathbf{z}}\otimes X_{\mathbf{y}}\right)(f), g
\right\rangle_{\mathcal{H}},\ f,g\in \mathcal{H},
\nonumber\\
\label{SCO}
\end{eqnarray}
\noindent   defines the spatial functional dependence structure of the infinite--dimensional Gaussian random field $\mathbf{X}.$
 We have applied Riesz representation theorem to define $\mathcal{R}_{\mathbf{z}-\mathbf{y}}^{\mathbf{X}}(f)(g)$ as   the dual element of $\mathcal{R}_{\mathbf{z}-\mathbf{y}}^{\mathbf{X}}(f)$ acting on $g\in \mathcal{H},$ for every $f,g\in \mathcal{H}.$
Here, we are restricting our attention to the class of nuclear or trace operators, i.e.,  in the space $\ell^{1}(\mathcal{H}),$ satisfying
$$\|\mathcal{R}_{\mathbf{z}-\mathbf{y}}^{\mathbf{X}}\|_{\ell^{1}(\mathcal{H})}=\sum_{j=1}^{\infty}\left\langle \left(\left[\mathcal{R}_{\mathbf{z}-\mathbf{y}}^{\mathbf{X}}\right]^{\star}\mathcal{R}_{\mathbf{z}-\mathbf{y}}^{\mathbf{X}}\right)^{1/2}(\varphi_{j}),
\varphi_{j}\right\rangle_{\mathcal{H}}<\infty,$$
 \noindent  for any orthonormal basis $\{\varphi_{j}\}_{j\geq 1}$ in $\mathcal{H}.$

\begin{remark}
Note that the approach presented is focused on modelling spatial functional (curve) dependence and variability in an $\mathcal{H}$--framework. This is the reason why, in our initial assumptions, the infinite--dimensional spatial field $\mathbf{X}$ is considered to be previously detrended, under spatial homogeneity. We refer  the reader to  Section 4 in \cite{Bosq14}, for instance, where a componentwise approach is adopted in the estimation of the functional trend of an infinite--dimensional Gaussian population, under  classical and Bayesian frameworks.
\end{remark}

From (\ref{SCO}),  $\mathcal{R}_{\mathbf{0}}^{\mathbf{X}},$ with kernel $r_{\mathbf{0}}^{\mathbf{X}},$ is a self-adjoint (symmetric) trace operator, satisfying
\begin{equation} \mathcal{R}_{\mathbf{0}}^{\mathbf{X}}(\phi_{j})=\lambda_{j}(\mathcal{R}_{\mathbf{0}}^{\mathbf{X}})\phi_{j},\quad j\geq 1, \label{eid}
\end{equation}
\noindent where  $\{\phi_{j},\ j\geq 1\}$  denotes the orthonormal system of eigenvectors
of   $\mathcal{R}_{\mathbf{0}}^{\mathbf{X}}$ in $\mathcal{H}.$
For each $\mathbf{z}\in \mathbb{R}^{d},$ $X_{\mathbf{z}}$ admits the following orthogonal expansion in $\mathcal{L}^{2}_{\mathcal{H}}(\Omega,\mathcal{A},P)$
(see \cite{Angulo97}))
\begin{eqnarray}
X_{\mathbf{z}}&=&\sum_{j=1}^{\infty}\left\langle X_{\mathbf{z}},\phi_{j}\right\rangle_{\mathcal{H}}\phi_{j}=\sum_{j=1}^{\infty}X_{\mathbf{z}}(\phi_{j})\phi_{j}\nonumber\\
r_{\mathbf{0}}^{\mathbf{X}}&\underset{\mathcal{H}\otimes \mathcal{H}}{=}&\sum_{j=1}^{\infty}\lambda_{j}(\mathcal{R}_{\mathbf{0}}^{\mathbf{X}})\phi_{j}\otimes \phi_{j}.
\label{kle}
\end{eqnarray}
\noindent
That is, $$E\left\| X_{\mathbf{z}}-\sum_{j=1}^{M}\left\langle X_{\mathbf{z}},\phi_{j}\right\rangle_{\mathcal{H}}\phi_{j}\right\|_{\mathcal{H}}^{2}\to 0,\quad M\to \infty,$$  \noindent with  $E\left[ \left\langle X_{\mathbf{z}},\phi_{j}\right\rangle_{\mathcal{H}}\left\langle X_{\mathbf{z}},\phi_{p}\right\rangle_{\mathcal{H}}\right]=\delta_{j,p}\lambda_{j}(\mathcal{R}_{\mathbf{0}}^{\mathbf{X}}),$ $j\geq 1,$ for each $\mathbf{z}\in \mathbb{R}^{d}.$
Here, $\delta_{j,p}$ denotes the Kronecker delta function.

Assume that $\mathbf{X}$ is such that, for every $\mathbf{z}\in \mathbb{R}^{d},$ $X_{\mathbf{z}}$  almost surely (a.s.) has support in the  bounded temporal interval  $\mathcal{T}\subset \mathbb{R}_{+}.$ Define, for each fixed  $\mathbf{z}\in \mathbb{R}^{d},$

\begin{equation}
\Lambda_{\mathbf{z}}(t)=\exp\left( X_{\mathbf{z}}(t)\right)=\sum_{p=0}^{\infty }\frac{C_{p}}{p!}H_{p}(X_{\mathbf{z}}(t)),\quad \forall  t\in \mathcal{T},\quad  \mbox{a.s},
\label{pd}
\end{equation}
\noindent where the last equality follows from Hermite polynomial expansion in the space $L_{2}\left(\mathbb{R},\varphi (u)du\right),$ with $\varphi (u)=(1/\sqrt{2\pi})\exp\left(-u^{2}/2\right).$ Here,  $H_{p}$ denotes the $p$th Hermite polynomial, and $C_{p}$ is the associated coefficient of function $G(u)=\exp(u),$ by projection in the space $L_{2}\left(\mathbb{R},\varphi (u)du\right).$

The next condition on $\Lambda_{\mathbf{z}},$ $\mathbf{z}\in \mathbb{R}^{d},$ allows the introduction from (\ref{pd})  of our functional model for   the   spatial  counting random density in the $\mathcal{L}^{p}_{\mathcal{H}}(\Omega ,\mathcal{A}, P)$ sense, $p\geq 1.$

\medskip

\noindent \textbf{Condition C1}. Assume that for any bounded set  $A\in \mathcal{B}^{d}$ of the Borel
$\sigma$--algebra $\mathcal{B}^{d}$ of $\mathbb{R}^{d},$ the following almost surely (a.s.) integral is finite:
  \begin{equation}\boldsymbol{\Lambda }(A)=\int_{A}\int_{\mathcal{T}}\Lambda_{\mathbf{z}}(t)dtd\mathbf{z}<\infty,\quad  \mbox{(a.s)}.\label{ce}\end{equation}

   Given the observations $\Psi_{\mathbf{z},\omega_{0} }=\int_{\mathcal{T}}\Lambda_{\mathbf{z},\omega_{0} }(t)dt,$ $\mathbf{z}\in A\subset \mathbb{R}^{d},$
  for certain  $\omega_{0}\in \Omega,$
    the number of events $\boldsymbol{\mathcal{C}}(A),$ that occur, during the period $\mathcal{T},$  at the region  $A,$ follows a Poisson probability distribution with mean $\boldsymbol{\Lambda }(A).$
        Note that, the least--squares predictor of $\boldsymbol{\mathcal{C}}(A)$ is given by $\boldsymbol{\Lambda }(A) ,$ introduced in
  (\ref{ce}), for any bounded Borel set $A \in \mathcal{B}^{d}.$
  From  (\ref{pd}), equation (\ref{kle}) leads to the following expression of the second--order variation of $\Psi_{\mathbf{z}}:$
  \begin{eqnarray}
 E[\Psi_{\mathbf{z}}^{2}]&=&\sum_{p=0}^{\infty}\int_{\mathcal{T}\times \mathcal{T}}\frac{C_{p}(t)C_{p}(s)}{p!}\left[\sum_{j=1}^{\infty}\lambda_{j}(\mathcal{R}_{\mathbf{0}}^{\mathbf{X}})\phi_{j}\otimes \phi_{j}(t,s)\right]^{p}dtds.
\label{som}
\end{eqnarray}

  \section{Spatial second--order analysis at different temporal scales}
  \label{SSAS}
  Consider the special case where  $\mathcal{H}=L^{2}(\mathcal{T}),$ the space of square integrable functions on $\mathcal{T}.$

  \begin{theorem}
  \label{th1}
  Under \textbf{Condition C1}, if $\left\{ \phi_{j},\ j\geq 1\right\},$   in  equation (\ref{kle}), are uniformly bounded in $\mathcal{T},$ $\Psi_{\mathbf{z}}$ defines a spatial second--order random  density.

  \end{theorem}

  \begin{proof} From equation
  (\ref{som}), applying  Proposition 4.9 in p. 92 in \cite{Marinucci}, after considering   Cauchy--Schwarz inequality, in terms of the inner product introduced in Formula (4.7)  in p.89 of \cite{Marinucci}, Hermite expansion properties lead to
 \begin{eqnarray} E[\Psi_{\mathbf{z}}^{2}]&=&\sum_{p=0}^{\infty}\int_{\mathcal{T}\times \mathcal{T}}\frac{C_{p}(t)C_{p}(s)}{p!}\left[\sum_{j=1}^{\infty}\lambda_{j}(\mathcal{R}_{\mathbf{0}}^{\mathbf{X}})\phi_{j}\otimes \phi_{j}(t,s)\right]^{p}dtds\nonumber\\
 &\leq & \sum_{p=0}^{\infty}\frac{1}{p!}\int_{\mathcal{T}\times \mathcal{T}}\sqrt{E[\Lambda_{\mathbf{z}}(t)]^{2}E[\Lambda_{\mathbf{z}}(s)]^{2}}
 \sqrt{E[H_{p}(X_{\mathbf{z}}(t))]^{2}E[H_{p}(X_{\mathbf{z}}(s))]^{2}}\nonumber\\
 &\times &
 \left[\sum_{j=1}^{\infty}\lambda_{j}(\mathcal{R}_{\mathbf{0}}^{\mathbf{X}})\phi_{j}\otimes \phi_{j}(t,t)
 \sum_{j=1}^{\infty}\lambda_{j}(\mathcal{R}_{\mathbf{0}}^{\mathbf{X}})\phi_{j}\otimes \phi_{j}(s,s)\right]^{p/2}
 dtds
 \nonumber\\
 &=&\sum_{p=0}^{\infty}\frac{1}{p!}\int_{\mathcal{T}\times \mathcal{T}}\exp\left(r_{0}(t,t)/2+r_{0}(s,s)/2\right)[r_{0}(t,t)r_{0}(s,s)]^{p/2}
 \nonumber\\&\times &
 \left[\sum_{j=1}^{\infty}\lambda_{j}(\mathcal{R}_{\mathbf{0}}^{\mathbf{X}})\phi_{j}\otimes \phi_{j}(t,t)
 \sum_{j=1}^{\infty}\lambda_{j}(\mathcal{R}_{\mathbf{0}}^{\mathbf{X}})\phi_{j}\otimes \phi_{j}(s,s)\right]^{p/2}
 dtds\nonumber\\
 &=& \sum_{p=0}^{\infty}\frac{1}{p!}\left\{\int_{\mathcal{T}}\exp\left(r_{0}(t,t)/2\right)\left[\sum_{j=1}^{\infty}\lambda_{j}(\mathcal{R}_{\mathbf{0}}^{\mathbf{X}})\phi_{j}\otimes \phi_{j}(t,t)\right]^{p}dt\right\}^{2}\nonumber\\
 &\leq & |\mathcal{T}|^{2}\exp\left(2\mathcal{M}^{2}\|\mathcal{R}_{\mathbf{0}}^{\mathbf{X}}\|_{\ell^{1}(\mathcal{H})}\right) \left\{\sum_{p=0}^{\infty}\frac{\mathcal{M}^{2p}}{p!}\left[\sum_{j=1}^{\infty}\lambda_{j}(\mathcal{R}_{\mathbf{0}}^{\mathbf{X}})\right]^{p}\right\}^{2}\nonumber\\ &= & |\mathcal{T}|^{2}\exp\left(4\|\mathcal{R}_{\mathbf{0}}^{\mathbf{X}}\|_{\ell^{1}(\mathcal{H})}\mathcal{M}^{2}\right)<\infty,\nonumber\label{equb}
 \end{eqnarray}
 \noindent where $\mathcal{M}>0,$ is such that $\sup_{t\in \mathcal{T}}\left|\phi_{j}(t)\right|\leq \mathcal{M},$ for any $j\geq 1.$
  \end{proof}

  Let $\{ \psi_{j:k},\ k\in \Gamma_{j},\ j\in \mathbb{Z}\}$ be an orthonormal basis of wavelets, providing a multiresolution analysis of $L^{2}(\mathcal{T})$  (see, e.g., \cite{RuizAng02}). For each $\mathbf{z}\in \mathbb{R}^{d},$ the zero--mean Gaussian random coefficient sequence $\{ X_{\mathbf{z}}(\psi_{j:k}),\ k\in \Gamma_{j},\ j\in \mathbb{Z}\},$ with $X_{\mathbf{z}}(\psi_{j:k})=\left\langle X_{\mathbf{z}},\psi_{j:k}\right\rangle_{L^{2}(\mathcal{T})},$  $k\in \Gamma_{j},$ $j\in \mathbb{Z},$ has covariance $$E[X_{\mathbf{z}}(\psi_{j_{1}:k_{1}})X_{\mathbf{z}}(\psi_{j_{2}:k_{2}})]=\mathcal{R}_{\mathbf{0}}^{\mathbf{X}}(\psi_{j_{1}:k_{1}})(\psi_{j_{2}:k_{2}}),\quad k\in \Gamma_{j_{i}},\ j_{i}\in \mathbb{Z},\ i=1,2,$$
  \noindent providing a multiscale analysis of the curve  dependence structure at spatial location $\mathbf{z},$ through  the autocovariance operator $\mathcal{R}_{\mathbf{0}}^{\mathbf{X}}.$ In a similar way,   for any $\mathbf{z},\mathbf{y}\in \mathbb{R}^{d},$  a multiscale analysis is induced by
   $$E[X_{\mathbf{z}}(\psi_{j_{1}:k_{1}})X_{\mathbf{y}}(\psi_{j_{2}:k_{2}})]=\mathcal{R}_{\mathbf{z}-\mathbf{y}}^{\mathbf{X}}(\psi_{j_{1}:k_{1}})(\psi_{j_{2}:k_{2}}),\quad k\in \Gamma_{j_{i}},\ j_{i}\in \mathbb{Z},\ i=1,2,$$
  \noindent on the curve cross--dependence structure  between the spatial locations $\mathbf{z}$ and $\mathbf{y},$ through  the cross-covariance operator $\mathcal{R}_{\mathbf{z}-\mathbf{y}}^{\mathbf{X}}.$
The covariance structure of the log--Gaussian sequence $\{ \exp\left(X_{\mathbf{z}}(\psi_{j:k})\right),\ k\in \Gamma_{j},\ j\in \mathbb{Z}\}$ also displays  a multiscale  analysis in time, in the space $L^{2}(\mathcal{T}),$ of the curve spatial dependence structure of the spatial infinite--dimensional intensity process  \linebreak $\left\{\Lambda_{\mathbf{z}}(t),\ t\in \mathcal{T},\ \mathbf{z}\in \mathbb{R}^{d}\right\}.$ Note that, from  (\ref{pd}),   applying Parseval identity in $[L^{2}(\mathcal{T})]^{\otimes p},$ $p\geq 1,$ and   Cauchy--Schwarz inequality in $L^{2}(\mathcal{T}),$
\begin{eqnarray}&&  \sum_{j=1}^{\infty}\sum_{k\in \Gamma_{j}}E\left[\exp\left(X_{\mathbf{z}}(\psi_{j:k})\right)\exp\left(X_{\mathbf{y}}(\psi_{j:k})\right)\right]\nonumber\\
&&=\sum_{j=1}^{\infty}\sum_{k\in \Gamma_{j}}
\exp\left(\mathcal{R}_{\mathbf{0}}^{\mathbf{X}}(\psi_{j:k})(\psi_{j:k})
+ \frac{R_{\mathbf{z}-\mathbf{y}}^{\mathbf{X}}(\psi_{j:k})(\psi_{j:k})+R_{\mathbf{y}-\mathbf{z}}^{\mathbf{X}}(\psi_{j:k})(\psi_{j:k})}{2}\right)
\nonumber\\
&&= \sum_{j=1}^{\infty}\sum_{k\in \Gamma_{j}}\sum_{p_{1},p_{2},p_{3}}\frac{(1/2)^{p_{2}+p_{3}}}{p_{1}!p_{2}!p_{3}!} \sum_{h_{1},\dots, h_{p_{1}}}
\sum_{l_{1},\dots, l_{p_{2}}}
\sum_{q_{1},\dots, q_{p_{3}}}
\left|\lambda_{h_{1}}\left(\mathcal{R}_{\mathbf{0}}^{\mathbf{X}}\right)\cdots \lambda_{h_{p_{1}}}\left(\mathcal{R}_{\mathbf{0}}^{\mathbf{X}}\right)\right|
\nonumber\\ &&
\hspace*{1cm}\times \left|\lambda_{l_{1}}\left(\mathcal{R}_{\mathbf{z}-\mathbf{y}}^{\mathbf{X}}\right)\cdots \lambda_{l_{p_{2}}}\left(\mathcal{R}_{\mathbf{z}-\mathbf{y}}^{\mathbf{X}}\right)\right|\left|\lambda_{q_{1}}\left(\mathcal{R}_{\mathbf{y}-\mathbf{z}}^{\mathbf{X}}\right)\cdots \lambda_{q_{p_{3}}}\left(\mathcal{R}_{\mathbf{y}-\mathbf{z}}^{\mathbf{X}}\right)\right|
\nonumber\\&&\hspace*{1cm}\times \left|\phi_{h_{1}}(\psi_{j:k})\cdots \phi_{h_{p_{1}}}(\psi_{j:k})\right|^{2}\nonumber\\
&&\hspace*{1cm}\times\left|\psi_{l_{1}}^{\mathbf{z}-\mathbf{y}}(\psi_{j:k})\cdots
\psi_{l_{p_{2}}}^{\mathbf{z}-\mathbf{y}}(\psi_{j:k})
\varphi_{l_{1}}^{\mathbf{z}-\mathbf{y}}(\psi_{j:k})\cdots\varphi_{l_{p_{2}}}^{\mathbf{z}-\mathbf{y}}(\psi_{j:k})\right|
\nonumber\\
&&\hspace*{1cm}\times\left|\psi_{q_{1}}^{\mathbf{y}-\mathbf{z}}(\psi_{j:k})\cdots
\psi_{q_{p_{3}}}^{\mathbf{y}-\mathbf{z}}(\psi_{j:k})
\varphi_{q_{1}}^{\mathbf{y}-\mathbf{z}}(\psi_{j:k})\cdots\varphi_{p_{3}}^{\mathbf{y}-\mathbf{z}}(\psi_{j:k}) \right|\nonumber\\
&&\leq \sum_{p_{1},p_{2},p_{3}}\frac{(1/2)^{p_{2}+p_{3}}}{p_{1}!p_{2}!p_{3}!}
\left[\sum_{h=1}^{\infty}\left|\lambda_{h}\left(\mathcal{R}_{\mathbf{0}}^{\mathbf{X}}\right)\right|\right]^{p_{1}}
\left[\sum_{l=1}^{\infty}\left|\lambda_{l}\left(\mathcal{R}_{\mathbf{z}-\mathbf{y}}^{\mathbf{X}}\right)\right|\right]^{p_{2}}
\left[\sum_{q=1}^{\infty}\left|\lambda_{q}\left(\mathcal{R}_{\mathbf{y}-\mathbf{z}}^{\mathbf{X}}\right)\right|\right]^{p_{3}}
\nonumber\\
&&
=\exp\left(
\|\mathcal{R}_{\mathbf{0}}^{\mathbf{X}}\|_{\ell^{1}(\mathcal{H})}+\frac{1}{2}\left[
\|\mathcal{R}_{\mathbf{z}-\mathbf{y}}^{\mathbf{X}}\|_{\ell^{1}(\mathcal{H})}+\|\mathcal{R}_{\mathbf{y}-\mathbf{z}}^{\mathbf{X}}\|_{\ell^{1}(\mathcal{H})}\right]
\right)
<\infty,\quad \forall \mathbf{z},\mathbf{y}\in \mathbb{R}^{d},\label{normintp}
\end{eqnarray}
\noindent which implies that the series  $$\sum_{j=1}^{\infty}\sum_{k\in \Gamma_{j}}E\left[\exp\left(X_{\mathbf{z}}(\psi_{j:k})\right)\right]^{2}=
\sum_{j=1}^{\infty}\sum_{k\in \Gamma_{j}}\exp\left(\mathcal{R}_{\mathbf{0}}^{\mathbf{X}}(\psi_{j:k})(\psi_{j:k})\right)
$$ \noindent is convergent, for every $\mathbf{z}\in \mathbb{R}^{d}.$   In (\ref{normintp}), we have considered (\ref{kle}), i.e.,
\begin{equation}
\mathcal{R}_{\mathbf{0}}^{\mathbf{X}}=\sum_{h=1}^{\infty}\lambda_{h}\left(\mathcal{R}_{\mathbf{0}}^{\mathbf{X}}\right)\phi_{h}\otimes \phi_{h}.
\label{eqedac}
\end{equation}
\noindent Also, we have applied that, for any $\mathbf{z},$ $\mathbf{y}\in \mathbb{R}^{d},$  $\mathcal{R}_{\mathbf{z}-\mathbf{y}}^{\mathbf{X}}$ and $\mathcal{R}_{\mathbf{y}-\mathbf{z}}^{\mathbf{X}}$ are nuclear operators admitting a singular value  decomposition, given by
\begin{eqnarray}
\mathcal{R}_{\mathbf{z}-\mathbf{y}}^{\mathbf{X}} &=& \sum_{l=1}^{\infty}\lambda_{l}\left(\mathcal{R}_{\mathbf{z}-\mathbf{y}}^{\mathbf{X}}\right)\psi_{l}^{\mathbf{z}-\mathbf{y}}\otimes \varphi_{l}^{\mathbf{z}-\mathbf{y}}\nonumber\\\mathcal{R}_{\mathbf{y}-\mathbf{z}}^{\mathbf{X}} &=& \sum_{q=1}^{\infty}\lambda_{q}\left(\mathcal{R}_{\mathbf{y}-\mathbf{z}}^{\mathbf{X}}\right)\psi_{q}^{\mathbf{y}-\mathbf{z}}\otimes \varphi_{q}^{\mathbf{y}-\mathbf{z}}.
\label{svd}
\end{eqnarray}

\section{Multiresolution spatial functional estimation in the spectral domain}
\label{sec3vf}
This section introduces the spatial functional estimation approach adopted in the spatial spectral domain following a multiscale componentwise parametric framework. In the next section, the Spatial Autoregressive Hilbertian model of order one (SAR$\mathcal{H}(1)$ model) is first introduced, in a spatial curve  and spectral model frameworks.
\subsection{A spatial  curve  state space equation}

Let $\mathbf{X}=\{X_{\mathbf{z}},\ \mathbf{z}\in \mathbb{R}^{d}\}$ be the Gaussian spatial  curve process introduced in Section \ref{s3b}.
Without loss of generality, we restrict our attention  here to the case $d=2,$  and  $\mathcal{H}=L^{2}(\mathcal{T}),$ $\mathcal{T}=[0,1].$   Assume $\mathbf{X}$ obeys a  Spatial Autoregressive  Hilbertian  State Equation (SAR$\mathcal{H}(1)$ equation), as given in \cite{Ruiz11a}.    Thus,
  \begin{equation}
X_{p,q}=Y_{p,q}-R=L_{1}(X_{p-1,q})+L_{2}(X_{p,q-1})+L_{3}(X_{p-1,q-1})+\epsilon_{p,q},\quad (p,q)\in \mathbb{Z}^{2},\label{Xij}
\end{equation}
\noindent where
 $R\in \mathcal{H}$ is the  functional mean, that is estimated  applying the methodology proposed in \cite{Bosq14}, from a compactly supported orthonormal wavelet basis  $\{ \psi_{j:k},\ k\in \Gamma_{j},\ j\in \mathbb{Z}\}$ in $L^{2}([0,1]).$
   The autocorrelation operators $L_{i},$ $i=1,2,3,$ are assumed to be bounded on $L^{2}([0,1])$. Random fluctuations, introduced by the external force, are represented in terms of the  $L^{2}([0,1])$--valued zero--mean Gaussian innovation process  $\epsilon=\{\epsilon_{p,q}, (p,q)\in \mathbb{Z}^2\}.$ Under spatial homogeneity, this process displays constant functional variance  $E||\epsilon_{p,q}||^2_{L^{2}([0,1])}=\sigma^{2},$ through the spatial locations  $(p,q)\in \mathbb{Z}^{2}.$ The spatial functional dependence structure of $\mathbf{X}$ is represented in terms of a  nuclear covariance operator, given by
$\mathcal{R}_{p,q}^{\epsilon }=E\left(\epsilon_{p+k,q+l}\bigotimes \epsilon_{k,l}\right)=E\left(\epsilon_{p,q}\bigotimes \epsilon_{0,0}\right),$ for every   $(p,q), (k,l) \in \mathbb{Z}^{2}.$
 In the following, we will work under the assumption  of  $\{\epsilon_{p,q},\ (p,q)\in \mathbb{Z}^{2}\}$ being  a  strong Gaussian white noise in $L^{2}([0,1]).$  Hence, $\mathcal{R}_{p,q}^{\epsilon }=0,$ for $p\neq q.$ In our framework, equation  (\ref{Xij}) is interpreted as the  discrete approximation of a spatial functional  log--intensity process over continuous space,  by considering constant values    within the  quadrants of the regular grid defining the spatial observation network  (see, e.g.,  \cite{Rathbun94}, in the real--valued case).  See also \cite{Ogata88} on spline
function approximation, to represent
the first-order intensity of a marked
inhomogeneous Poisson point process.

 In the implementation of our  wavelet based estimation, in the spectral domain,   of the spatial functional dependence structure of $\{\Lambda_{\mathbf{z}}(\cdot ),$ $\mathbf{z}\in \mathbb{R}^{d}\},$  we work under the conditions assumed in  Propositions 3 and 4 in \cite{Ruiz11a}, for the existence of a unique stationary solution to equation (\ref{Xij}); additionally, we also consider the following assumption:

\medskip

 \noindent \textbf{Condition C2}.
  $\mathcal{R}_{p,q}^{\mathbf{X}}$ is such that
  $\sum_{(p,q)\in \mathbb{Z}^{2}}\|\mathcal{R}_{p,q}^{\mathbf{X}}\|_{l^{1}(\mathcal{H})}<\infty.$

\medskip

 Under \textbf{Condition C2}, the spectral density operator is given by
 \begin{equation}\mathcal{F}_{\omega_{1},\omega_{2} }:=\frac{1}{(2\pi)^{2}}\sum_{(p,q)\in \mathbb{Z}^{2}}\mathcal{R}_{p,q}^{\mathbf{X}}
 \exp\left(-i(p\omega_{1}+q\omega_{2})\right),\quad (\omega_{1},\omega_{2})\in [0,2\pi)\times [0,2\pi),\label{sdo}\end{equation}
 \noindent which is a trace non--negative
 self--adjoint   operator.

  For a given functional sample of size
  $N=S_{1}\times S_{2},$ $\{X_{p,q},\ p=1,\dots, S_{1},\ q=1,\dots, S_{2}\},$ its functional Discrete Fourier Transform (fDFT) is defined as
\begin{equation}\widetilde{X}_{\omega_{1},\omega_{2} }^{N}(\cdot):=\frac{1}{2\pi \sqrt{N}}\sum_{p=1}^{S_{1}}\sum_{q=1}^{S_{2}}
 X_{p,q}(\cdot)\exp\left(-i(p\omega_{1}+q\omega_{2})\right).\label{fdft}\end{equation}
 \noindent This  transform is linear, periodic and Hermitian.
 Under suitable cumulant kernel conditions (see Theorem 2.2 in \cite{Panaretos13}), the fDFT  (\ref{fdft}) at  frequencies $\omega_{1}:=\omega_{1,N}=0,$ $\omega_{2,N}:=\omega_{2}=\pi,$
 $\omega_{j,N}\in \left\{\frac{2\pi}{N},\dots, \frac{2\pi\left[(N-1)/2\right]^{-}}{N}\right\};$ $\omega_{j,N}\to \omega_{j},$ $N\to \infty,$ $j=3,\dots,J,$ converges, as $N\to \infty,$ to
 independent Gaussian elements in $L^{2}\left( [0,1],\mathbb{R}\right),$ for $j=1,2,$ and in $L^{2}\left( [0,1],\mathbb{C}\right),$ for $j=3,\dots,J,$ with respective covariance operators
 $\mathcal{F}_{\omega_{j}},$ $j=1,\dots,J$ (see equation (\ref{sdo})).

 From a functional sample of size $N,$
the periodogram operator  at frequency $(\omega_{1},\omega_{2})\in [0,2\pi)\times  [0,2\pi)$ is given by

\begin{eqnarray}&&\mathcal{I}_{\omega_{1},\omega_{2}}^{N}(\cdot, \cdot)\nonumber\\
  &&:=
\sum_{p=1}^{S_{1}}\sum_{q=1}^{S_{2}}\sum_{p^{\prime }=1}^{S_{1}}\sum_{q^{\prime }=1}^{S_{2}}\frac{X_{p,q}\otimes X_{p^{\prime },q^{\prime }}(\cdot, \cdot)\exp\left(-i(p-p^{\prime })\omega_{1}-(q-q^{\prime })\omega_{2}\right)}{(2\pi)^{2}N},\nonumber\\
\label{kpo}
\end{eqnarray}

\noindent or, equivalently by

\begin{equation}
\mathcal{I}_{\omega_{1},\omega_{2}}^{N}:=\widetilde{X}_{\omega_{1},\omega_{2} }^{N}
\otimes \overline{\widetilde{X}_{\omega_{1},\omega_{2} }^{N}.
}\label{po}\end{equation}

For a given orthonormal   basis of compactly supported wavelets $\{ \psi_{j:k},\ k\in \Gamma_{j},\ j\in \mathbb{Z}\}$ in $L^{2}(\mathcal{T}),$
from equations (\ref{fdft})--(\ref{po}),
\begin{eqnarray}&&
\widetilde{X}_{\omega_{1},\omega_{2} }^{N}(\psi_{j:k})=
\frac{1}{2\pi \sqrt{N}}\sum_{p=1}^{S_{1}}\sum_{q=1}^{S_{2}}
 X_{p,q}(\psi_{j:k})\exp\left(-i(p\omega_{1}+q\omega_{2})\right),\ k\in \Gamma_{j},\ j\in \mathbb{Z}\nonumber\\ \label{fdft2a}\\&&
 \mathcal{I}_{\omega_{1},\omega_{2}}^{N}(\psi_{j_{1}:k_{1}})(\psi_{j_{2}:k_{2}})=\widetilde{X}_{\omega_{1},\omega_{2} }^{N}
(\psi_{j_{1}:k_{1}}) \overline{\widetilde{X}_{\omega_{1},\omega_{2} }^{N}}(\psi_{j_{2}:k_{2}})\nonumber\\
&&=\sum_{p=1}^{S_{1}}\sum_{q=1}^{S_{2}}\sum_{p^{\prime }=1}^{S_{1}}\sum_{q^{\prime }=1}^{S_{2}}\frac{X_{p,q}(\psi_{j_{1}:k_{1}}) X_{p^{\prime },q^{\prime }}(\psi_{j_{2}:k_{2}})\exp\left(-i(p-p^{\prime })\omega_{1}-(q-q^{\prime })\omega_{2}\right)}{(2\pi)^{2}N},\nonumber\\
\label{fdft2c}
 \end{eqnarray}
\noindent for any $k_{i}\in \Gamma_{j_{i}},$ $j_{i}\in \mathbb{Z},$ $i=1,2.$ Thus, a multiscale analysis in time is considered in the spatial spectral domain.

\subsection{The estimation approach}
 \label{sec_par}

From (\ref{Xij})--(\ref{fdft2c}), define the diagonal wavelet parameter vector sequence \begin{eqnarray}&&\boldsymbol{\theta}_{j:k} = \left(\theta_{j:k,1},\theta_{j:k, 2},\theta_{j:k,3}\right)\label{ps}\\ &&=      \left(L_{1}(\psi_{j:k})(\psi_{j:k}), L_{2}(\psi_{j:k})(\psi_{j:k}), L_{3}(\psi_{j:k})(\psi_{j:k})\right)\in \Theta_{j:k}\subset \Theta,\ k\in \Gamma_{j},\ j\in \mathbb{Z}\nonumber
\end{eqnarray}
 \noindent
  and the multiresolution approximation  $$\left\{X_{p,q}(\psi_{j:k}),\ p=1,\dots,S_{1},\ q=1,\dots, S_{2},\  k\in \Gamma_{j},\ j\in \mathbb{Z}\right\}$$\noindent   in time of the  spatial sample information. Note that, here, for every  $k\in \Gamma_{j},$ $j\in \mathbb{Z},$ $\Theta_{j:k}$ is finite, and $\Theta=\cup_{j\in \mathbb{Z}}\cup_{k\in \Gamma_{j}}\Theta_{j:k}$ is a compact set. For  $k\in \Gamma_{j},$ $j\geq 1,$ we also assume that the true parameter value $\boldsymbol{\theta}_{0,j:k}$ always lies in the interior of $\Theta_{j:k},$  and our spatial spectral model is identifiable in the wavelet domain.

 For  any   node $k\in \Gamma_{j},$ at  resolution level $j\in \mathbb{Z},$  one can consider
    the  parameter
    estimator    $\widehat{\boldsymbol{\theta}}_{N,j:k} = \left(\widehat{\theta}_{N,j:k,1},\widehat{\theta}_{N,j:k, 2},\widehat{\theta}_{N,j:k, 3}\right)$  of $\boldsymbol{\theta}_{j:k} = \left(\theta_{j:k,1},\theta_{j:k, 2},\theta_{j:k,3}\right),$   computed from the loss function
    \begin{eqnarray}&&K_{j:k}(\boldsymbol{\theta}_{0,j:k},\boldsymbol{\theta}_{j:k}):=\int_{[0,2\pi)\times [0,2\pi)}f_{j:k}(\boldsymbol{\varpi},\boldsymbol{\theta}_{0,j:k})
\eta_{j:k}(\boldsymbol{\varpi})\log\frac{\Psi_{j:k}(\boldsymbol{\varpi},\boldsymbol{\theta}_{0,j:k})}{\Psi_{j:k}(\boldsymbol{\varpi},\boldsymbol{\theta}_{j:k})}
d\boldsymbol{\varpi}\nonumber\\
&&= U_{j:k}(\boldsymbol{\theta}_{j:k})- U_{j:k}(\boldsymbol{\theta}_{0,j:k}),\label{eqcontrastf}
\end{eqnarray}
    \noindent where $\boldsymbol{\theta}_{0,j:k}$ denotes the true parameter value, associated with  node $k$ at scale $j\in \Gamma_{j}.$ The    multiscale normalised spatial spectral density $$\left\{\Psi_{j:k}(\boldsymbol{\varpi},\boldsymbol{\theta}_{j:k}),\ k\in \Gamma_{j},\ j\in \mathbb{Z},\ \boldsymbol{\varpi}\in [0,2\pi)\times [0,2\pi) \right\}$$ \noindent  is obtained from the identities  \begin{eqnarray}f_{j:k}(\boldsymbol{\varpi},\boldsymbol{\theta}_{j:k})&=&\sigma^{2}(\boldsymbol{\theta}_{j:k})
    \Psi_{j:k}(\boldsymbol{\varpi},\boldsymbol{\theta}_{j:k})\nonumber\\
    &=&\left[\int_{[0,2\pi)\times [0,2\pi)}
f_{j:k}(\boldsymbol{\varpi},\boldsymbol{\theta}_{j:k})\eta_{j:k}(\boldsymbol{\varpi})d\boldsymbol{\varpi}\right]
\Psi_{j:k}(\boldsymbol{\varpi},\boldsymbol{\theta}_{j:k})\nonumber\\
f_{j:k}(\boldsymbol{\varpi},\boldsymbol{\theta}_{j:k})&=&\frac{\sigma^{2}_{\epsilon (\psi_{j:k})}}{2\pi^{2}} \left|1-\theta_{j:k,1}e^{i\varpi_{1}}-\theta_{j:k,2}e^{i\varpi_{2}}-\theta_{j:k,3}e^{i(\varpi_{1}+\varpi_{2})}\right|^{-2},
\label{sdob0}\end{eqnarray}
   \noindent for every $\boldsymbol{\varpi}=(\varpi_{1},\varpi_{2})\in [0,2\pi)\times [0,2\pi),$ with, for each  $k\in \Gamma_{j},$ and $j\in \mathbb{Z},$
   $\eta_{j:k}(\boldsymbol{\varpi})$ being a nonnegative symmetric spatial  function, such that $\eta_{j:k}(\boldsymbol{\varpi})f_{j:k}(\boldsymbol{\varpi},\boldsymbol{\theta}_{j:k})\in L_{1}\left([0,2\pi)\times [0,2\pi)\right),$  the space of absolute integrable functions on $[0,2\pi)\times [0,2\pi),$ for each  $\boldsymbol{\theta}_{j:k} \in \Theta_{j:k}\subset \Theta.$

   The loss functions in (\ref{eqcontrastf}) measure the discrepancy, at different temporal resolution levels, between the true spatial spectral parametric model \linebreak  $\Psi_{j:k}(\boldsymbol{\varpi},\boldsymbol{\theta}_{0,j:k}),$ underlying  the data, and the parametric candidates \linebreak $\Psi_{j:k}(\boldsymbol{\varpi},\boldsymbol{\theta}_{j:k}),$  $\boldsymbol{\theta}_{j:k}\in \Theta_{j:k}\subset \Theta,$ at node $k,$ within  the temporal variation  scale $j\in \mathbb{Z}.$
In the last identity in equation (\ref{eqcontrastf}), for each scale $j\in \mathbb{Z},$
\begin{equation}U_{j:k}(\boldsymbol{\theta}_{j:k}):=-\int_{[0,2\pi)\times [0,2\pi)}f_{j:k}(\boldsymbol{\varpi},\boldsymbol{\theta}_{0,j:k})
 \eta_{j:k}(\boldsymbol{\varpi})\log\Psi_{j:k}(\boldsymbol{\varpi},\boldsymbol{\theta}_{j:k})d\boldsymbol{\varpi},\ k\in \Gamma_{j}.\label{eqcfield}
 \end{equation}

 In practice, we can then consider the \emph{empirical  multiscale  functional}
   \begin{equation}\hat{U}_{N,j:k}(\boldsymbol{\theta}_{j:k}):=-\int_{[0,2\pi)\times [0,2\pi)}I_{N,j:k}(\boldsymbol{\varpi})
 \eta_{j:k}(\boldsymbol{\varpi})\log\Psi_{j:k}(\boldsymbol{\varpi},\boldsymbol{\theta}_{j:k})d\boldsymbol{\varpi},\label{ecf}\end{equation}
\noindent where $I_{N,j:k}(\boldsymbol{\varpi})=\mathcal{I}_{\varpi_{1},\varpi_{2}}^{N}(\psi_{j:k})(\psi_{j:k})$ denotes, as before, the multiscale   periodogram   introduced in (\ref{fdft2c}), for  $k\in \Gamma_{j},$ and $j\in \mathbb{Z}.$

  For each  $k\in \Gamma_{j},$ and  $j\in \mathbb{Z},$    $\eta_{j:k}$  must satisfy   suitable conditions (see Theorem 2.1 in \cite{Alomari17}), such that the loss function (\ref{eqcontrastf}) has a minimum at the true parameter value, for each node at any scale, and the following asymptotic behaviour holds (see, e.g.,  \cite{Alomari17}):
 \begin{equation}\hat{U}_{N,j:k}(\boldsymbol{\theta}_{j:k})-\hat{U}_{N,j:k}(\boldsymbol{\theta}_{0,j:k}) \to_{P_{0,j:k}} K_{j:k}(\boldsymbol{\theta}_{0,j:k},\boldsymbol{\theta}_{j:k}),\quad N\to \infty,\label{confcv}
 \end{equation}
\noindent for  each $\boldsymbol{\theta}_{j:k}\in \Theta_{j:k}\subset \Theta,$   where $P_{0,j:k}$ denotes the measure  associated with density function $f_{j:k}(\boldsymbol{\varpi},\boldsymbol{\theta}_{0,j:k}),$ for each $k\in \Gamma_{j},$ and  $j\in \mathbb{Z}.$ To minimize the divergence in (\ref{eqcontrastf}), in practice,  we can   compute the minimum of $\hat{U}_{N,j:k}(\boldsymbol{\theta}_{j:k})$ over $\boldsymbol{\theta}_{j:k}\in \Theta_{j:k},$  through the different nodes $k$ at each scale $j\in \mathbb{Z}.$ That is, we will consider the multiscale parameter estimators

\begin{equation}
\widehat{\boldsymbol{\theta }}_{N,j:k}= \mbox{arg}\ \min_{\boldsymbol{\theta}_{j:k}\in \Theta_{j:k}}\hat{U}_{N,j:k}(\boldsymbol{\theta}_{j:k}),\quad k\in \Gamma_{j},\ j\in \mathbb{Z}.
  \label{mce}
  \end{equation}

The same estimation  procedure, based on   the multiscale   periodogram  in (\ref{fdft2c}), is applied for the remaining coefficients   in the two--dimensional wavelet transforms of operators $L_{i},$ $i=1,2,3,$, including  the scaling function coefficients, with respect to the basis $\left\{\varphi_{j_{0}:\widetilde{k}},\ \widetilde{k}\in \Upsilon_{j_{0}}\right\}$  of the space $V_{0}\subset L^{2}([0,1])$. That is, similar estimators are computed for the parameters
\begin{eqnarray}&&\boldsymbol{\theta}_{j_{0}:\widetilde{k},\widetilde{l}} = \left(\theta_{j_{0}:\widetilde{k},\widetilde{l},1},\theta_{j_{0}:\widetilde{k},\widetilde{l}, 2},\theta_{j_{0}:\widetilde{k},\widetilde{l},3}\right)\nonumber\\ &&=      \left(L_{1}(\varphi_{j_{0}:\widetilde{l}})(\varphi_{j_{0}:\widetilde{k}}), L_{2}(\varphi_{j_{0}:\widetilde{l}})(\varphi_{j_{0}:\widetilde{k}}), L_{3}(\varphi_{j_{0}:\widetilde{l}})(\varphi_{j_{0}:\widetilde{k}})\right)\in \Theta_{j_{0}:\widetilde{k}}\times \Theta_{j_{0}:\widetilde{l}}\nonumber\\
&&\hspace*{1cm}
\Theta_{j_{0}:\widetilde{k}}\times \Theta_{j_{0}:\widetilde{l}}
\subset \Theta\times \Theta,\quad \widetilde{k},\widetilde{l}\in \Upsilon_{j_{0}}
\nonumber\\
&&\boldsymbol{\theta}_{j_{0}:\widetilde{k}; j:k} = \left(\theta_{j_{0}:\widetilde{k}; j:k,1},\theta_{j_{0}:\widetilde{k}; j:k, 2},\theta_{j_{0}:\widetilde{k}; j:k,3}\right)\nonumber\\ &&=      \left(L_{1}(\psi_{j:k})(\varphi_{j_{0}:\widetilde{k}}), L_{2}(\psi_{j:k})(\varphi_{j_{0}:\widetilde{k}}), L_{3}(\psi_{j:k})(\varphi_{j_{0}:\widetilde{k}})\right)\in \Theta_{j_{0}:\widetilde{k}}\times \Theta_{j:k}\nonumber\\
&&\hspace*{1cm}
\Theta_{j_{0}:\widetilde{k}}\times \Theta_{j:k}
\subset \Theta\times \Theta,\quad \widetilde{k}\in \Upsilon_{j_{0}};\ k\in \Gamma_{j}, \ j\geq j_{0}
\nonumber\\
&&\boldsymbol{\theta}_{j:k; j_{0}:\widetilde{k}} = \left(\theta_{j:k; j_{0}:\widetilde{k},1},\theta_{j:k; j_{0}:\widetilde{k}, 2},\theta_{j:k; j_{0}:\widetilde{k},3}\right)\nonumber\\ &&=      \left(L_{1}(\varphi_{j_{0}:\widetilde{k}})(\psi_{j:k}), L_{2}(\varphi_{j_{0}:\widetilde{k}})(\psi_{j:k}), L_{3}(\varphi_{j_{0}:\widetilde{k}})(\psi_{j:k})\right)\in \Theta_{j:k}\times \Theta_{j_{0}:\widetilde{k}} \nonumber\\
&&\hspace*{1cm}\Theta_{j:k}\times
\Theta_{j_{0}:\widetilde{k}}
\subset \Theta\times \Theta,\quad   k\in \Gamma_{j}, \ j\geq j_{0}; \   \widetilde{k}\in \Upsilon_{j_{0}}
\nonumber\\
&&\boldsymbol{\theta}_{j:k;l:h} = \left(\theta_{j:k; l:h,1},\theta_{j:k; l:h, 2},\theta_{j:k; l:h,3}\right)\nonumber\\ &&=      \left(L_{1}(\psi_{j:k})(\psi_{l:h}), L_{2}(\psi_{j:k})(\psi_{l:h}), L_{3}(\psi_{j:k})(\psi_{l:h})\right)\in \Theta_{j:k}\times \Theta_{l:h} \nonumber\\
&&\hspace*{0.5cm}\Theta_{j:k}\times
\Theta_{l:h}
\subset \Theta\times \Theta,\quad   k\in \Gamma_{j}, \ j\geq j_{0}; \    h\in \Gamma_{l}, \ l\geq j_{0};\ (j,k)\neq (h,l).
\nonumber\\\label{ps2}
\end{eqnarray}

  The resulting multiscale SARH$\ell^{2}$(1)  plug--in predictor is  computed, for any   spatial location  $(p,q),$ as

\begin{eqnarray}\widehat{X}_{N,p,q}(\cdot)&=&
\sum_{\widetilde{k}\in \Upsilon_{j_{0}}}\sum_{\widetilde{l}\in \Upsilon_{j_{0}}}\widehat{\theta}_{N,j_{0}:\widetilde{k},\widetilde{l},1}X_{p-1,q}(\varphi_{j_{0}:\widetilde{l}})\varphi_{j_{0}:\widetilde{k}}(\cdot)\nonumber\\
&+&\sum_{\widetilde{k}\in \Upsilon_{j_{0}}}\sum_{j\geq j_{0}}\sum_{k\in \Gamma_{j}} \widehat{\theta }_{N,j_{0}:\widetilde{k}; j:k,1} X_{p-1,q}(\psi_{j:k})\varphi_{j_{0}:\widetilde{k}}(\cdot)
\nonumber\\
&+&\sum_{j\geq j_{0}}\sum_{k\in \Gamma_{j}}
\sum_{\widetilde{k}\in \Upsilon_{j_{0}}} \widehat{\theta }_{N,j:k; j_{0}:\widetilde{k},1} X_{p-1,q}(\varphi_{j_{0}:\widetilde{k}})\psi_{j:k}(\cdot)
\nonumber\\
&+&\sum_{j\geq j_{0}}\sum_{k\in \Gamma_{j}}\sum_{l\geq j_{0}}\sum_{h\in \Gamma_{l}}
\widehat{\theta }_{N,j:k; l:h,1} X_{p-1,q}(\psi_{l:h})\psi_{j:k}(\cdot)
\nonumber\\
&+&
\sum_{\widetilde{k}\in \Upsilon_{j_{0}}}\sum_{\widetilde{l}\in \Upsilon_{j_{0}}}\widehat{\theta}_{N,j_{0}:\widetilde{k},\widetilde{l},2}X_{p,q-1}(\varphi_{j_{0}:\widetilde{l}})\varphi_{j_{0}:\widetilde{k}}(\cdot)\nonumber\\
&+&\sum_{\widetilde{k}\in \Upsilon_{j_{0}}}\sum_{j\geq j_{0}}\sum_{k\in \Gamma_{j}} \widehat{\theta }_{N,j_{0}:\widetilde{k}; j:k,2} X_{p,q-1}(\psi_{j:k})\varphi_{j_{0}:\widetilde{k}}(\cdot)
\nonumber\end{eqnarray}
\begin{eqnarray}
&+&\sum_{j\geq j_{0}}\sum_{k\in \Gamma_{j}}
\sum_{\widetilde{k}\in \Upsilon_{j_{0}}} \widehat{\theta }_{N,j:k; j_{0}:\widetilde{k},2} X_{p,q-1}(\varphi_{j_{0}:\widetilde{k}})\psi_{j:k}(\cdot)
\nonumber\\
&+&\sum_{j\geq j_{0}}\sum_{k\in \Gamma_{j}}\sum_{l\geq j_{0}}\sum_{h\in \Gamma_{l}}
\widehat{\theta }_{N,j:k; l:h,2} X_{p,q-1}(\psi_{l:h})\psi_{j:k}(\cdot)
\nonumber\\
&+&
\sum_{\widetilde{k}\in \Upsilon_{j_{0}}}\sum_{\widetilde{l}\in \Upsilon_{j_{0}}}\widehat{\theta}_{N,j_{0}:\widetilde{k},\widetilde{l},3}X_{p-1,q-1}(\varphi_{j_{0}:\widetilde{l}})\varphi_{j_{0}:\widetilde{k}}(\cdot)\nonumber\\
&+&\sum_{\widetilde{k}\in \Upsilon_{j_{0}}}\sum_{j\geq j_{0}}\sum_{k\in \Gamma_{j}} \widehat{\theta }_{N,j_{0}:\widetilde{k}; j:k,3} X_{p-1,q-1}(\psi_{j:k})\varphi_{j_{0}:\widetilde{k}}(\cdot)
\nonumber\\
&+&\sum_{j\geq j_{0}}\sum_{k\in \Gamma_{j}}
\sum_{\widetilde{k}\in \Upsilon_{j_{0}}} \widehat{\theta }_{N,j:k; j_{0}:\widetilde{k},3} X_{p-1,q-1}(\varphi_{j_{0}:\widetilde{k}})\psi_{j:k}(\cdot)
\nonumber\\
&+&\sum_{j\geq j_{0}}\sum_{k\in \Gamma_{j}}\sum_{l\geq j_{0}}\sum_{h\in \Gamma_{l}}
\widehat{\theta }_{N,j:k; l:h,3} X_{p-1,q-1}(\psi_{l:h})\psi_{j:k}(\cdot).
\label{SARHpp}
\end{eqnarray}

\noindent In practice, we select a finite number $D$  of scales, according to the  adopted discretisation step size in time,  in the preprocessing procedure involved in the  construction of our curve data set. Note that, as commented before, for a given scale  $j\in \{1,\dots,D\},$  the corresponding number of  nodes $k(j)$ is finite.

  \section{Simulation study}
\label{s5}

To illustrate the asymptotic properties of the formulated  multiscale estimators, an increasing  spatial curve  sample  size  sequence $N=100, 900, 2500, 4900,$ $8100,$ $12100, 16900, 22500,$ has been considered.
The   Haar  wavelet system has been selected for our implementation (see, e.g., \cite{Daubechies92}). In particular,  let $L_{3}=-L_{1}L_{2},$ and, as before,  $\mathcal{T}=[0,1].$  Operators $L_{1}$ and $L_{2}$ are defined in terms of the common eigenvectors
\begin{eqnarray}
\phi_{p} \left( t\right) &=&  \displaystyle \sin \left(
\pi p t\right),\quad t\in (0,1),\quad p\geq 1,\label{eigvalLap}
\end{eqnarray}
\noindent with $\phi_{p}(0)=\phi_{p}(1)=0.$  The corresponding  systems of eigenvalues  $\{\lambda_{pl},\ p\geq 1,\ l=1,2\}$ satisfy  conditions  (i)--(iii) in Proposition 3  of \cite{Ruiz11a}, for the existence of a unique stationary solution to the SAR$\mathcal{H}(1)$ equation.
Note that   the conditions assumed in Theorem \ref{th1}, and \textbf{Condition C2} also hold, under this scenario.
In the orthogonal decomposition
(\ref{kle}), we have considered the truncation parameter  $k_{N}=k_{22500}=[\ln\left(N\right)]^{-}=[\ln\left(22500\right)]^{-}=10,$ where we have selected the most unfavorable case (i.e., the largest truncation order corresponding to the functional sample size $N=22500$).
Table \ref{t1} displays
the  $k_{N}=10$ eigenvalues $\left\{\lambda_{pi},\ p=1,\ldots,10\right\}$  of operators $L_{i},$ $i=1,2.$

\begin{table}[t!] %***
\caption{Eigenvalues $\lambda_{p1}$, $\lambda_{p2},$
$p=1,\ldots,10$.}
\label{t1}\par
\vskip .2cm
\centerline{\tabcolsep=3truept \begin{tabular}{|lcc|}
\hline
&$\lambda_{p1}^{0}$&$\lambda_{p2}^{0}$\\\hline
$p=1$&0.300 &0.500 \\
$p=2$&0.270 &0.470 \\
$p=3$&0.230 &0.430 \\
$p=4$&0.200 &0.400 \\
$p=5$&0.170 &0.370 \\
$p=6$&0.130 &0.330 \\
$p=7$&0.100 &0.300 \\
$p=8$&0.030 &0.230 \\
$p=9$&0.010 &0.200 \\
$p=10$&0.005 &0.150.\\\hline
\end{tabular}}
\end{table}

  The  large--scale   sample  properties (the draft)  of  $\mathbf{X}$ are obtained by its projection onto  the space $V_{0}\subset L^{2}([0,1]),$
 generated by the  scaling functions $\left\{\varphi_{j_{0}:\widetilde{k}},\ \widetilde{k}\in \Upsilon_{j_{0}}\right\}$ at the coarser scale $j_{0}.$
  The sample local variability (\emph{details}) of $\mathbf{X}$ is reproduced at different resolution levels, by its projection onto the subspaces $W_{j}\subset L^{2}([0,1]),$ $j=j_{0},\dots,D,$ generated by the wavelet bases
$\left\{\psi_{j:k},\ k\in \Gamma_{j}\right\},$ $j=j_{0},\dots,D,$ respectively.  Figures \ref{D10}--\ref{D7} show the displayed  temporal variability at different scales of the generated curve data, over some of the nodes of a $30\times 30$ spatial regular grid  ($N=900$).

 In the estimation of  the multiscale parameters
(\ref{ps}) and (\ref{ps2}),   equations (\ref{fdft2c})--(\ref{mce}), and their non--diagonal counterparts  are respectively computed.
Function    $\eta $ is constant over the nodes of the $D$ scales considered in the two-dimensional wavelet transform of operators $L_{i},$ $i=1,2.$ In particular, the choice $\eta (\boldsymbol{\varpi })=|\varpi_{1}|^{2}|\varpi_{2}|^{2}$ has been made, for every  $\boldsymbol{\varpi }= (\varpi_{1},\varpi_{2})\in [0,2\pi)\times [0,2\pi).$  The average by scale of the  empirical mean quadratic  errors, associated with the multiscale parameter  estimators of (\ref{ps}) and (\ref{ps2}), for
$j_{0}=6,$ $D=9,$  based on $100$ generations of the functional samples  of size
 $$N=100, 900, 2500, 4900, 8100, 12100, 16900, 22500,$$
 \noindent are displayed  in Tables \ref{L1}--\ref{L2}.

\begin{figure}[h!]
\includegraphics[width=12cm,height=10cm]{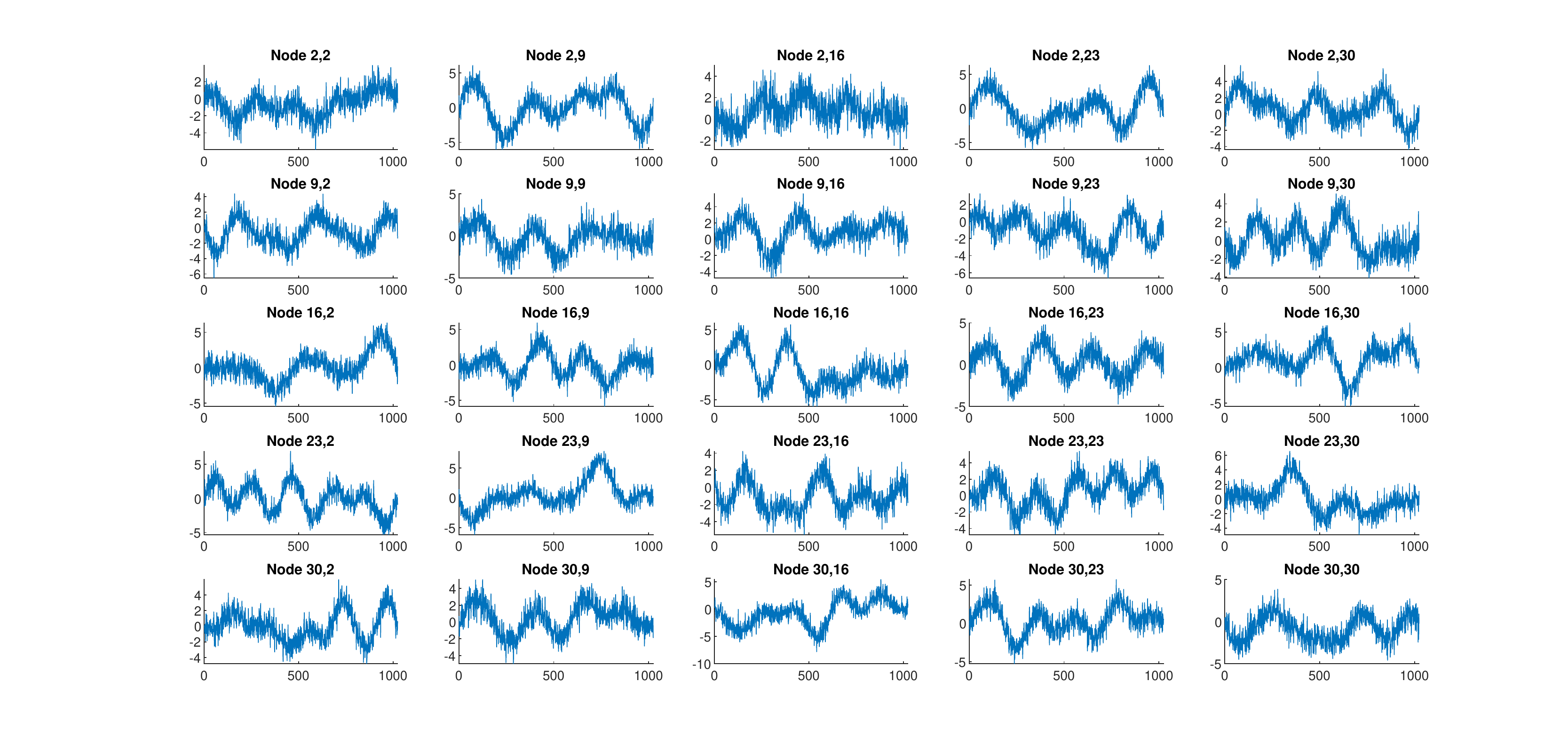}
\caption{Scale 10, $N=900.$ Curve data  over some nodes of a $30\times 30$ spatial regular  grid}\label{D10}
 \end{figure}

\begin{figure}[h!]
\includegraphics[width=12cm,height=10cm]{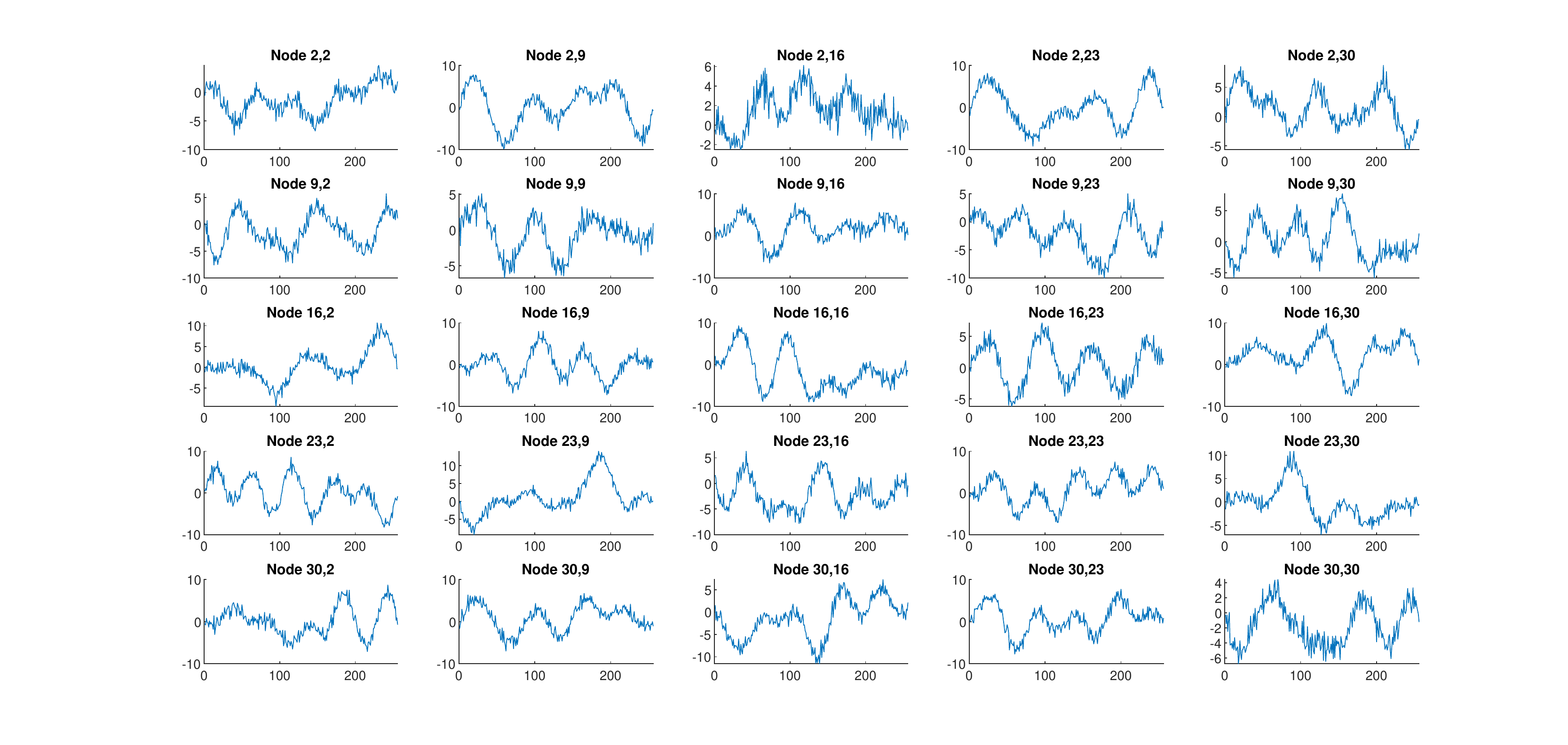}
\caption{Scale 9, $N=900.$  Curve data  over some nodes of a $30\times 30$ spatial regular  grid}\label{D9}
 \end{figure}

 \begin{figure}[h!]
\includegraphics[width=12cm,height=10cm]{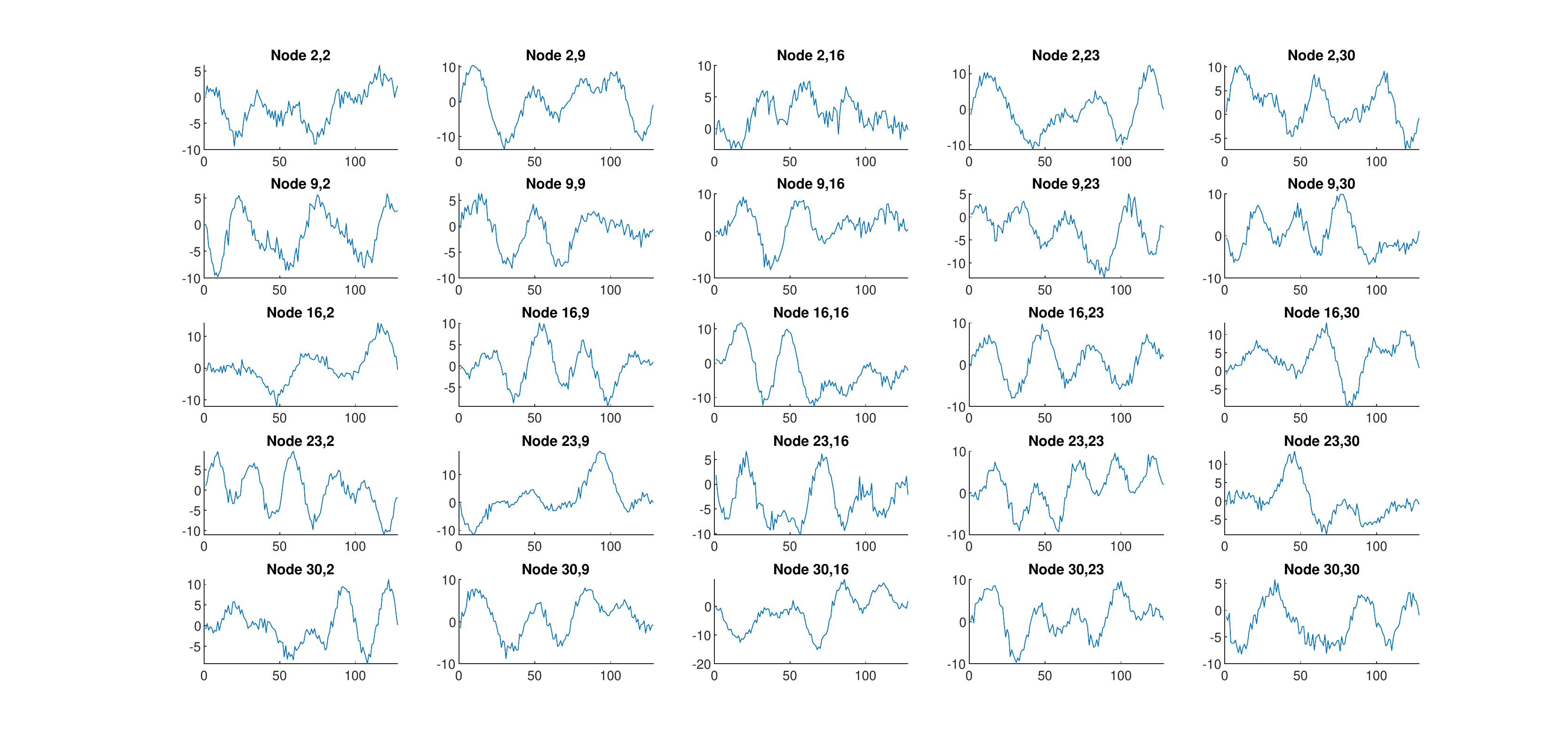}
\caption{Scale 8, $N=900.$  Curve data  over some nodes of a $30\times 30$ spatial regular  grid}\label{D8}
 \end{figure}

 \begin{figure}[h!]
\includegraphics[width=12cm,height=10cm]{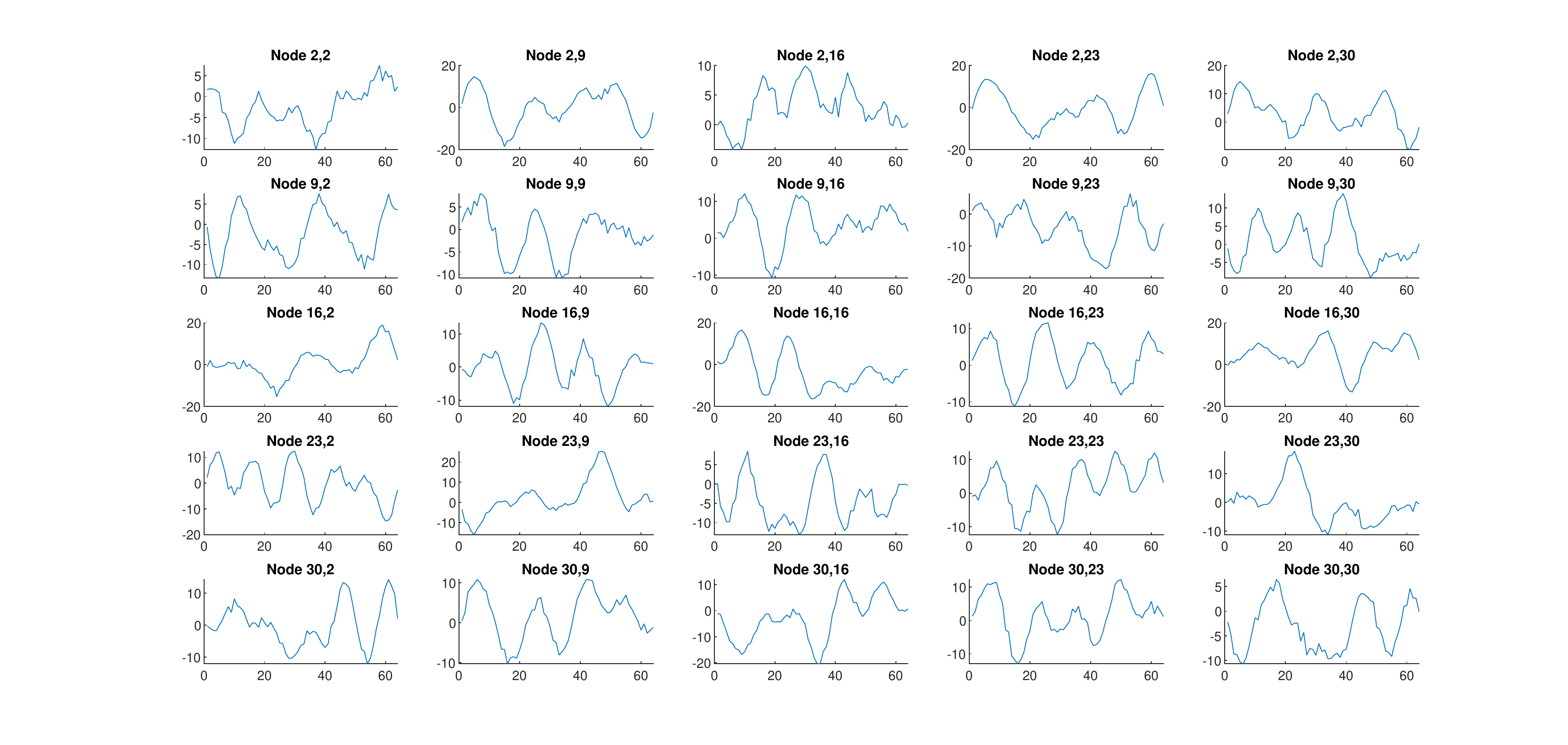}
\caption{Scale 7, $N=900.$  Curve data   over some nodes of a $30\times 30$ spatial regular  grid}\label{D7}
 \end{figure}

 \begin{table}[hptb]
\caption{Empirical mean quadratic errors by scale of the parameter estimators for $L_{1},$  based on 100 generations of the functional samples of size $N=100, 900, 2500, 4900, 8100, 12100, 16900, 22500$. Scales  $j=6, 7, 8, 9,$ and truncation parameter $k_{N}=10$ are considered}
\begin{center}
\begin{tabular}{ccccc}
\hline
$N$& Scale 6& Scale 7&Scale 8&Scale 9\\\hline
100&6.828e-02 & 4.158e-03  &2.606e-04&1.630e-05\\
900&9.171e-03 & 5.678e-04  &3.558e-05&2.226e-06\\
2500&3.540e-03 & 2.174e-04  &1.362e-05&8.520e-07\\
4900&1.930e-03 & 1.309e-04  &8.205e-06&5.132e-07\\
8100&1.151e-03 & 7.640e-05  &4.789e-06&2.995e-07\\
12100&9.036e-04 & 5.234e-05  &3.280e-06&2.052e-07\\
16900&6.120e-04 & 3.712e-05  &2.326e-06&1.455e-07\\
22500&4.663e-04 & 3.066e-05  &1.921e-06&1.202e-07\\
\hline
\end{tabular}
\end{center}\label{L1}
\end{table}

\begin{table}[hptb]
\caption{Empirical mean quadratic errors by scale of the parameter estimators,  for $L_{2},$  based on 100 generations of the functional samples of size $N=100, 900, 2500, 4900, 8100, 12100, 16900, 22500$. Scales  $j=6, 7, 8, 9,$ and truncation parameter $k_{N}=10$ are considered}
\begin{center}
\begin{tabular}{ccccc}
\hline
$N$&Scale 6&Scale 7&Scale 8&Scale 9\\\hline
100&7.080e-02 & 4.476e-03  &2.806e-04&1.755e-05\\
900&9.360e-03 & 5.916e-04  &3.708e-05&2.319e-06\\
2500&3.642e-03 & 2.302e-04  &1.443e-05&9.025e-07\\
4900&1.787e-03 & 1.129e-04  &7.078e-06&4.427e-07\\
8100&1.093e-03 & 6.909e-05  &4.330e-06&2.708e-07\\
12100&9.795e-04 & 6.188e-05  &3.878e-06&2.425e-07\\
16900&6.367e-04 & 4.022e-05  &2.521e-06&1.577e-07\\
22500&4.472e-04 & 2.825e-05  &1.771e-06&1.107e-07\\
\hline
\end{tabular}\label{L2}
\end{center}
\end{table}

  Figure \ref{fig.mse} shows the empirical mean quadratic errors, associated with   the estimates
  $\left\{ \widehat{\lambda}_{N,p,1},\widehat{\lambda}_{N,p,2},\ p=1,\dots,k_{N}\right\}$ of the pure point spectra of $L_{1}$ and $L_{2},$ computed   from the  empirical  two--dimensional wavelet reconstructions  of $L_{1}$ and $L_{2}$ at scale $D=10,$ based on $100$ realisations of the multiscale parameter estimators. The boxplots of their sample values can be found  in Figure \ref{f1}. Finally, the true operators $L_{1}$ and $L_{2},$ and  their functional estimates,  at scales $j=7,8,9, 10,$  are displayed in
Figure \ref{OPL1}. The contour plots in Figure \ref{figoest}  provide the multiscale (scales 7--10) description of the original and estimated spatial log--intensity field $\mathbf{X},$ at  $t=1/2.$ At the same time, Figure \ref{figoest2} displays the  smoothed   original and estimated log--intensity values at different scales or resolution levels ($j=7,8,9,10$). One  can  observe the effect of the Functional Data Analysis (FDA)
preprocessing procedure, and the effect of  increasing the number of spatial nodes, when comparing the spatial patterns observed at different temporal  scales  in Figures \ref{figoest2} and \ref{figoest}.

\begin{figure}[h!]
\begin{center}
\includegraphics[width=5cm,height=4.5cm]{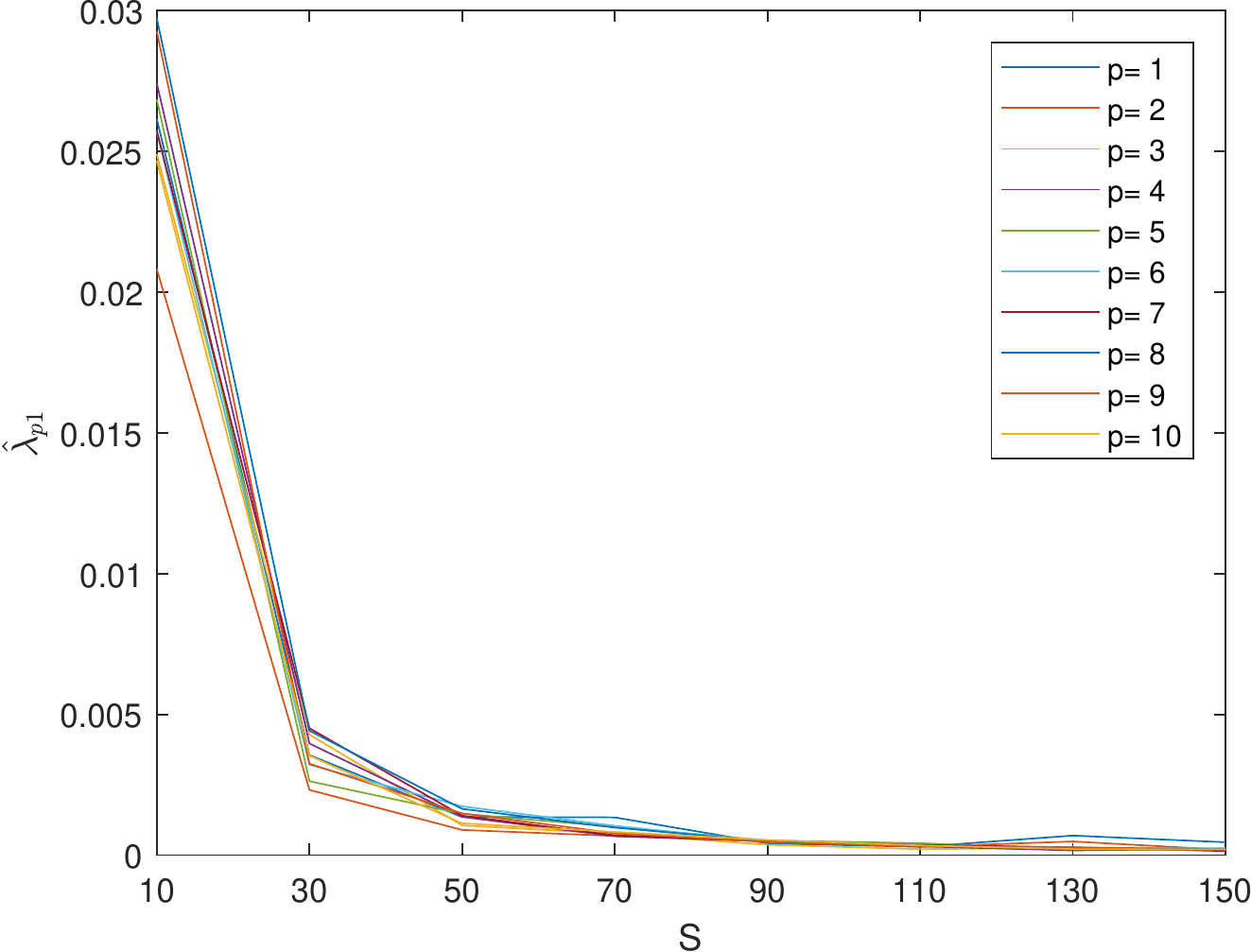}
\includegraphics[width=5cm,height=4.5cm]{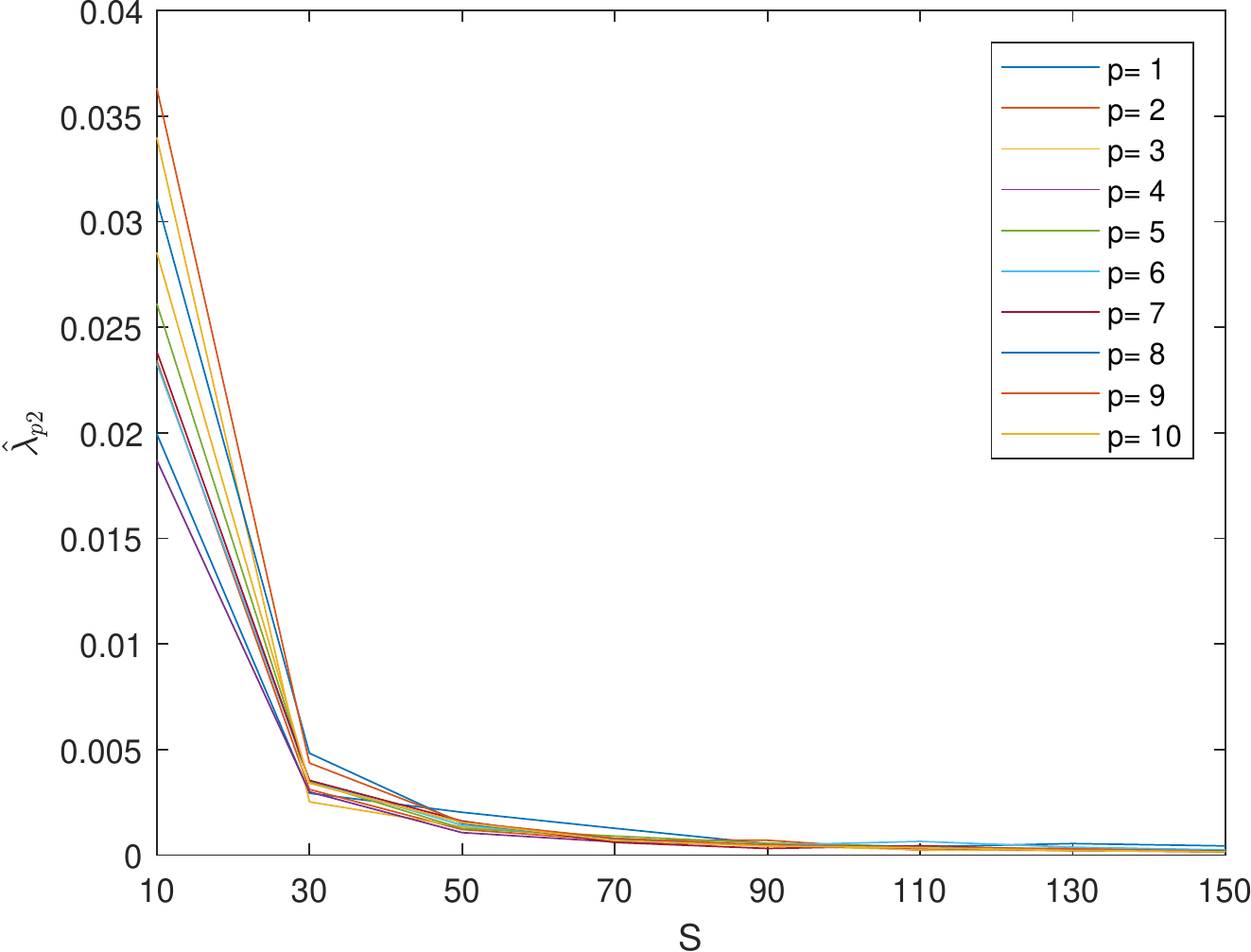}
\end{center}
\caption{Empirical mean quadratic errors (E.M.S.Es), associated with  $\left\{ \widehat{\lambda}_{N,p,1},\widehat{\lambda}_{N,p,2},\ p=1,\dots,k_{N}\right\},$  $N=100, 900, 2500, 4900, 8100, 12100, 16900, 22500$}\label{fig.mse}
\end{figure}

\begin{figure}[h!]
\includegraphics[width=3.9cm,height=3.9cm]{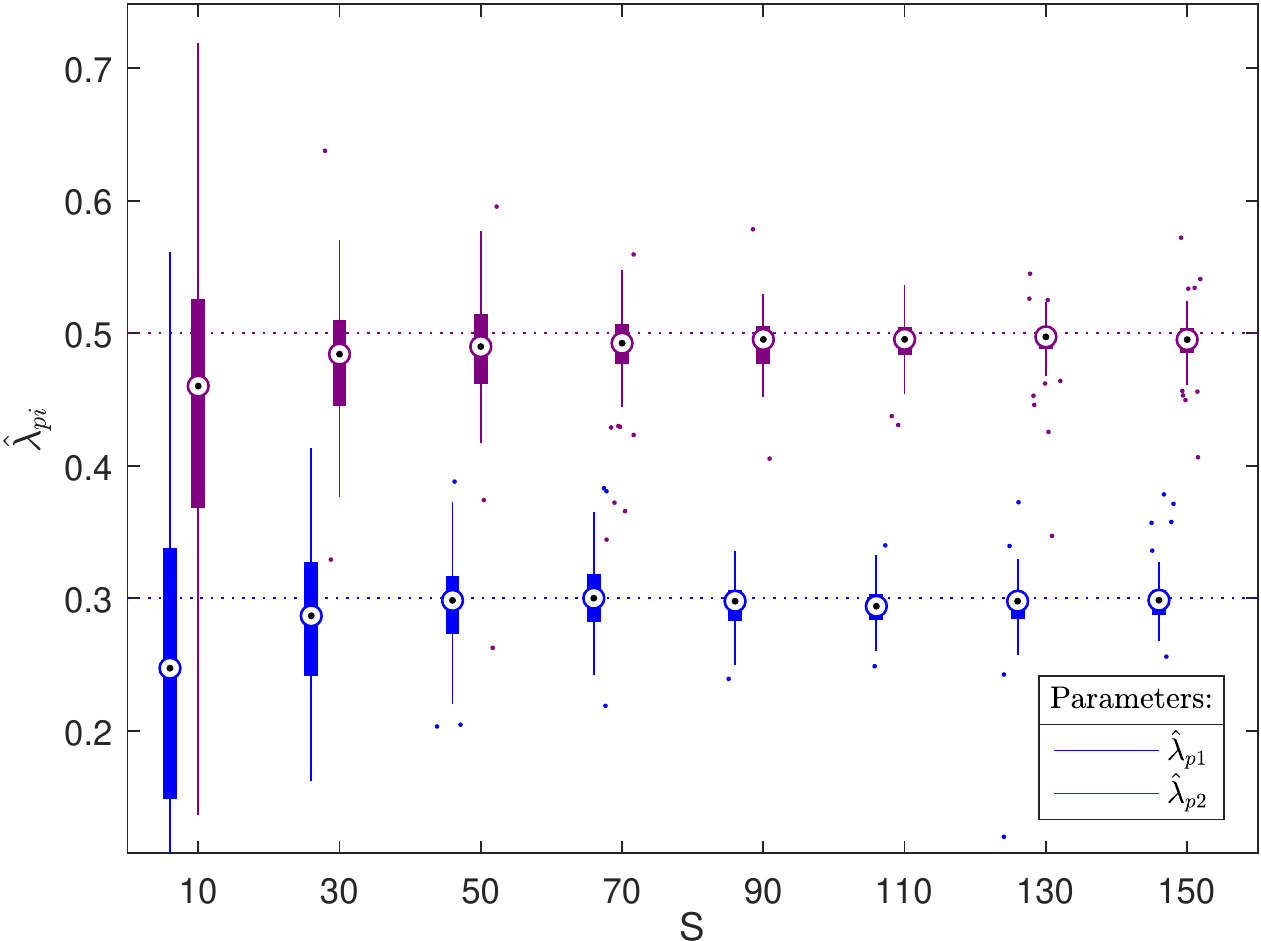}
\includegraphics[width=3.9cm,height=3.9cm]{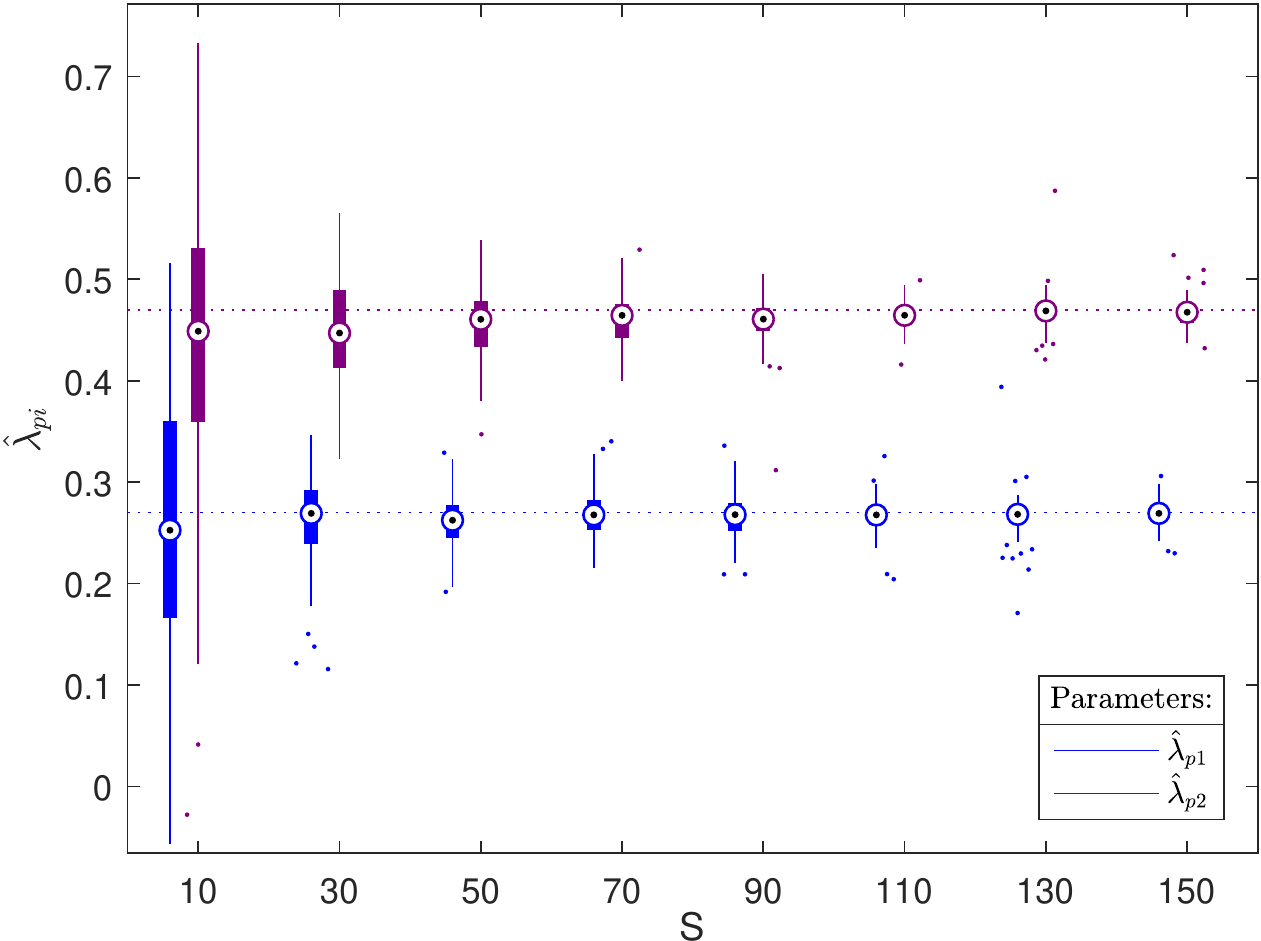}
\includegraphics[width=3.9cm,height=3.9cm]{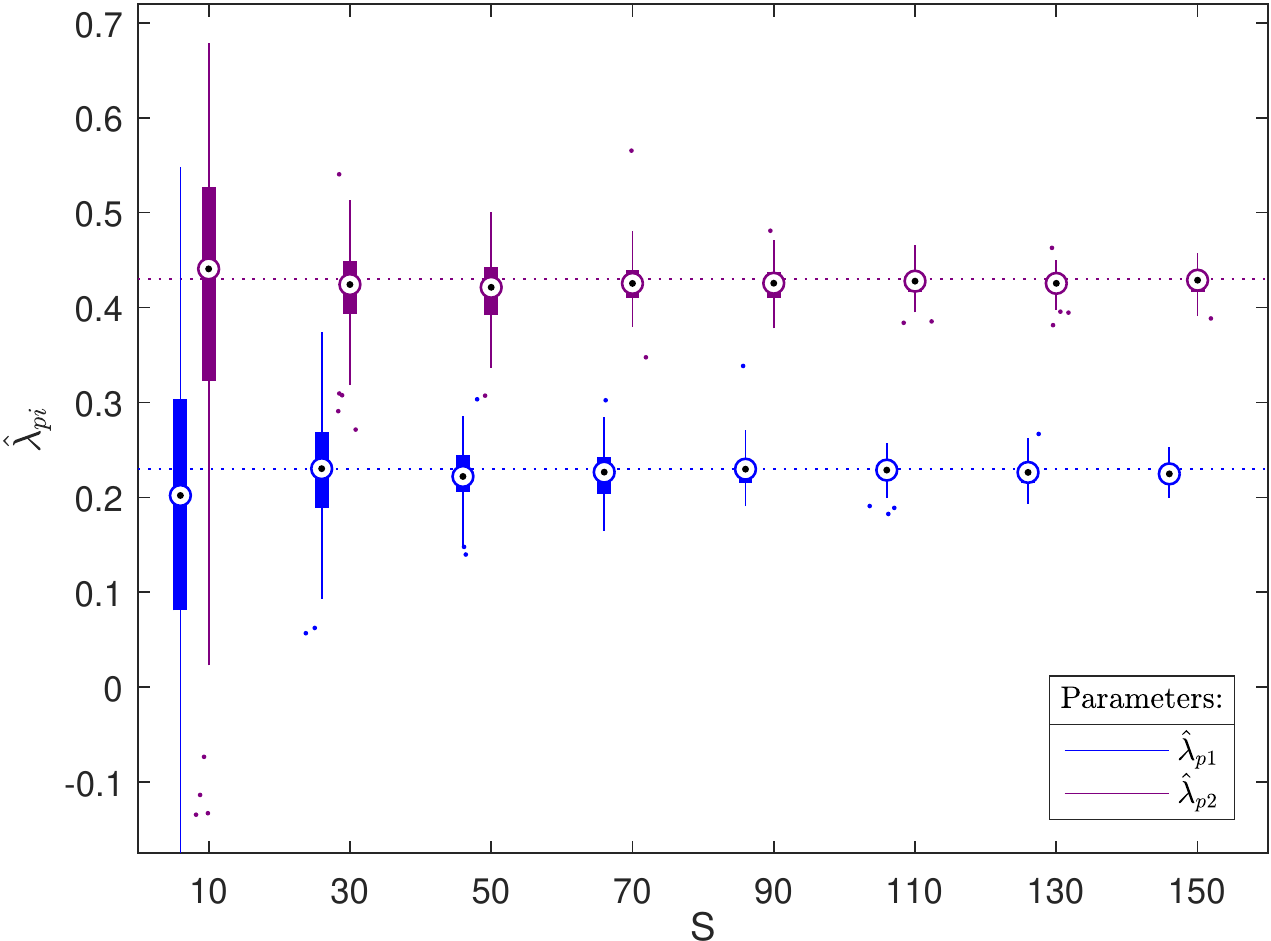}
\includegraphics[width=3.9cm,height=3.9cm]{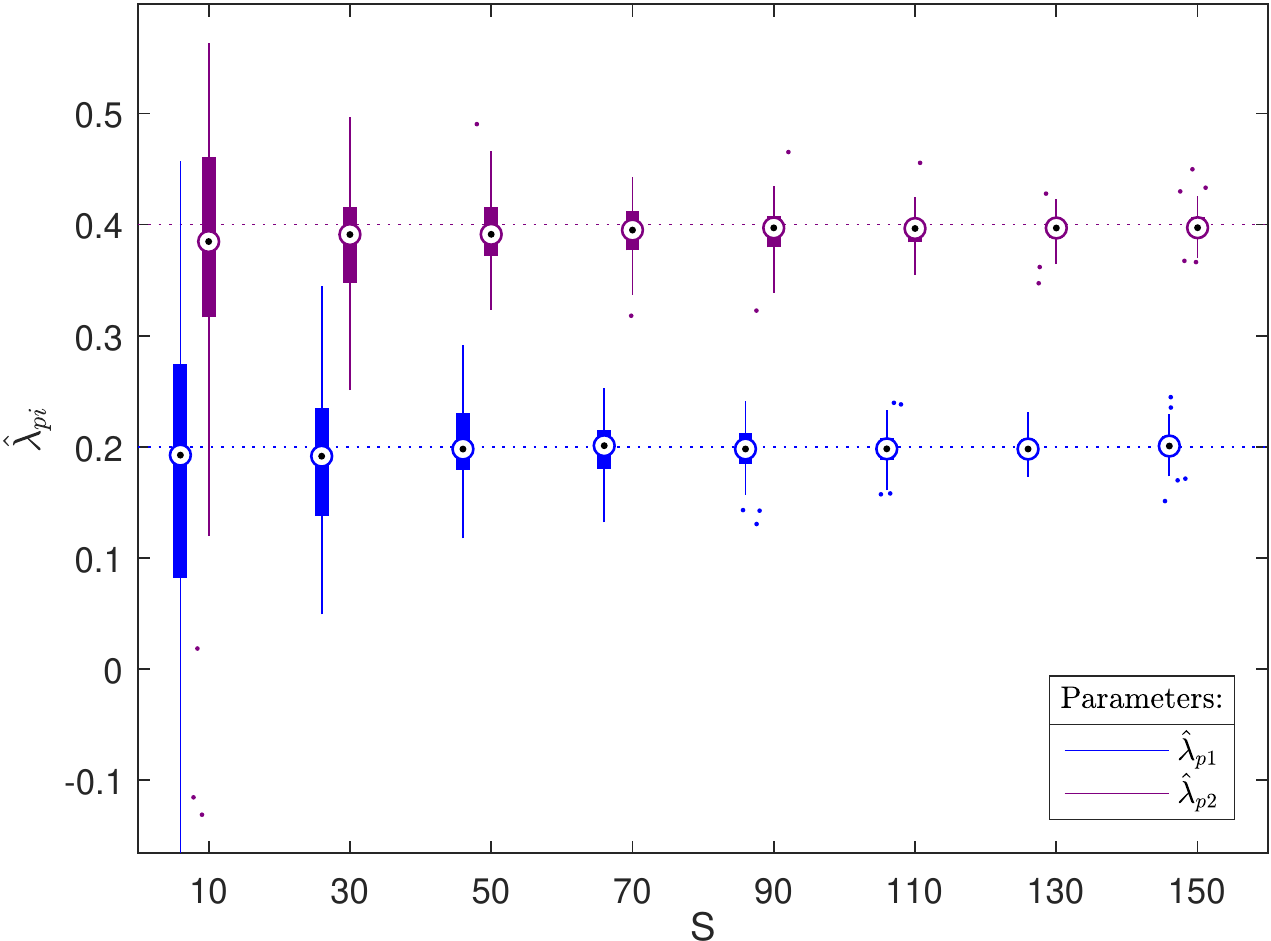}
\includegraphics[width=3.9cm,height=3.9cm]{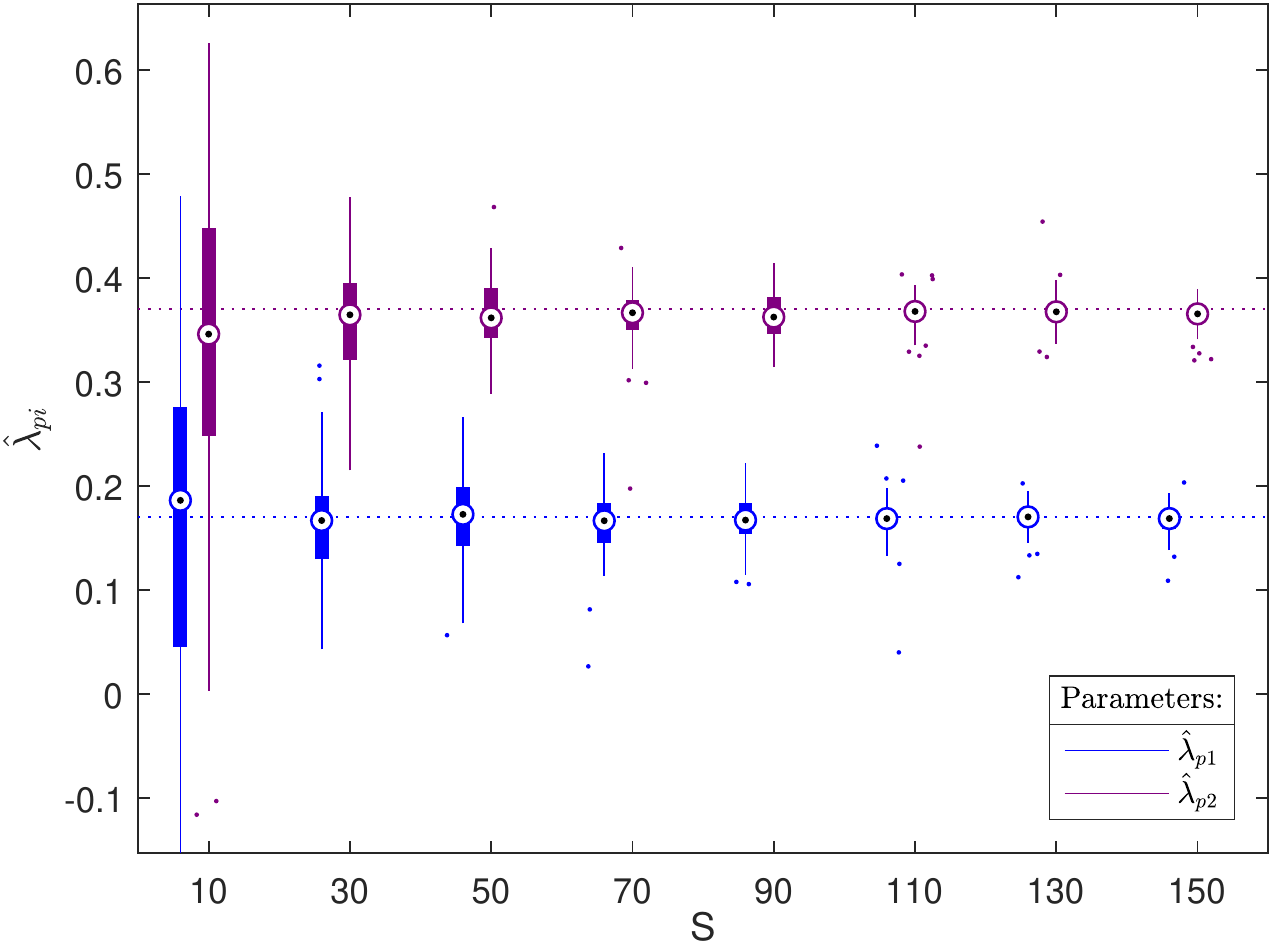}
\includegraphics[width=3.9cm,height=3.9cm]{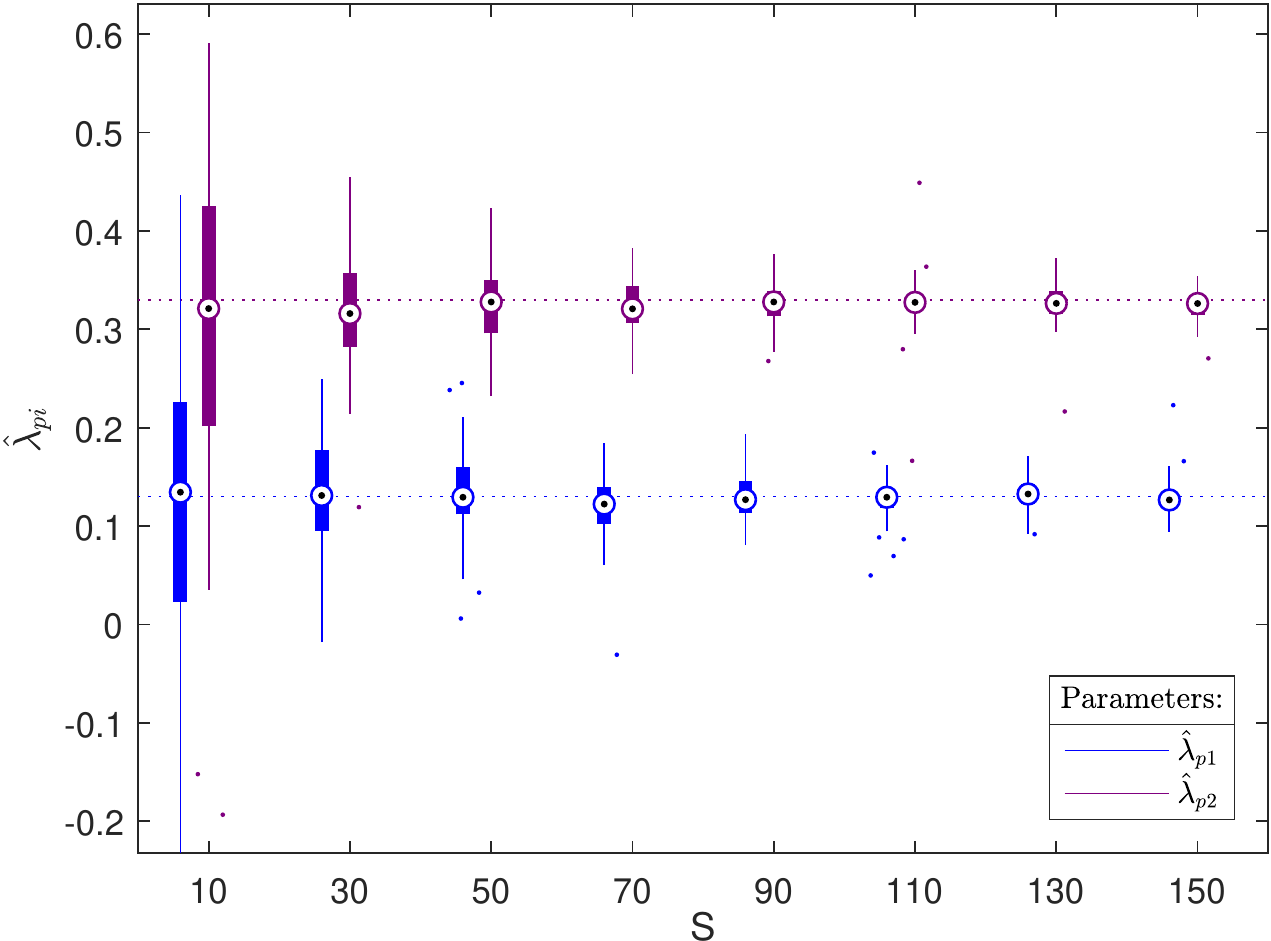}
\includegraphics[width=3.9cm,height=3.9cm]{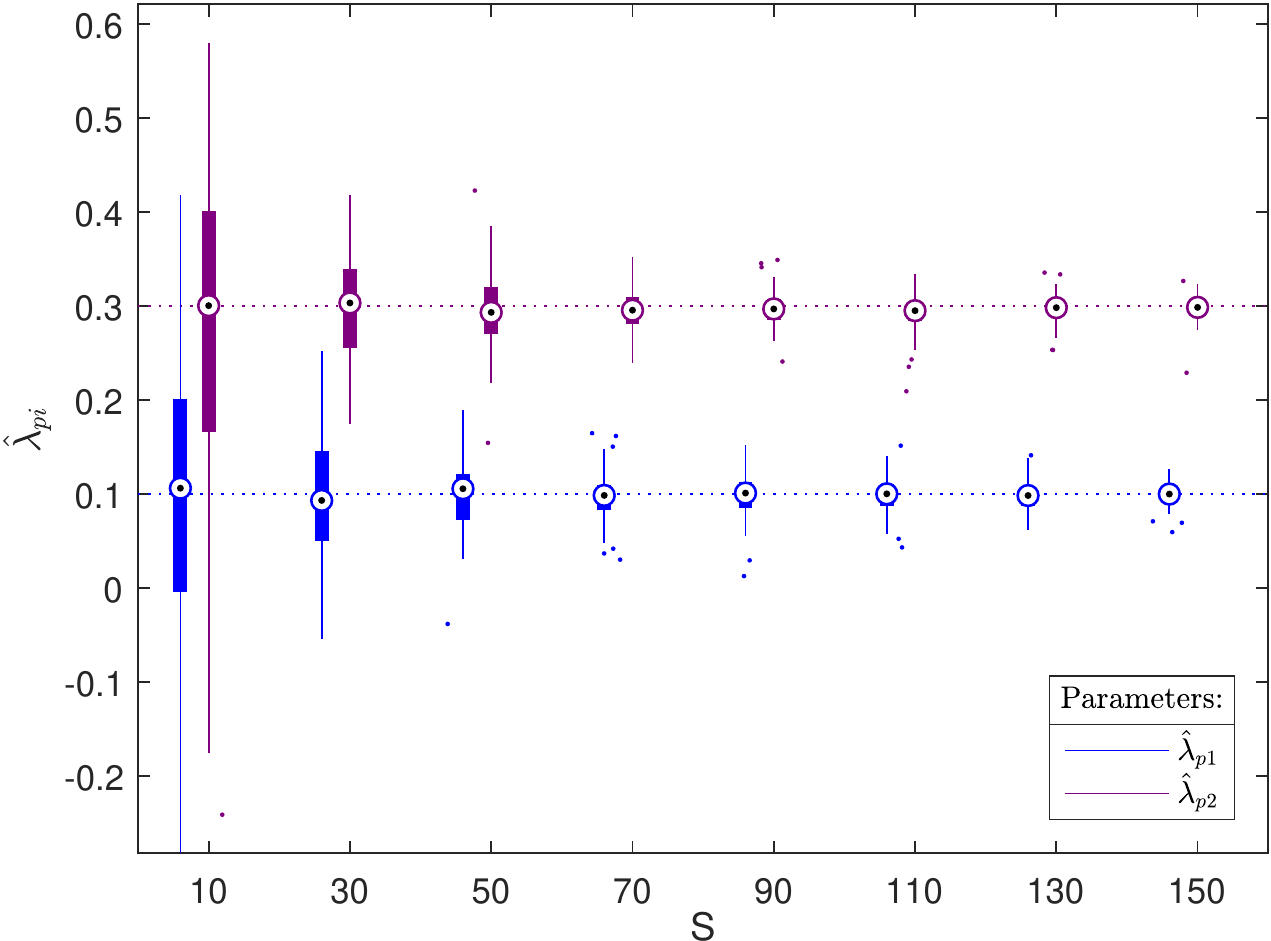}
\includegraphics[width=3.9cm,height=3.9cm]{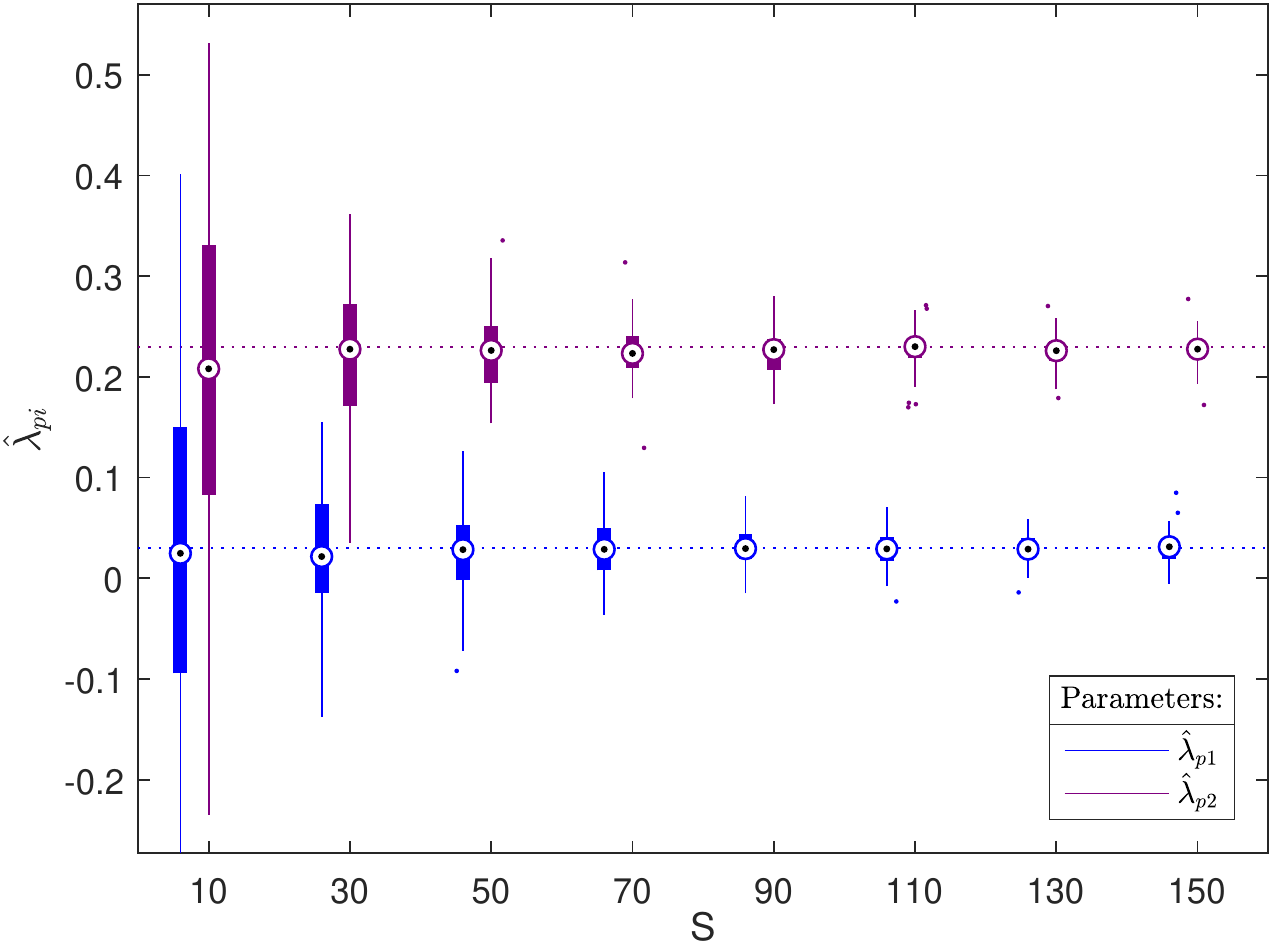}
\includegraphics[width=3.9cm,height=3.9cm]{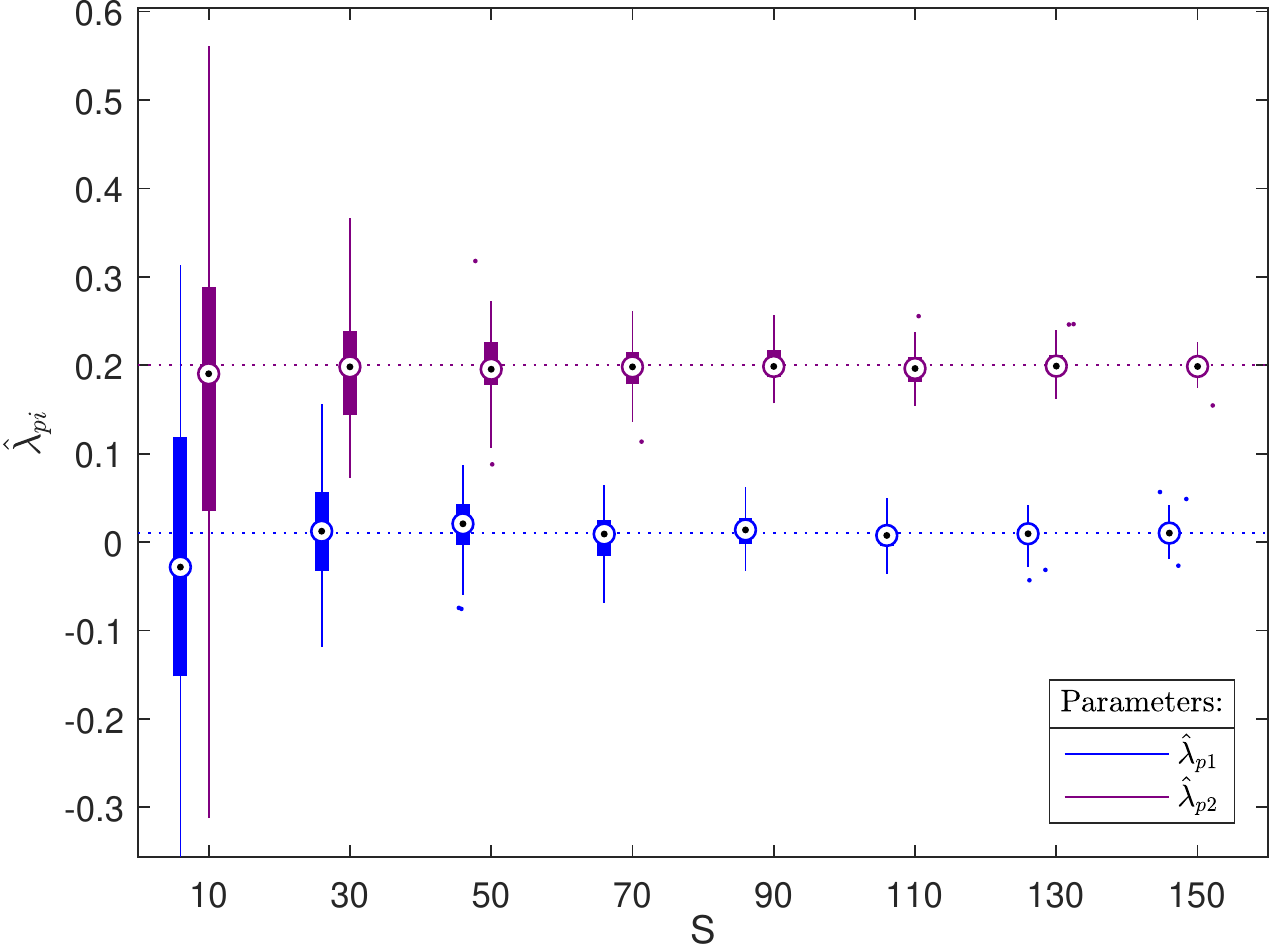}
\includegraphics[width=3.9cm,height=3.9cm]{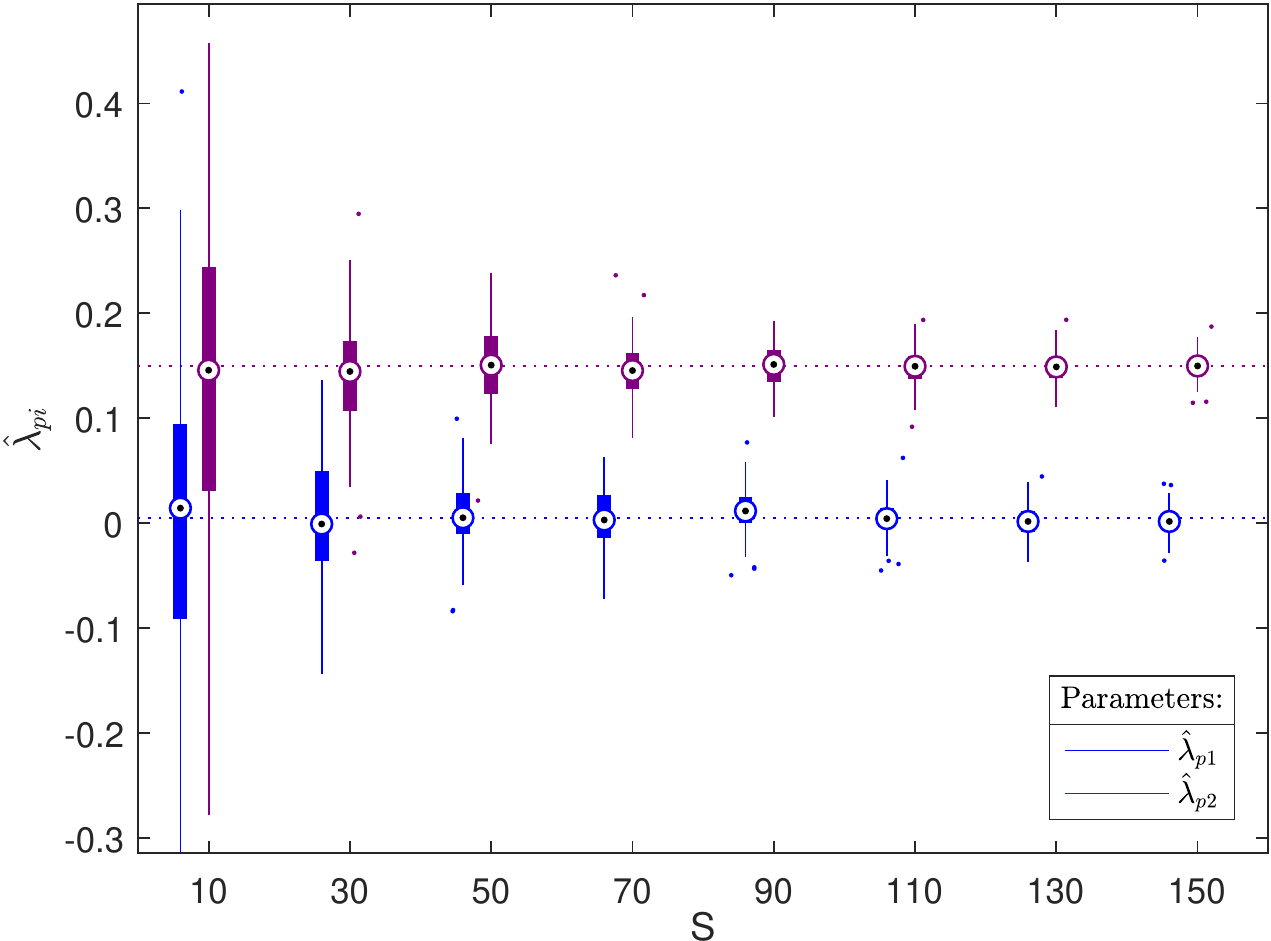}
        \caption{Boxplots of the sample values of $\widehat{\lambda}_{N,p,i}$,  $p=1,\ldots,10,$ $i=1,2$ (from  left to  right, and from  top to   bottom), based on $100$ generations of the functional samples of size  $N=100, 900, 2500, 4900,$ $8100,$ $12100, 16900, 22500$. The  true parameter value is reflected in dotted line}\label{f1}
\end{figure}

 \begin{figure}[h!]
\includegraphics[width=12cm,height=10cm]{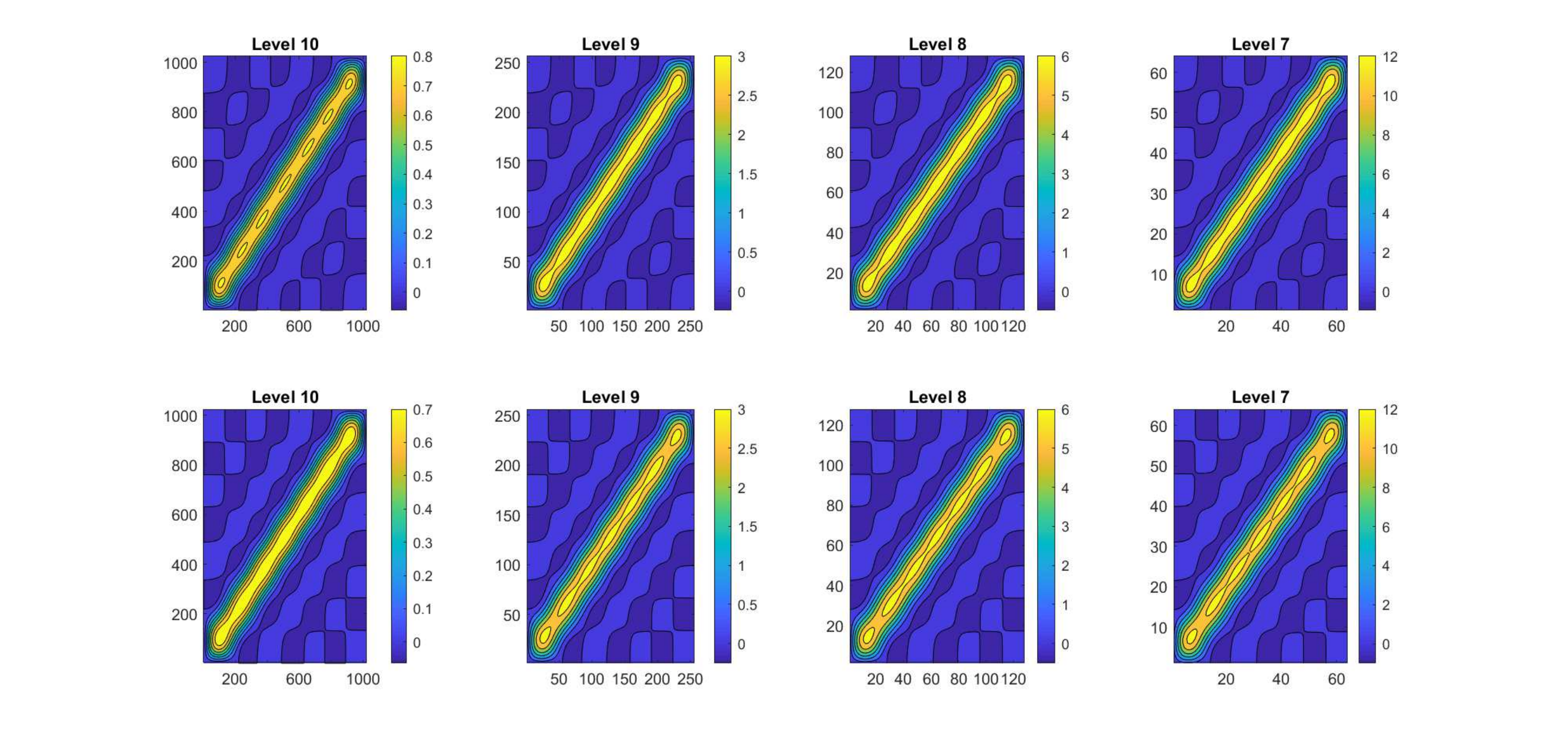}
\includegraphics[width=12cm,height=10cm]{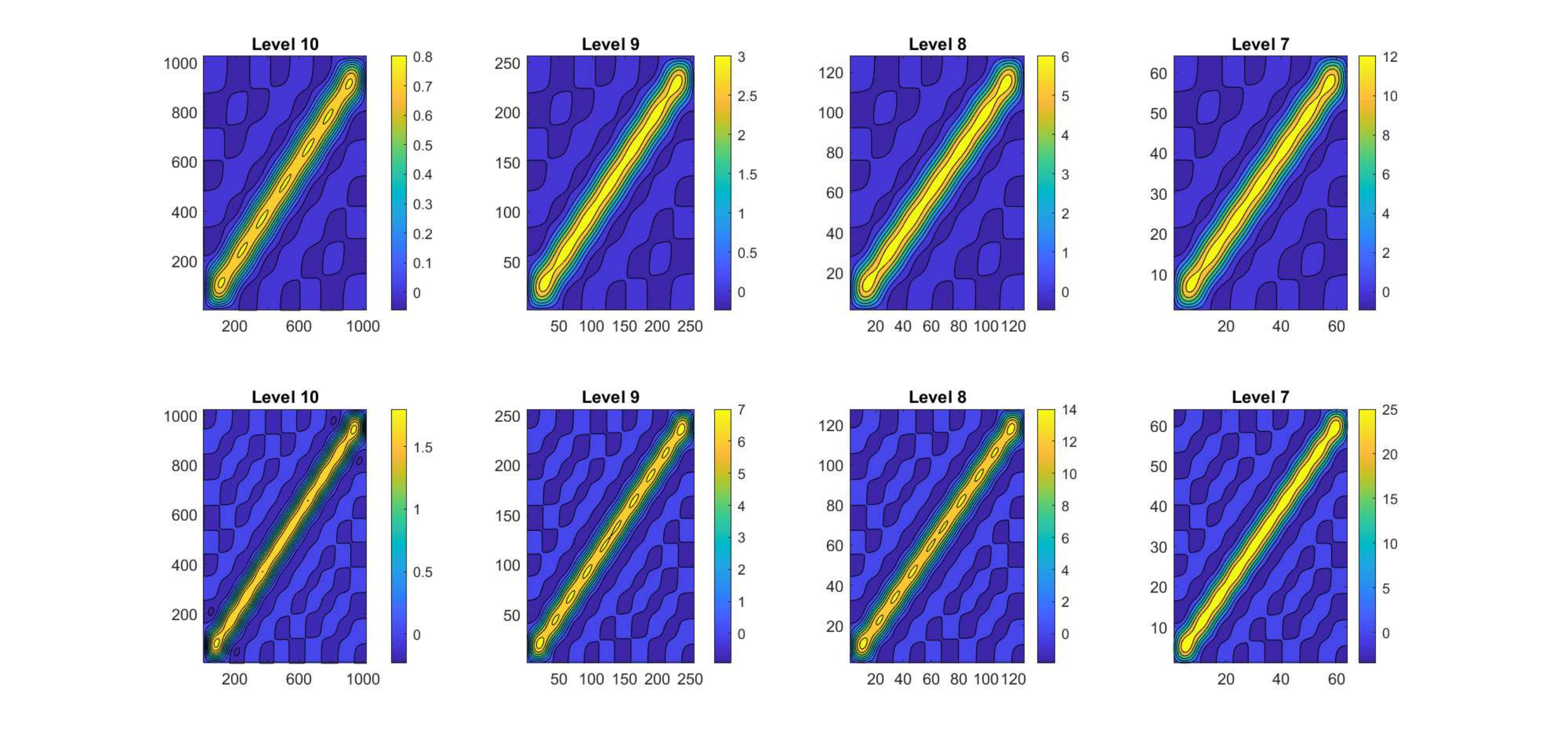}
\caption{True $L_{1}$ at the top row, and its multiscale  estimate at the second row, for scales $j=7, 8, 9, 10$ (from  right to  left).
True  $L_{2}$ at the third row, and its multiscale  estimation at the bottom row, for scales $j=7, 8, 9, 10$  (from right to  left), over a $30\times 30$ spatial regular grid}\label{OPL1}
 \end{figure}

\begin{center}
 \begin{figure}[!h]
 \centering
\includegraphics[width=2.5cm,height=2.5cm]{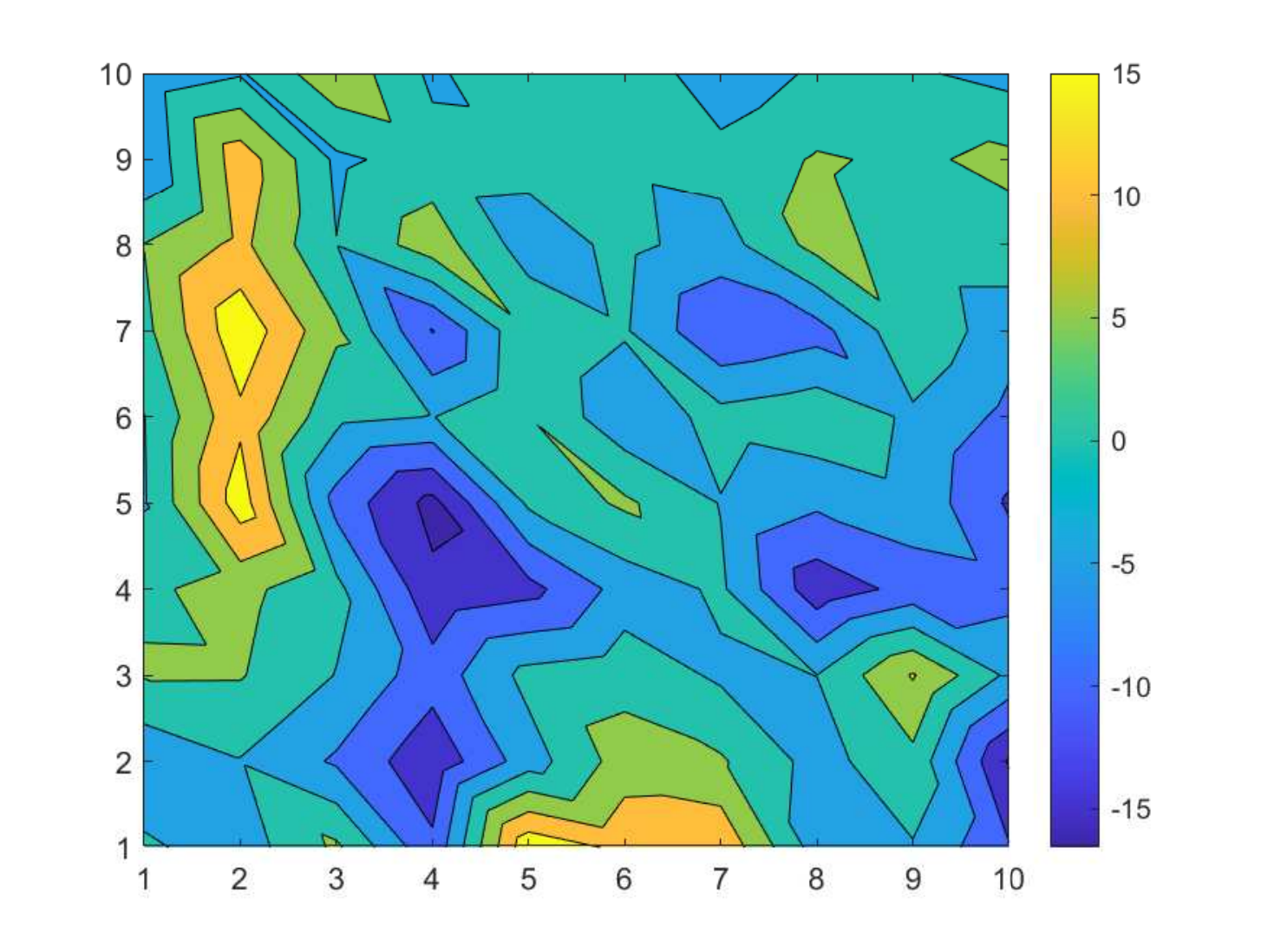}
\includegraphics[width=2.5cm,height=2.5cm]{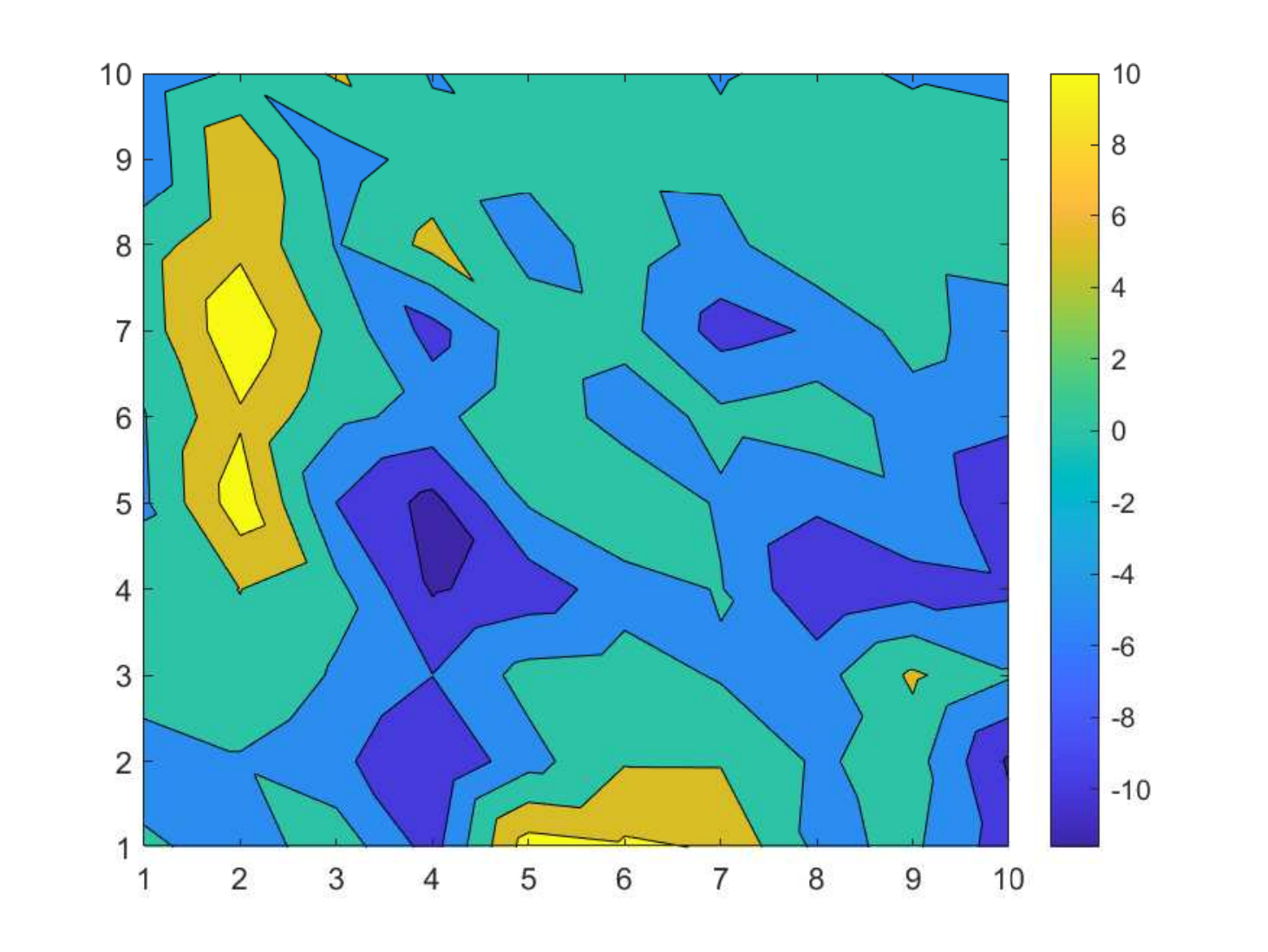}
\includegraphics[width=2.5cm,height=2.5cm]{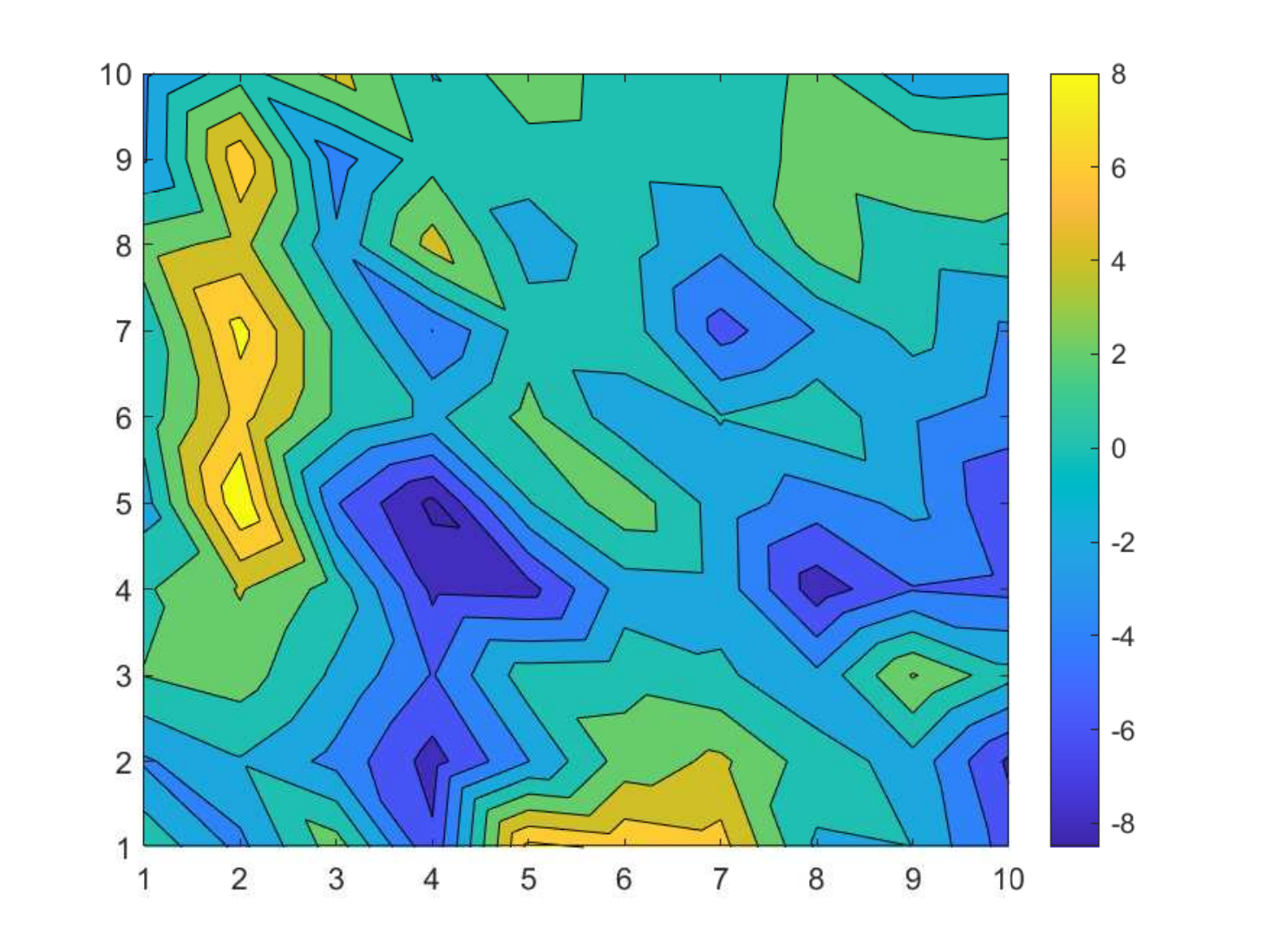}
\includegraphics[width=2.5cm,height=2.5cm]{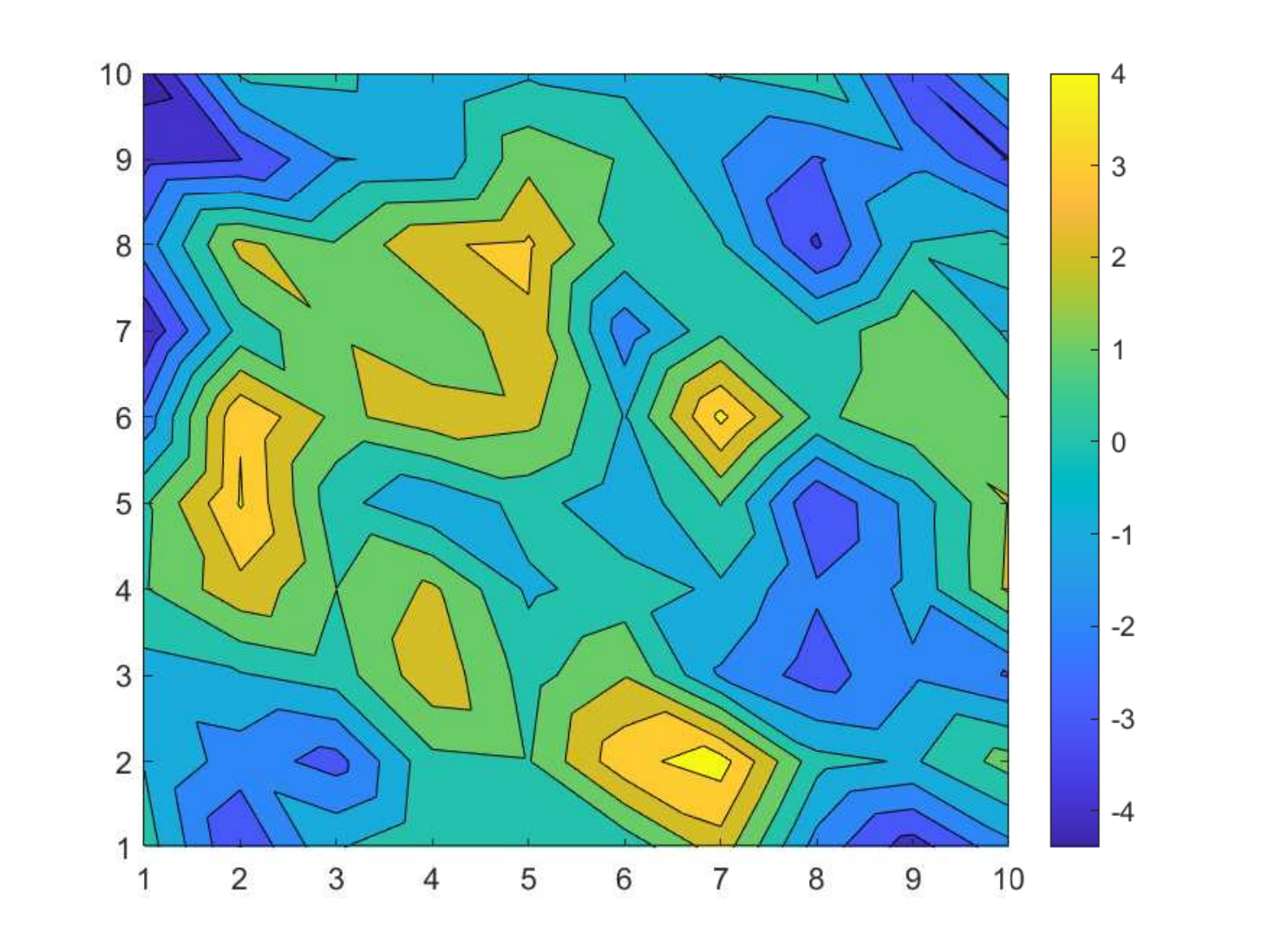}
\hspace*{0.56cm}
\includegraphics[width=2.5cm,height=2.5cm]{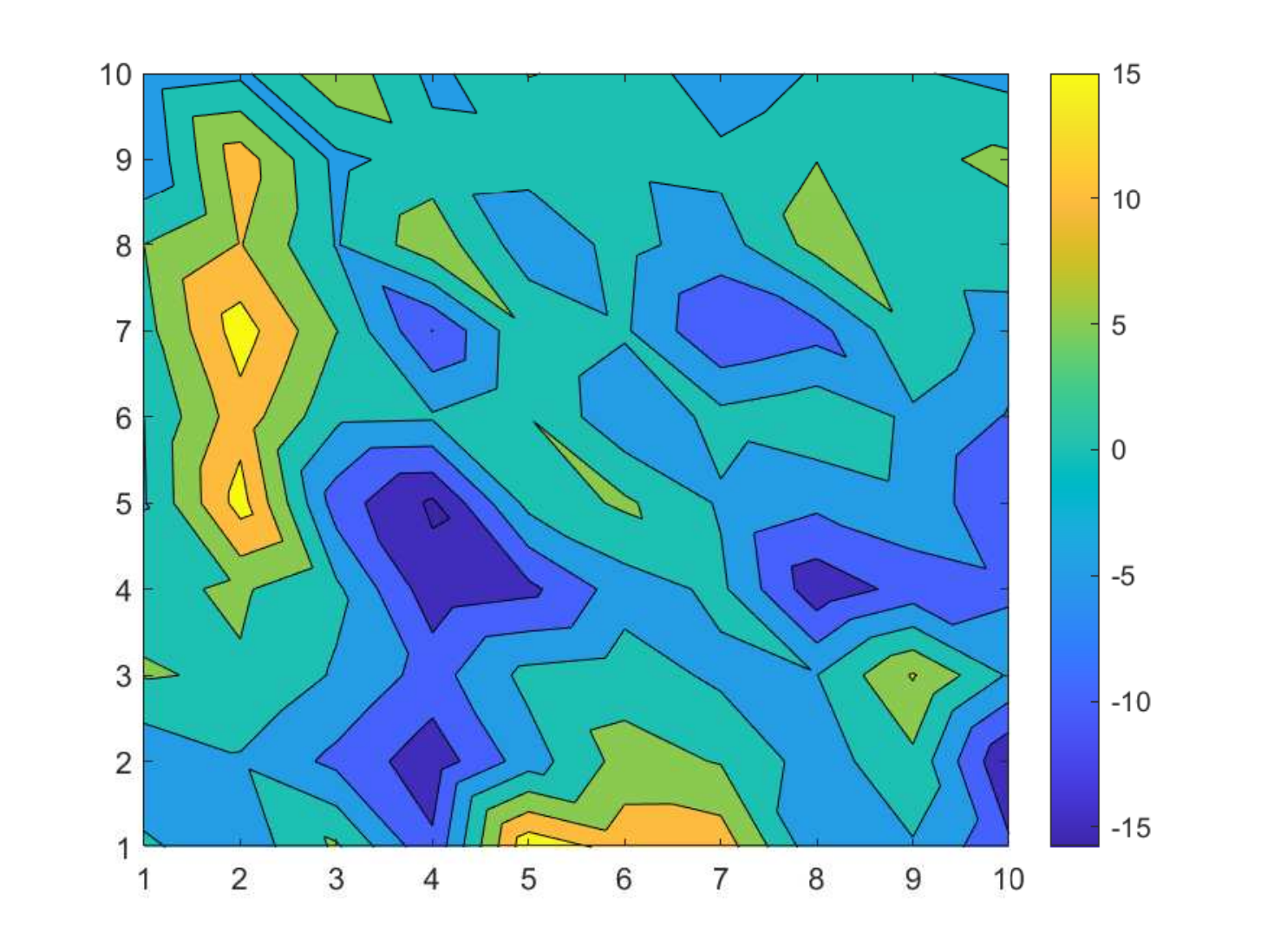}
\includegraphics[width=2.5cm,height=2.5cm]{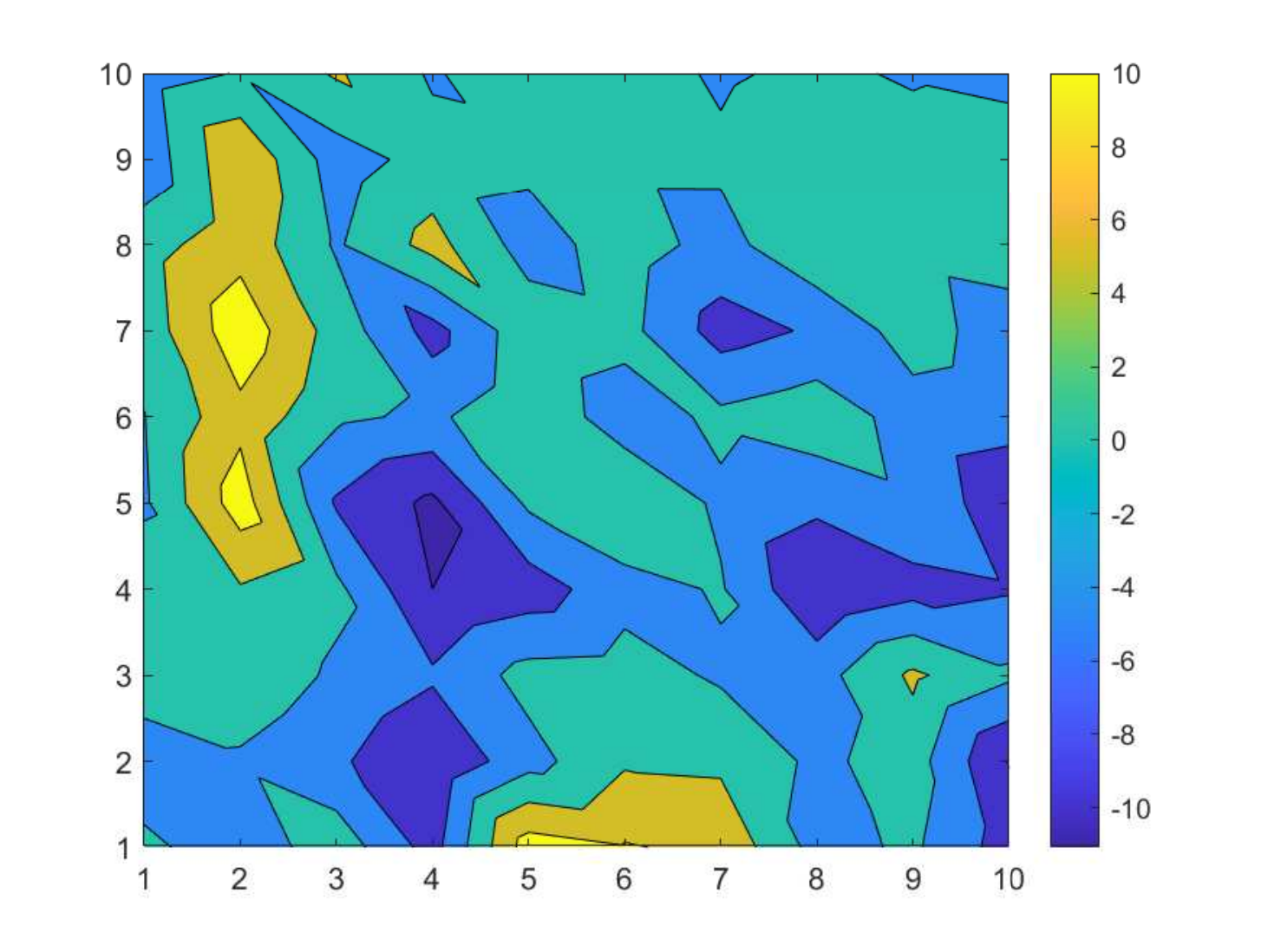}
\includegraphics[width=2.5cm,height=2.5cm]{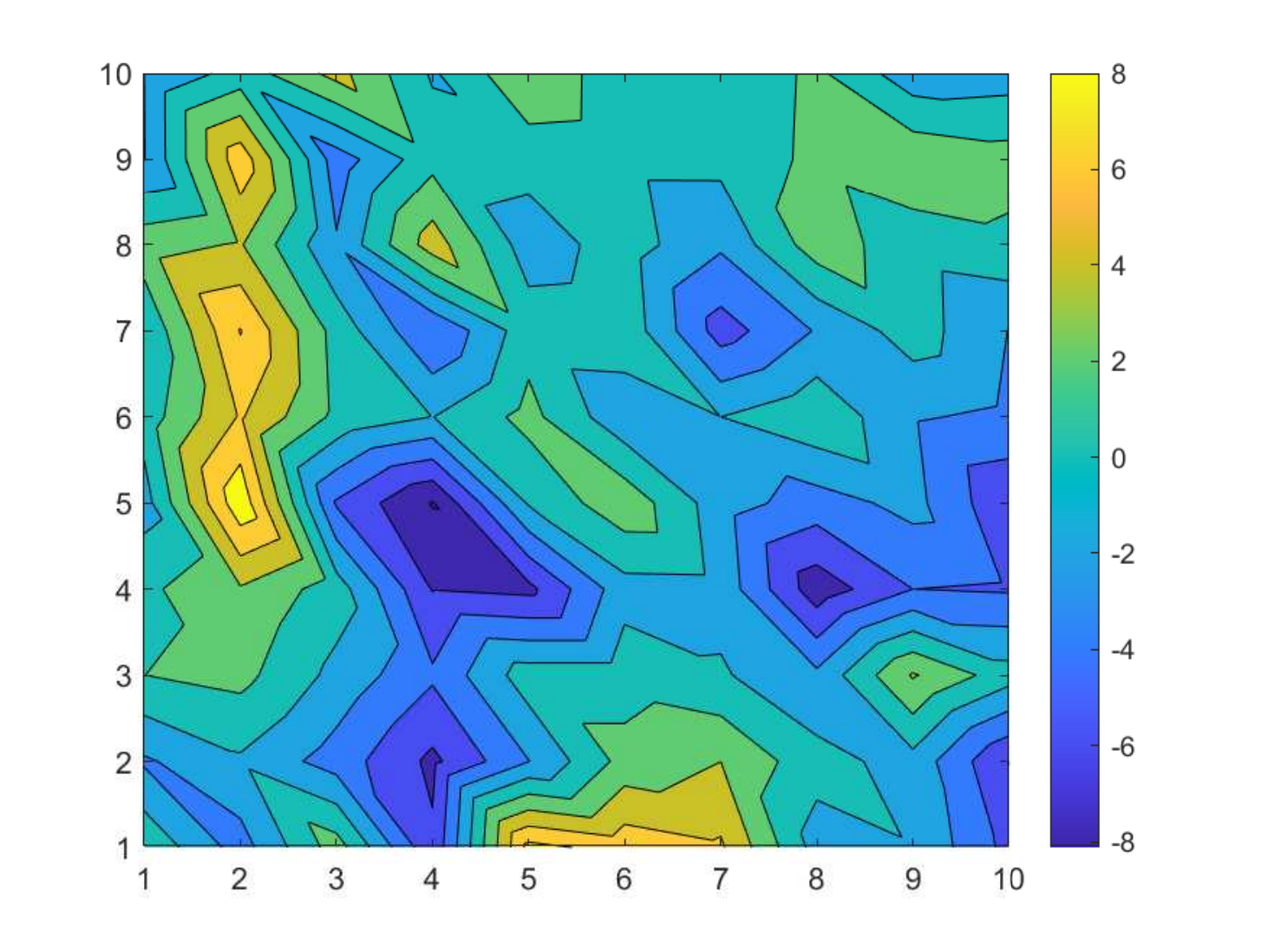}
\includegraphics[width=2.5cm,height=2.5cm]{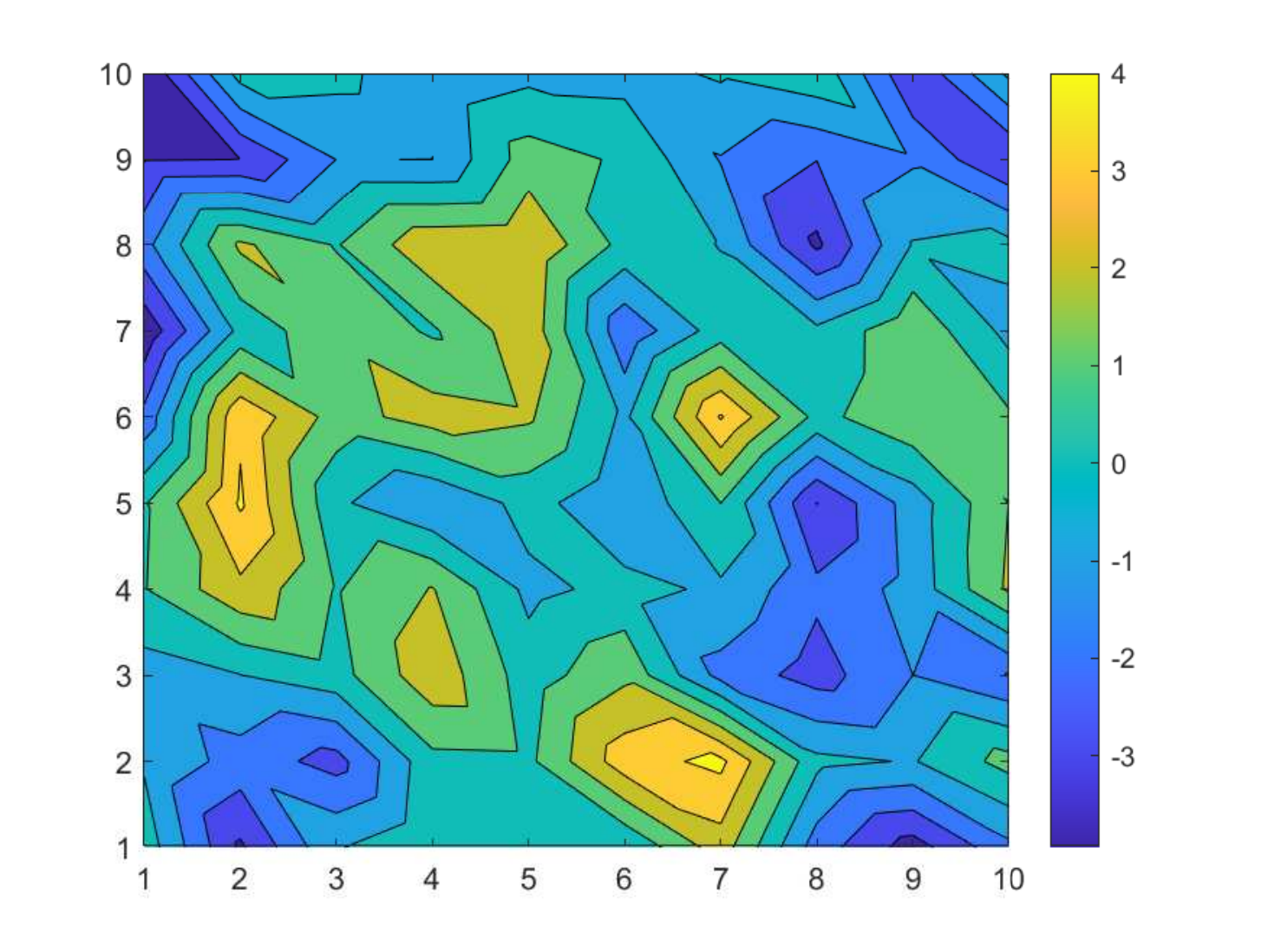}
 \caption{Original (top--row) and estimated (bottom--row) spatial log--intensity field $\mathbf{X},$ at time $t=1/2,$ through the scales $j=7,8,9,10$ (from  left to  right), over a $10\times 10$ spatial regular grid, from smoothed curve data}
 \label{figoest2}
 \end{figure}
\end{center}

\begin{center}
 \begin{figure}[!h]
 \centering
        \includegraphics[width=2.5cm,height=2.5cm]{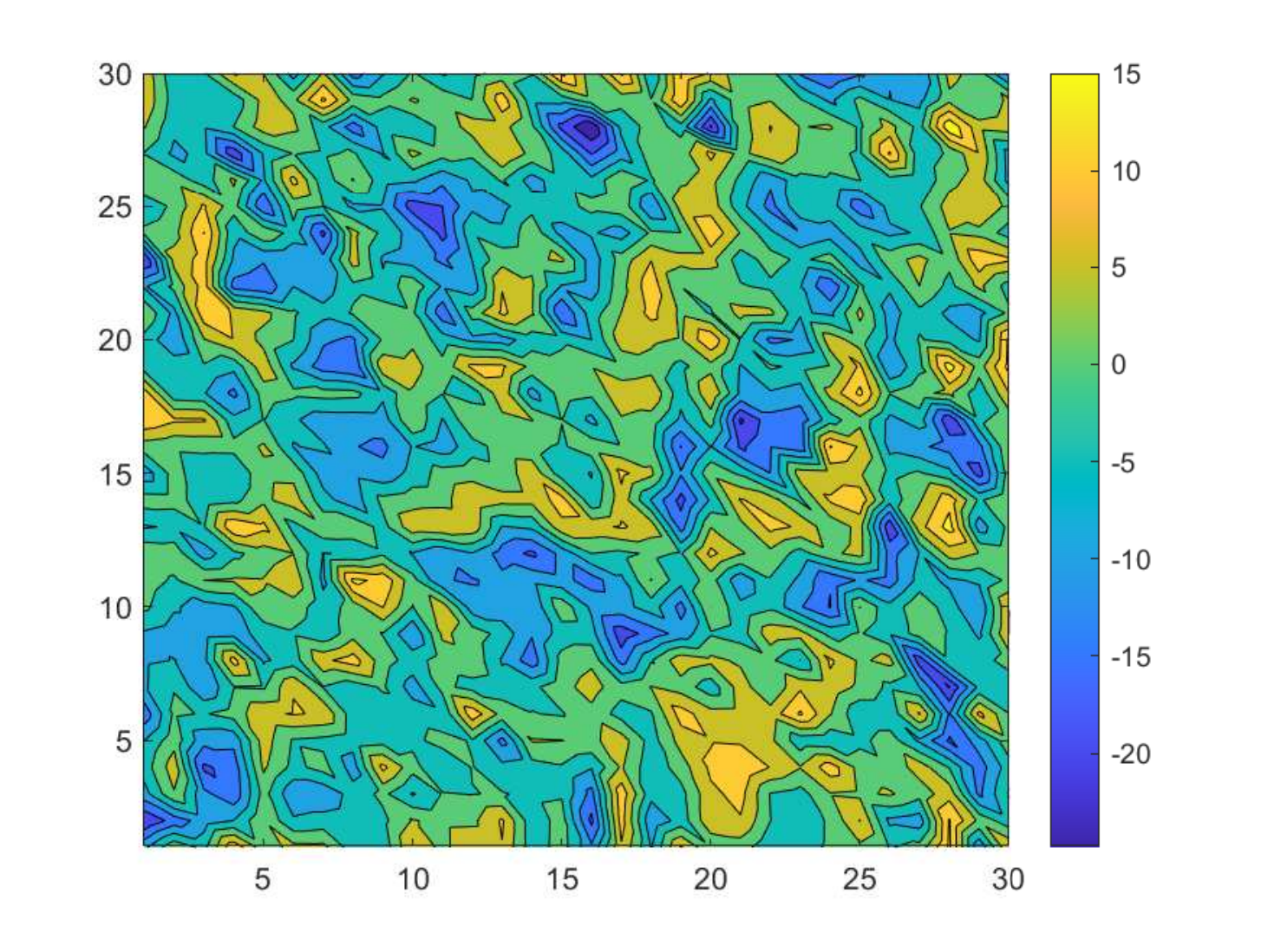}
        \includegraphics[width=2.5cm,height=2.5cm]{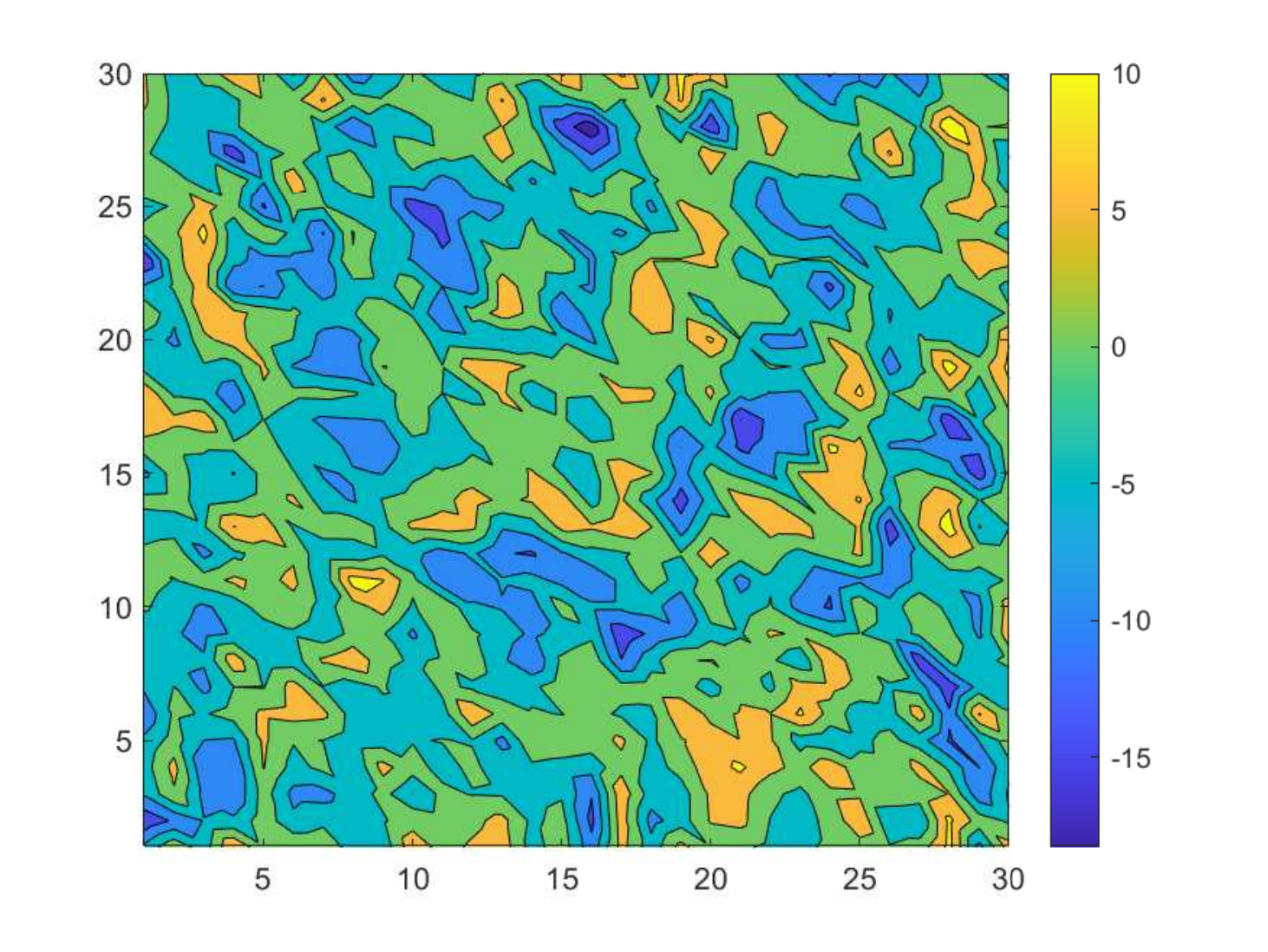}
        \includegraphics[width=2.5cm,height=2.5cm]{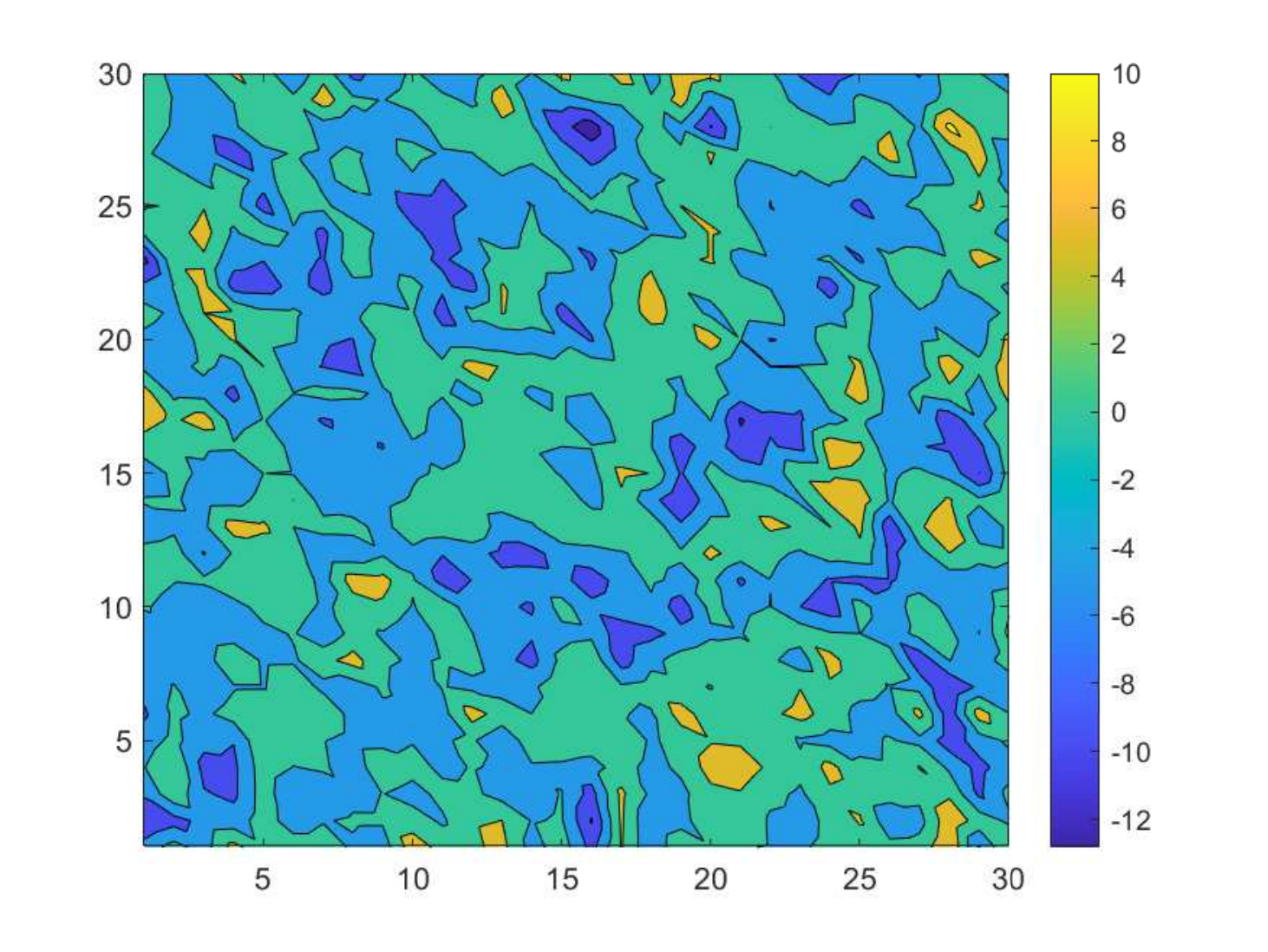}
        \includegraphics[width=2.5cm,height=2.5cm]{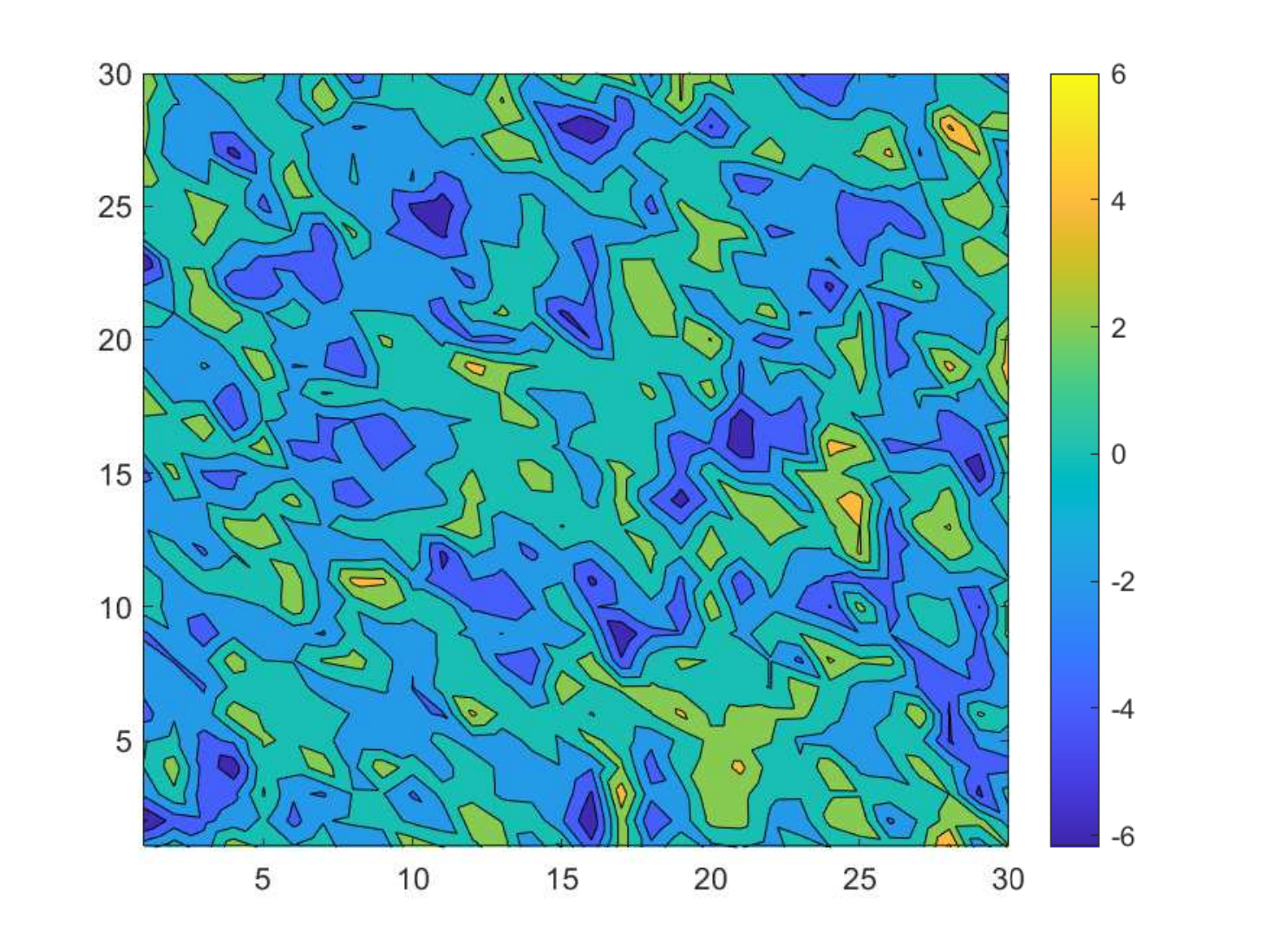}
        \hspace*{0.56cm}
        \includegraphics[width=2.5cm,height=2.5cm]{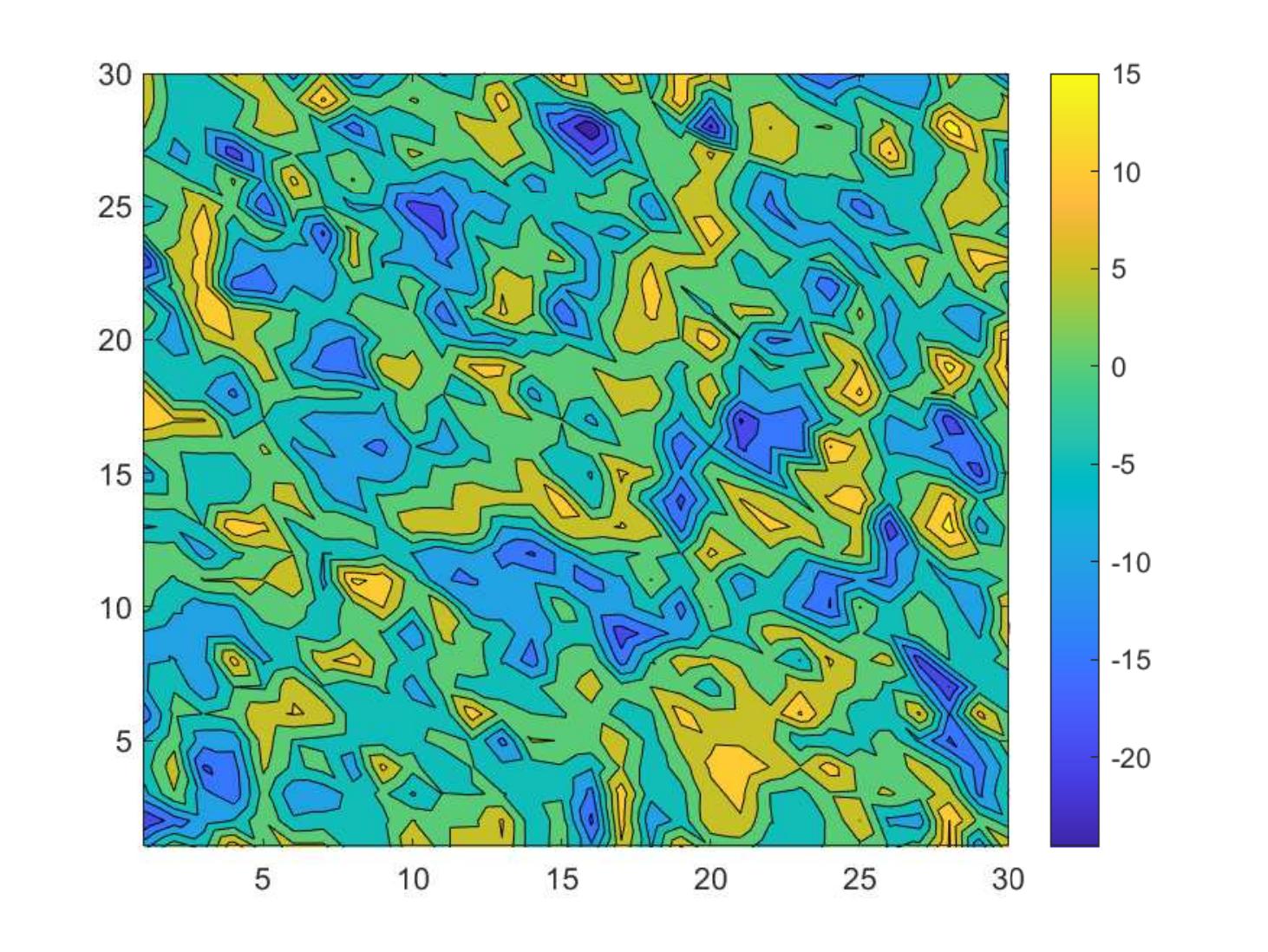}
        \includegraphics[width=2.5cm,height=2.5cm]{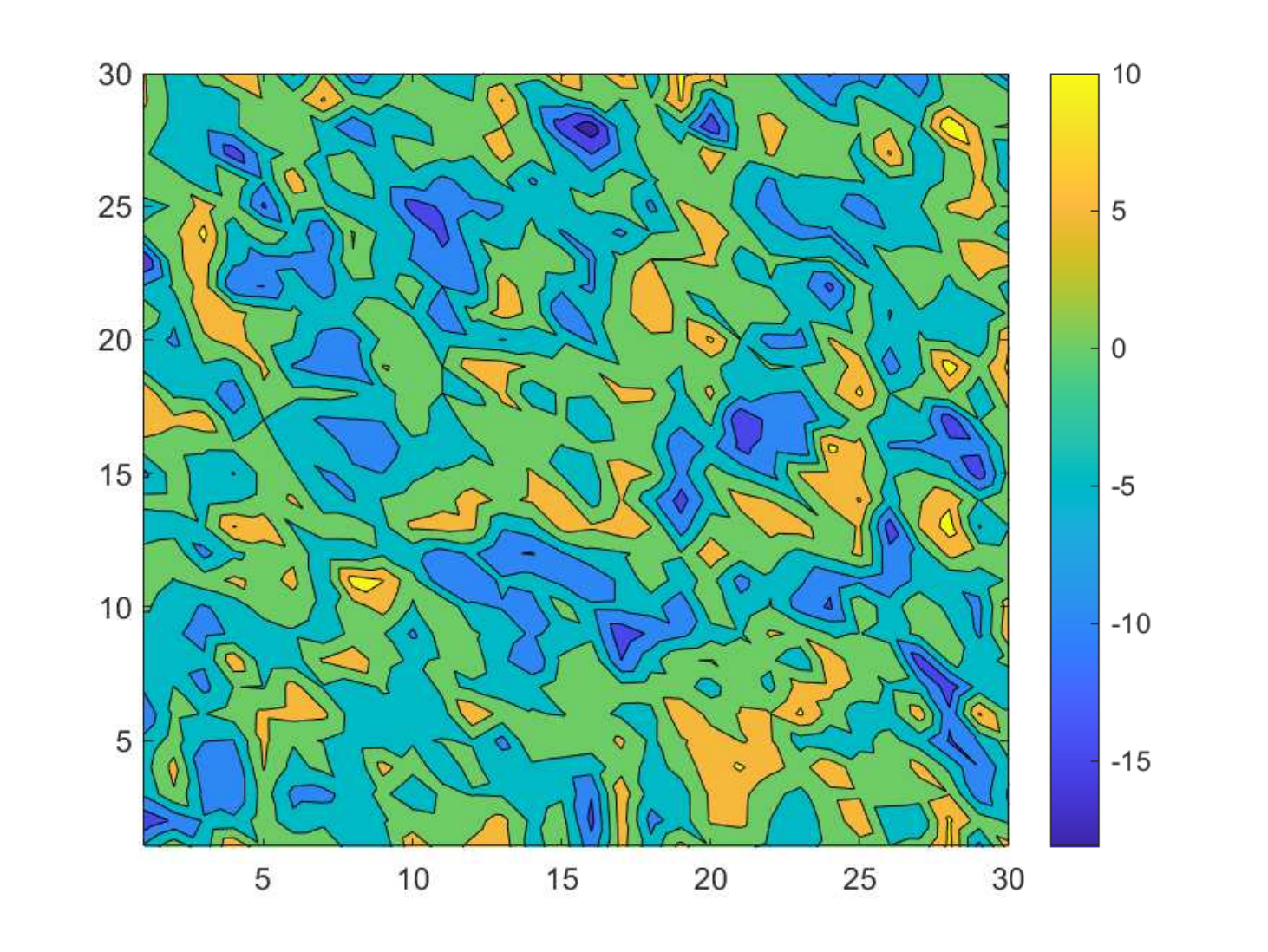}
        \includegraphics[width=2.5cm,height=2.5cm]{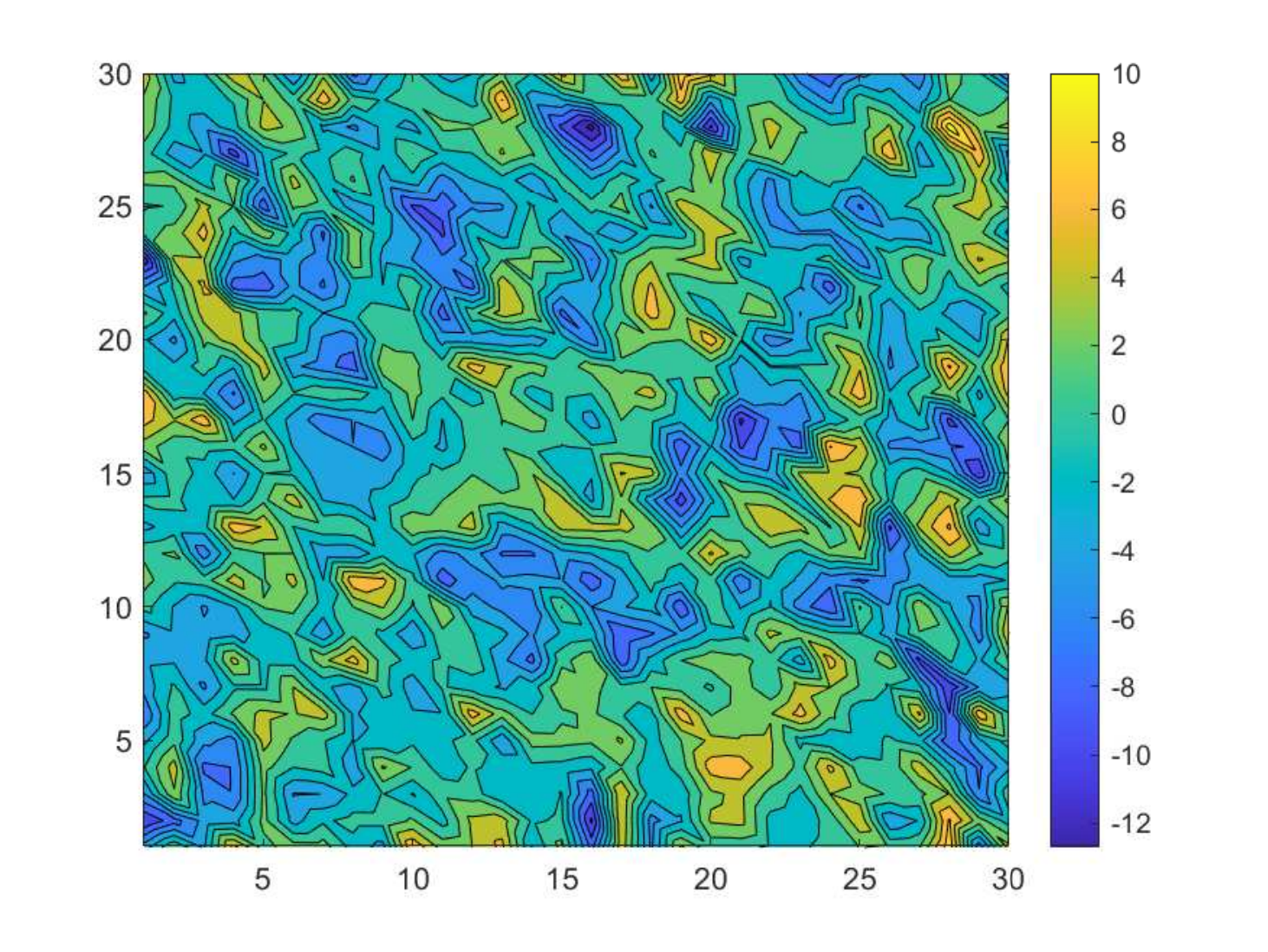}
        \includegraphics[width=2.5cm,height=2.5cm]{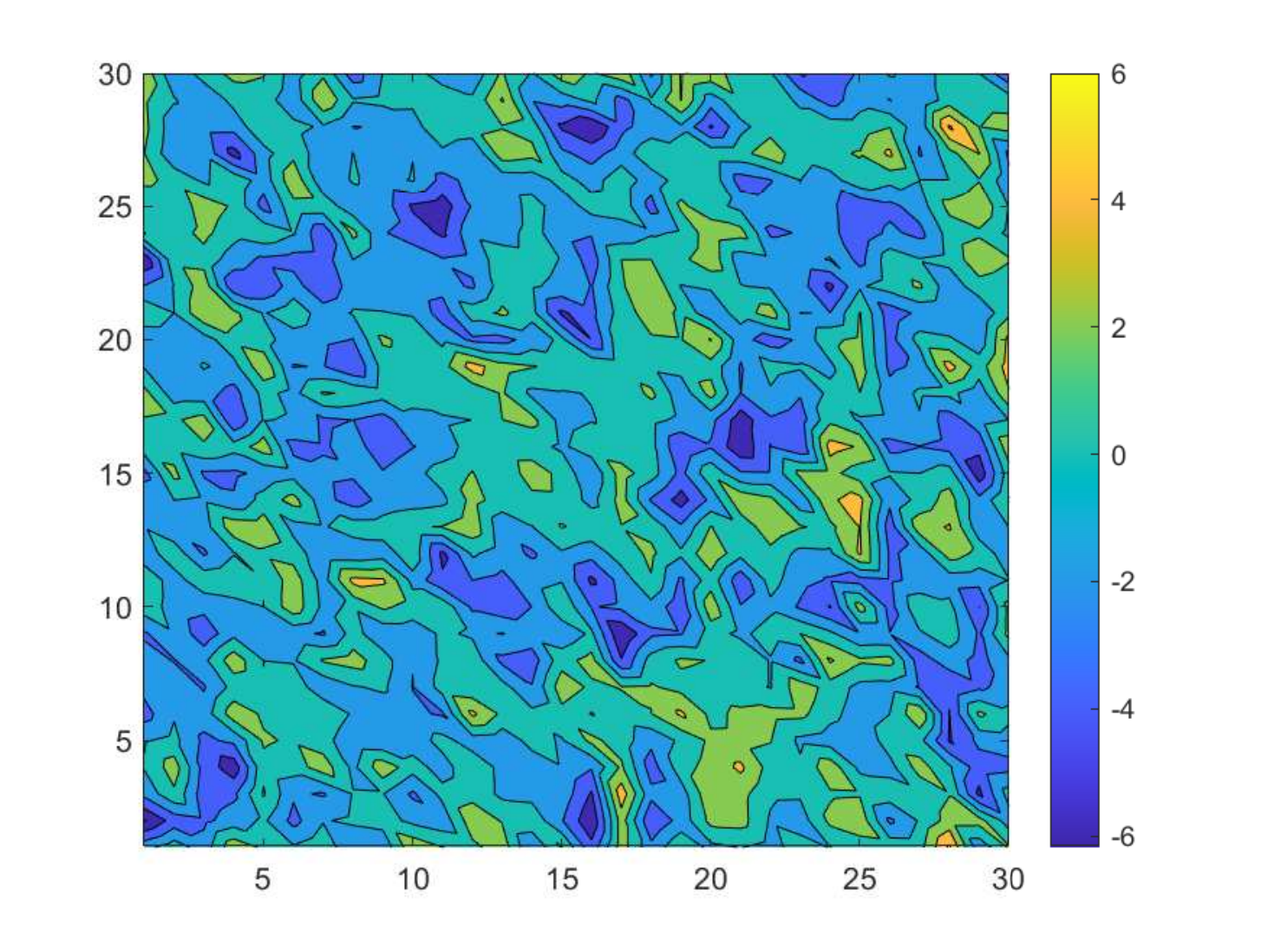}
 \caption{Original, non--smoothed (top--row), and estimated (bottom--row) spatial log--intensity field $\mathbf{X},$ at time $t=1/2,$ through the scales $j=7,8,9,10$ (from  left to right), over a $30\times 30$ spatial regular grid}
 \label{figoest}
 \end{figure}
\end{center}

\clearpage
 \section{Real-data example}
\label{s6}

The Spanish National Statistical Institute  provided the  data  on the observed cases of respiratory disease deaths, consisting of 432 monthly records, in the period 1980--2015, collected at the  48  Spanish  provinces in the Iberian Peninsula. The data are temporal, and spatial interpolated over a $20\times 20$ regular grid. Specifically, $1725$  temporal nodes, and 400 spatial nodes are considered.
A flexible  fitting of the underlying local behaviour (or singularity) of the observed and interpolated data is obtained, from a suitable choice of the scale or   resolution level (see Figures \ref{DR10NS}--\ref{D9}).   Note that FDA preprocessing  usually leads to an over--smoothing. That is the case of B--spline smoothing often applied  to construct  curve data sets (see Figure \ref{DR10S}).

\begin{figure}[h!]
\includegraphics[width=12cm,height=10cm]{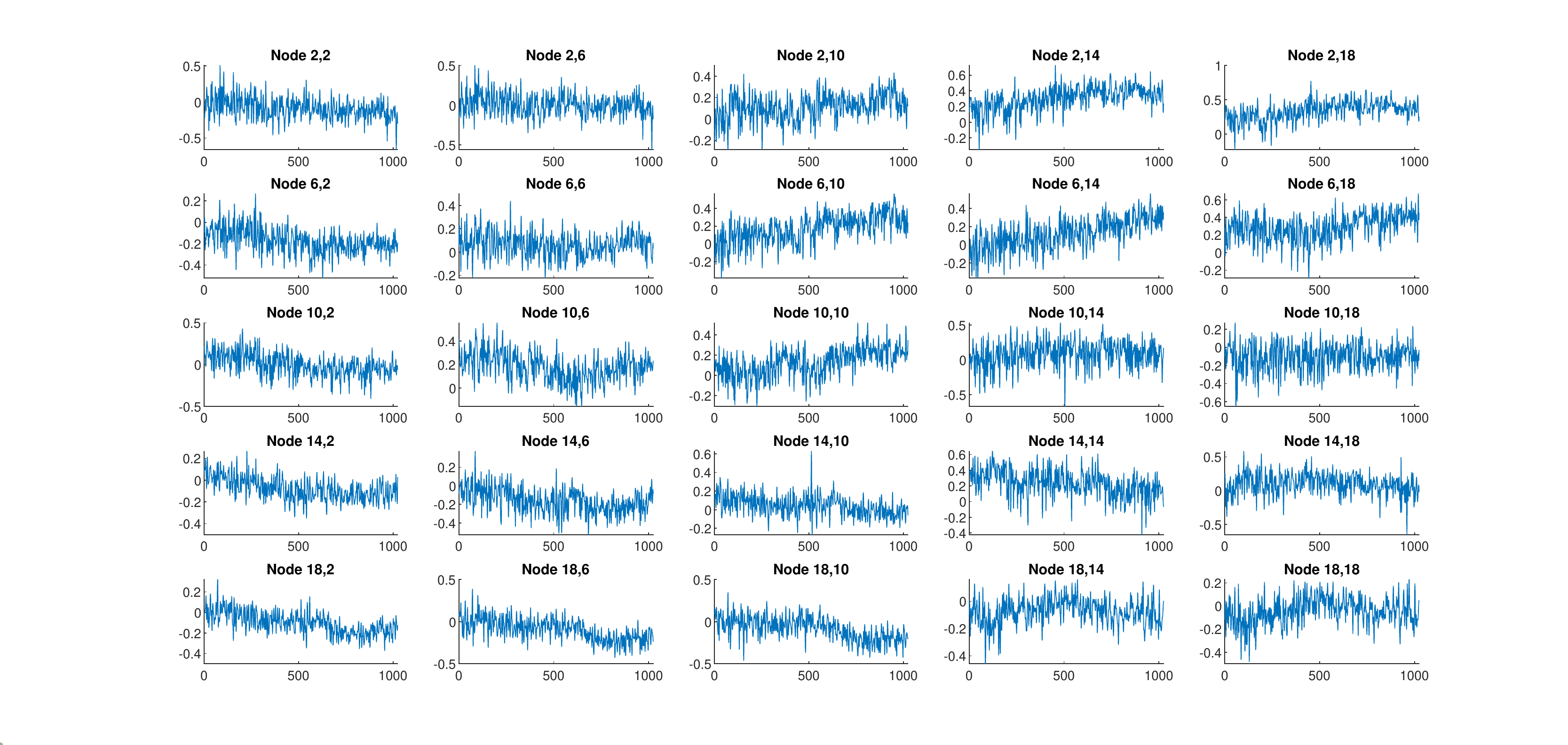}
\caption{Temporal and spatial interpolated data over a $20\times 20$ spatial regular grid}\label{DR10NS}
 \end{figure}

 \begin{figure}[h!]
\includegraphics[width=12cm,height=10cm]{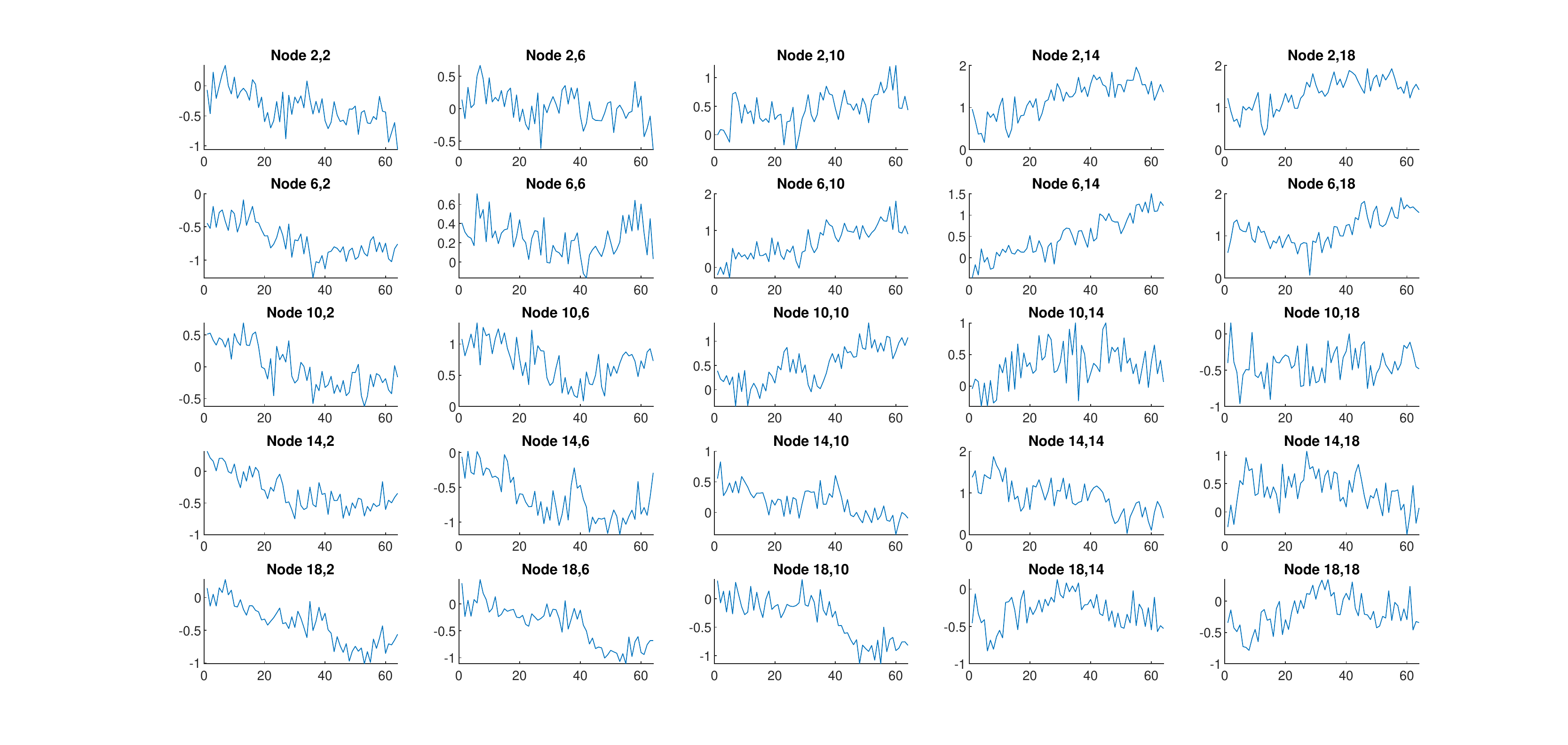}
\caption{Temporal and spatial interpolated data over  a $20\times 20$ spatial regular grid at scale (resolution level) 7}\label{D9}
 \end{figure}

 \begin{figure}[h!]
\includegraphics[width=12cm,height=10cm]{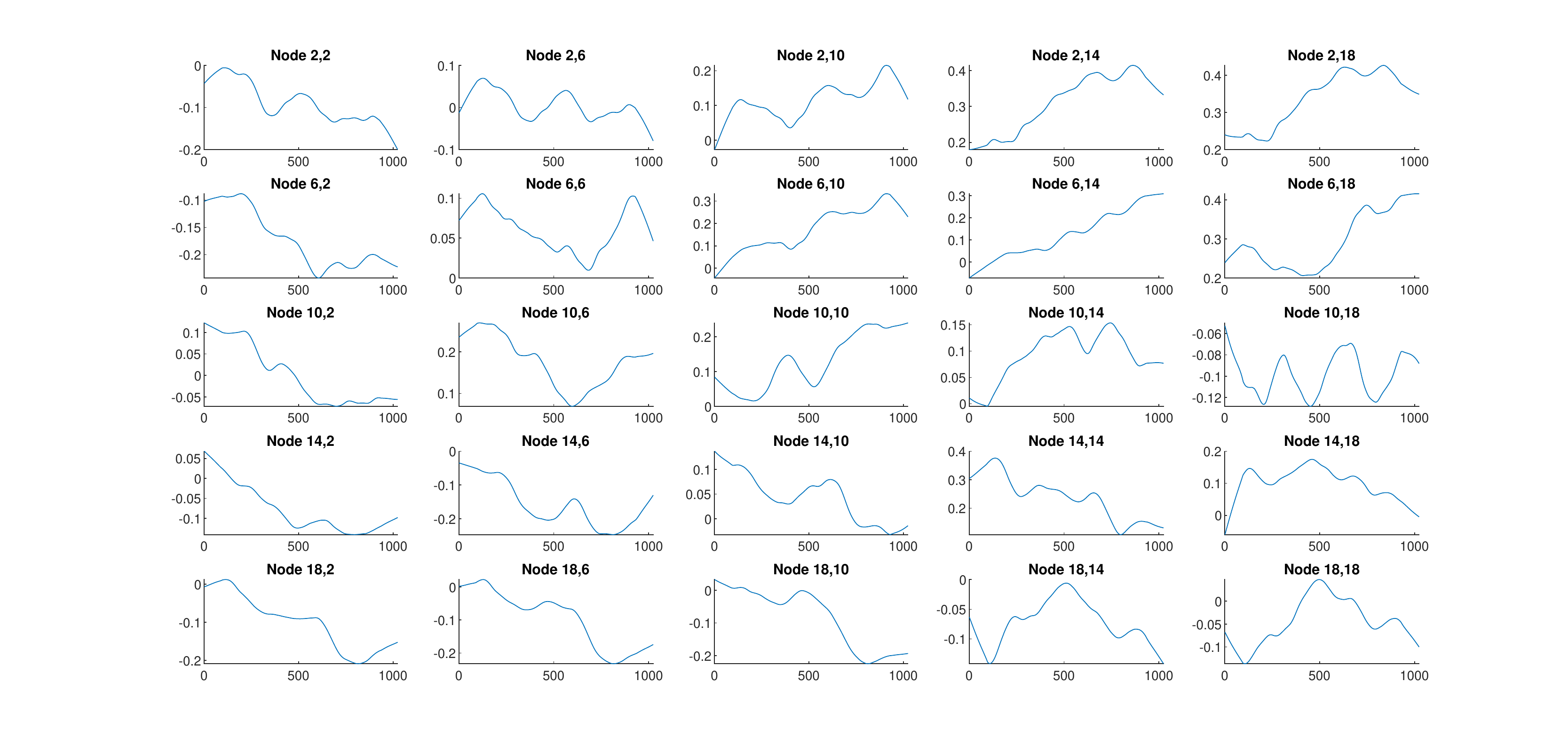}
\caption{B--spline smoothed curve data over  a $20\times 20$ spatial regular grid}\label{DR10S}
 \end{figure}
 \clearpage

\subsection{Multiscale estimation}
Equations (\ref{ps})--(\ref{SARHpp}) are implemented in terms of the empirical eigenvectors, and the Haar wavelet basis.
The computed  estimates at scales
(resolution levels) $j=7,8,9,10,$  of the autocorrelation operators $L_{1}$ and $L_{2}$ can be found  in Figure \ref{OPL1EA},  for $k_{N}=\left[\ln(N)\right]^{-}= \left[\ln(400)\right]^{-}=5=k_{400}.$
Contour plots in Figure \ref{figa1} display the spatial patterns of the observed and estimated log--intensity field over  a $20\times 20$ spatial regular grid, at monthly times $t=108$ and $t=216,$ through scales $j=7,8,9,10.$ Here, the multiscale analysis has been implemented from the interpolated non--smoothed data. Figure \ref{figa2} shows the original and estimated values  of the log--intensity field over the same temporal and spatial nodes,  from the  B-spline smoothed curve data. One can observe the loss of information in  Figure \ref{figa2},  about spatial variability displayed by the log--intensity field at scales $j=9,10,$ with respect to   Figure \ref{figa1}. Thus, similar spatial patterns are  observed, at  scales $j=7,8,9,10,$ when B--spline smoothed curve data are considered, hiding the heterogeneities that the log-intensity field presents through different scales.

\begin{figure}[h!]
\includegraphics[width=12cm,height=10cm]{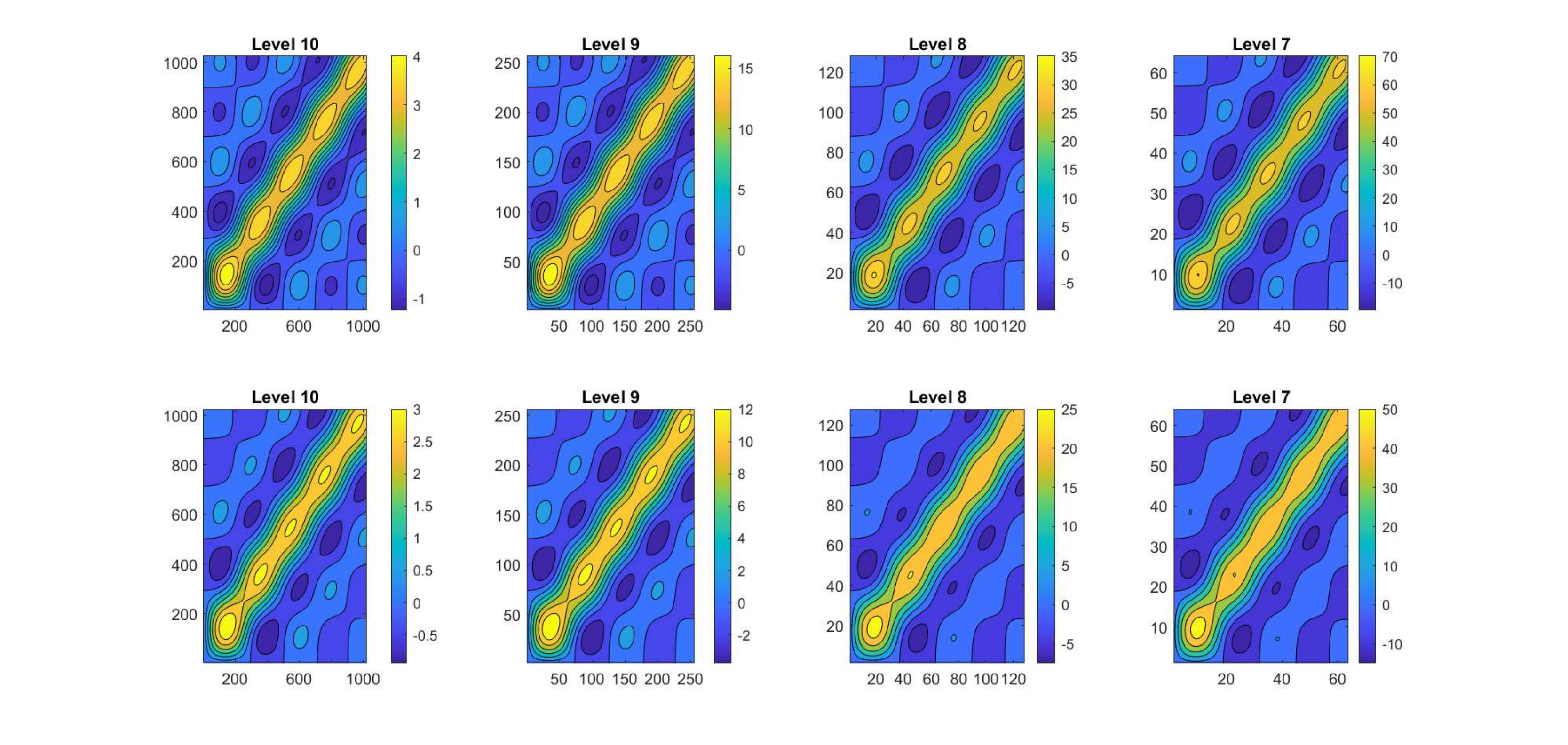}
\caption{Multiscale estimation of operator $L_{1}$ (top), and of $L_{2}$ (bottom) through scales (multiresolution levels) $j=7,8,9,10$ (from  right to left)}\label{OPL1EA}
 \end{figure}

\begin{center}
 \begin{figure}[!h]
 \centering
\includegraphics[width=2.5cm,height=2.5cm]{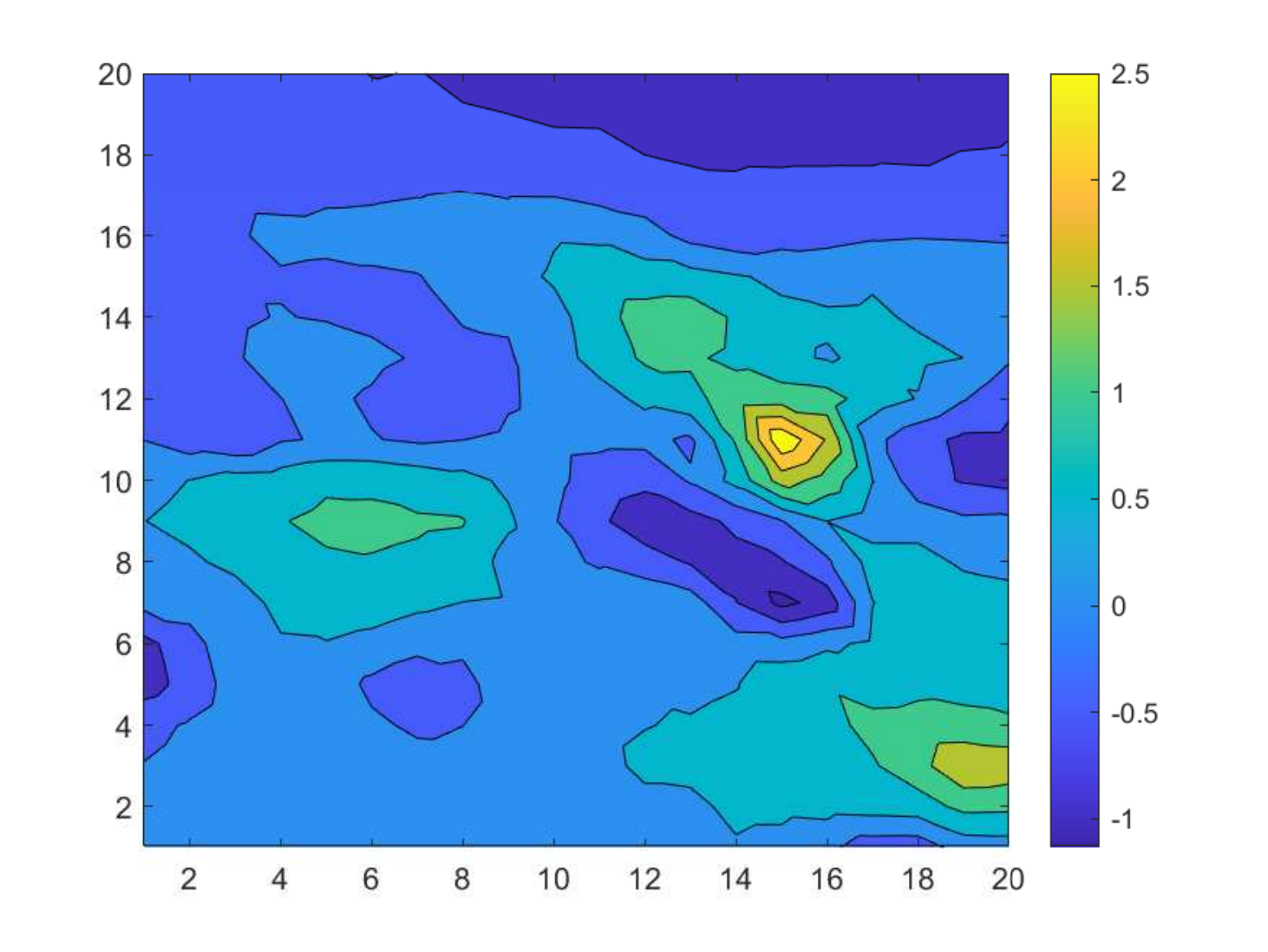}
        \includegraphics[width=2.5cm,height=2.5cm]{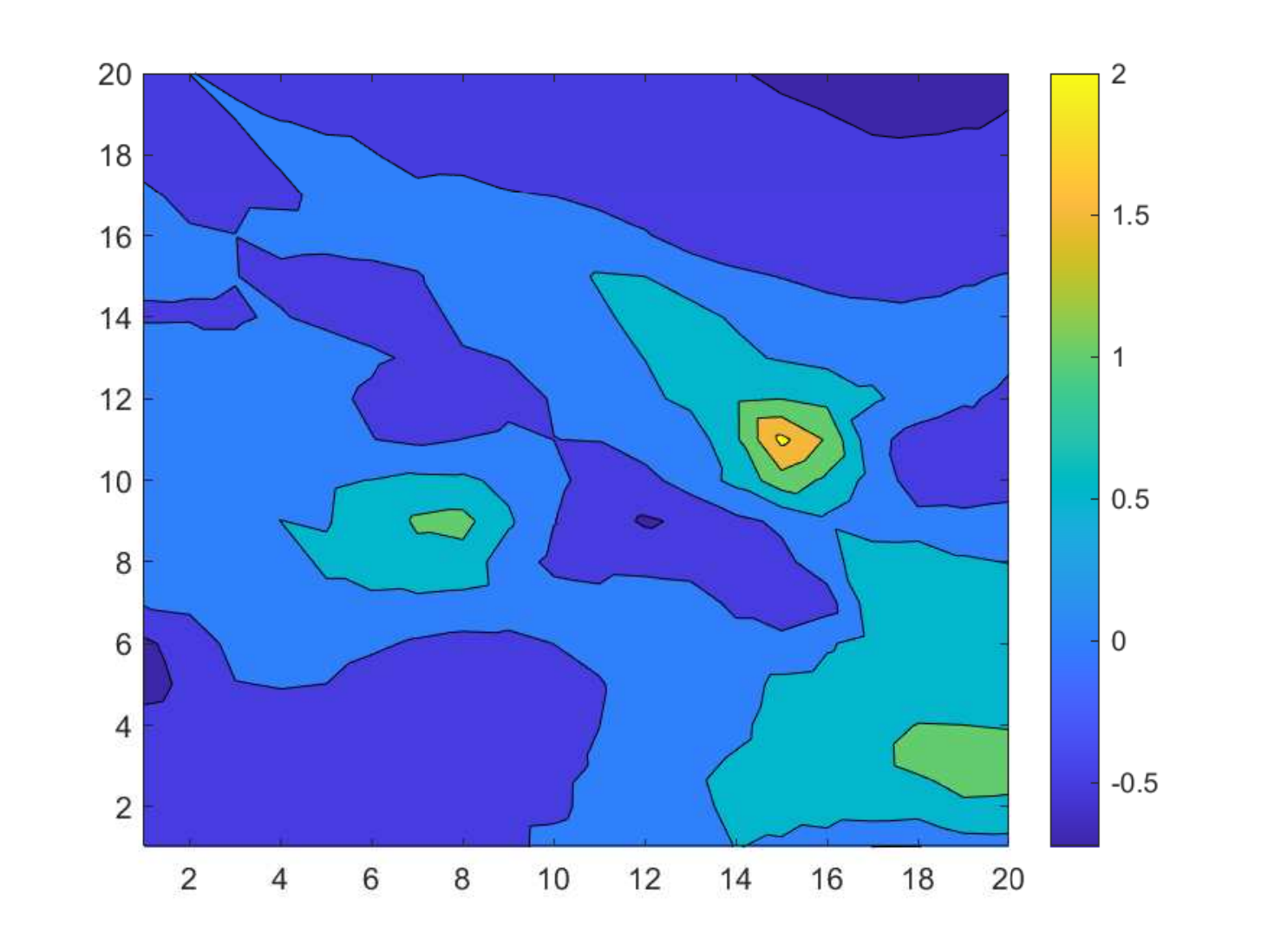}
        \includegraphics[width=2.5cm,height=2.5cm]{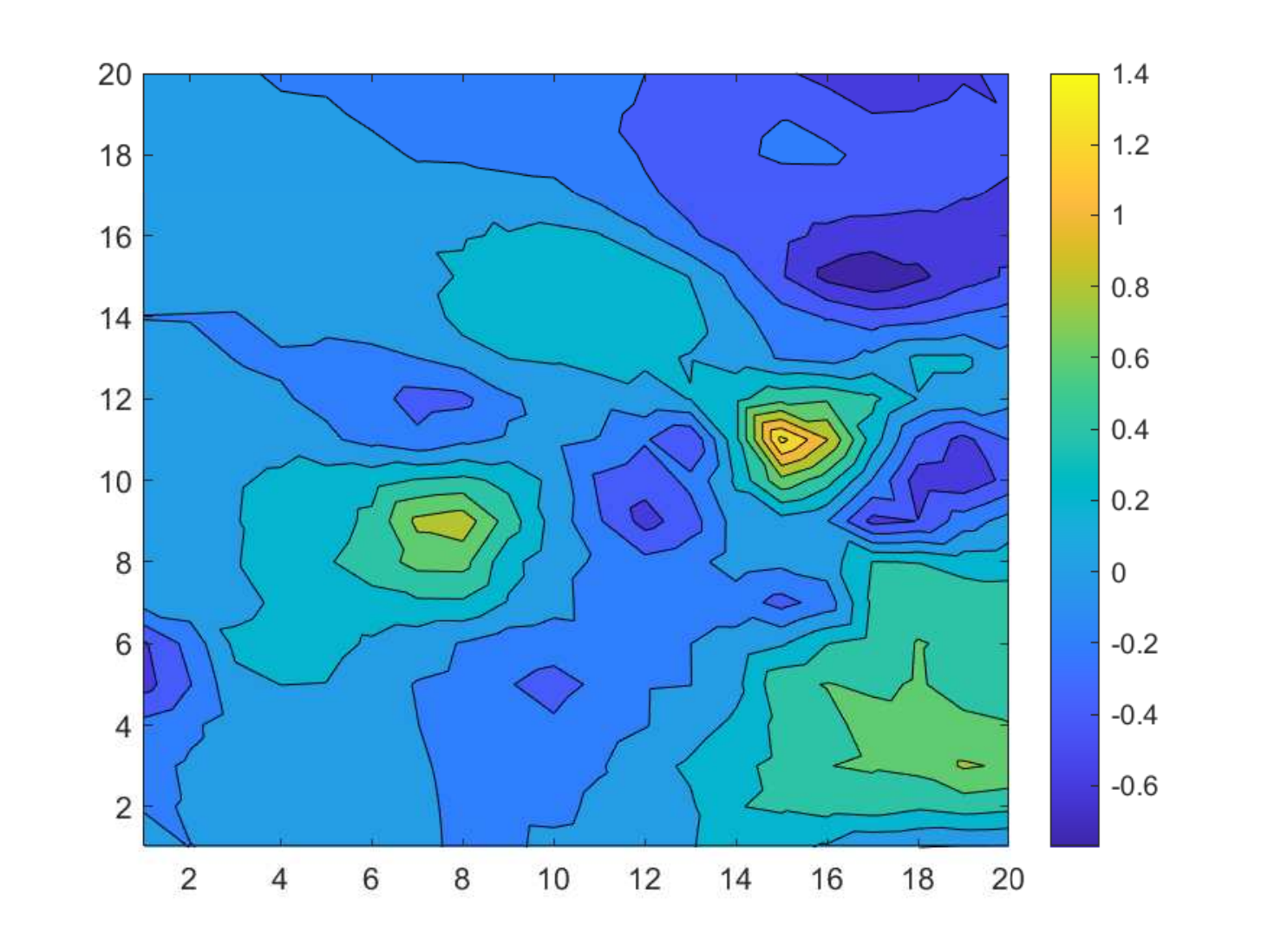}
        \includegraphics[width=2.5cm,height=2.5cm]{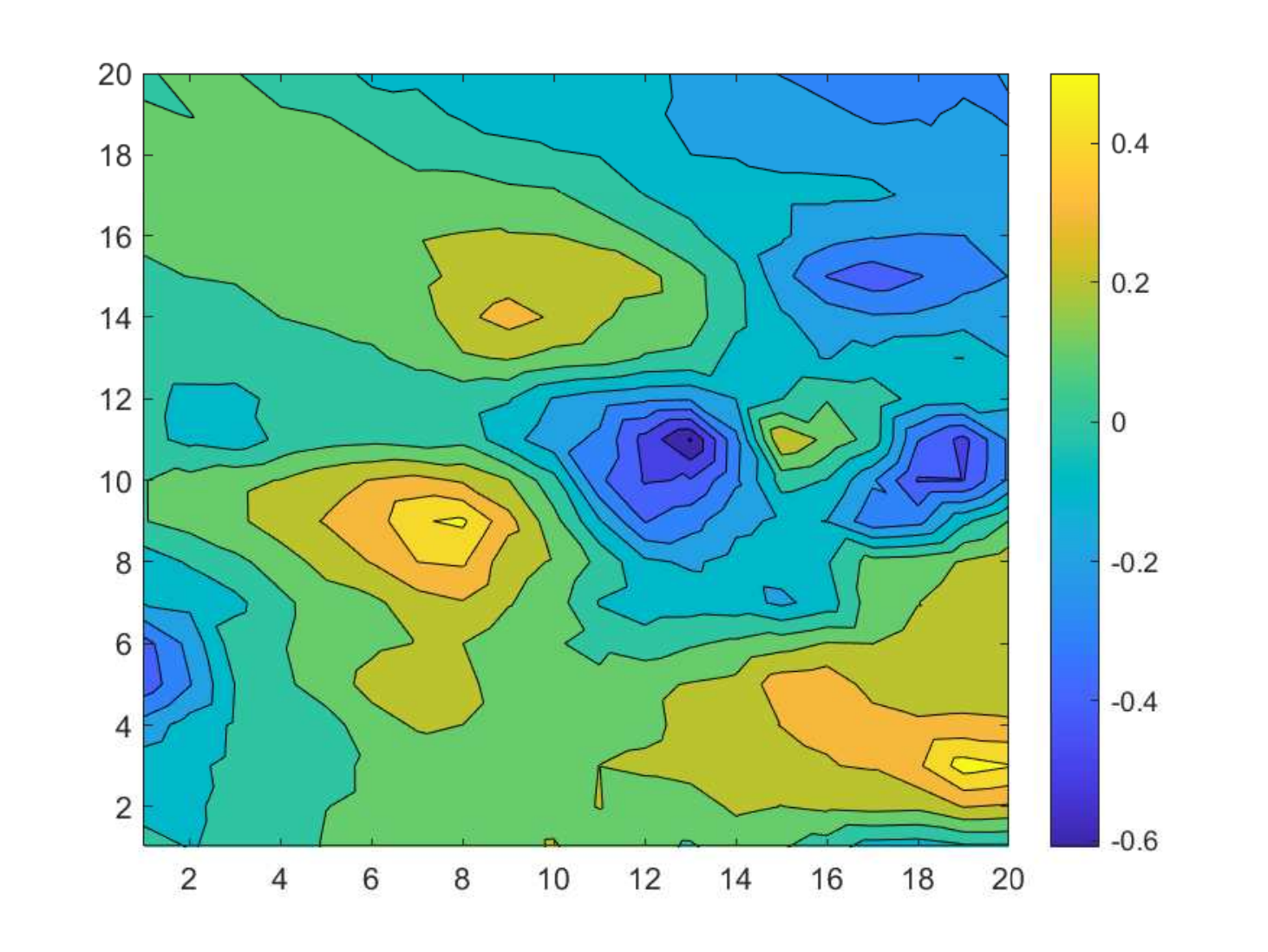}
        \hspace*{0.56cm}
        \includegraphics[width=2.5cm,height=2.5cm]{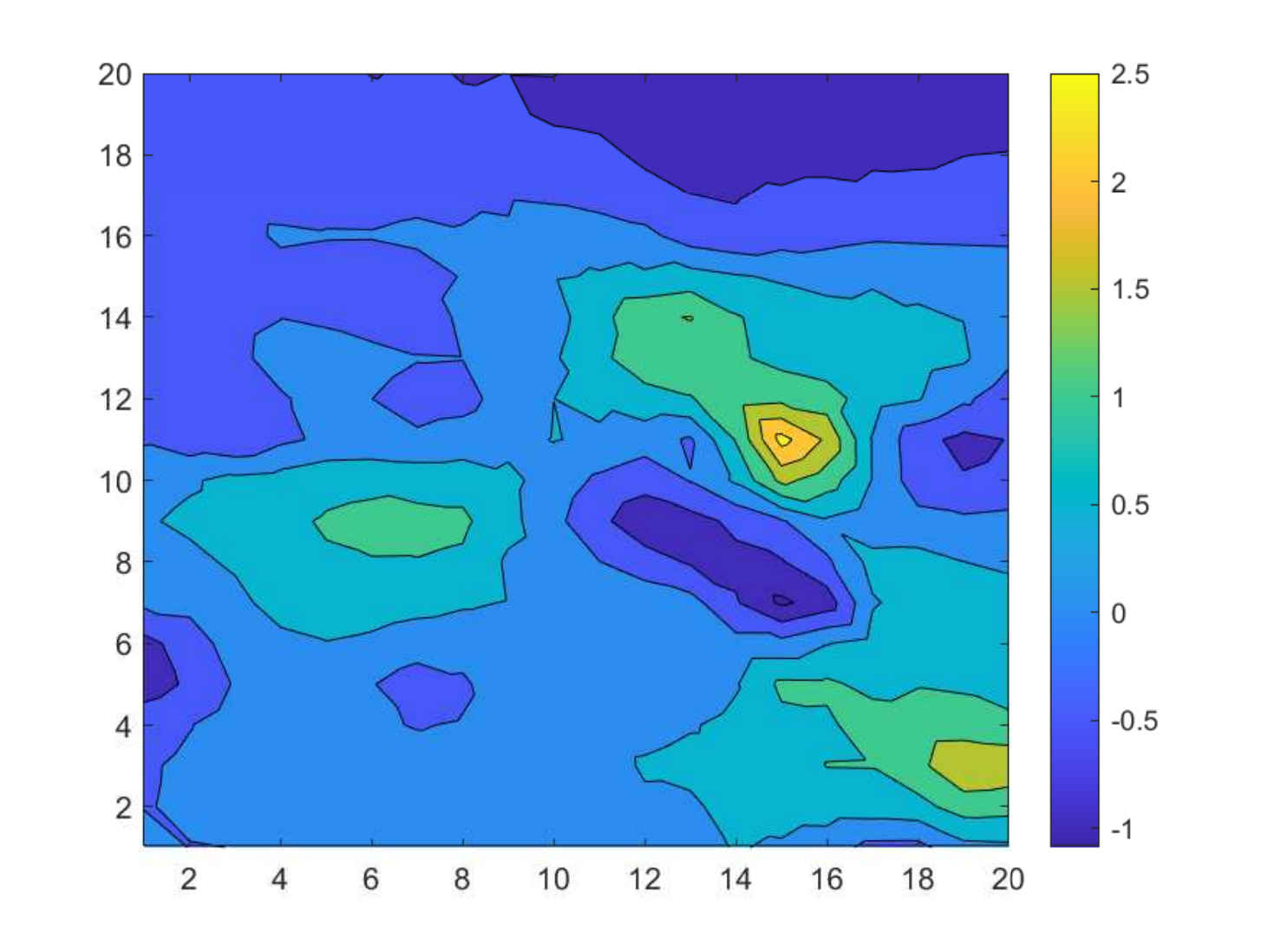}
        \includegraphics[width=2.5cm,height=2.5cm]{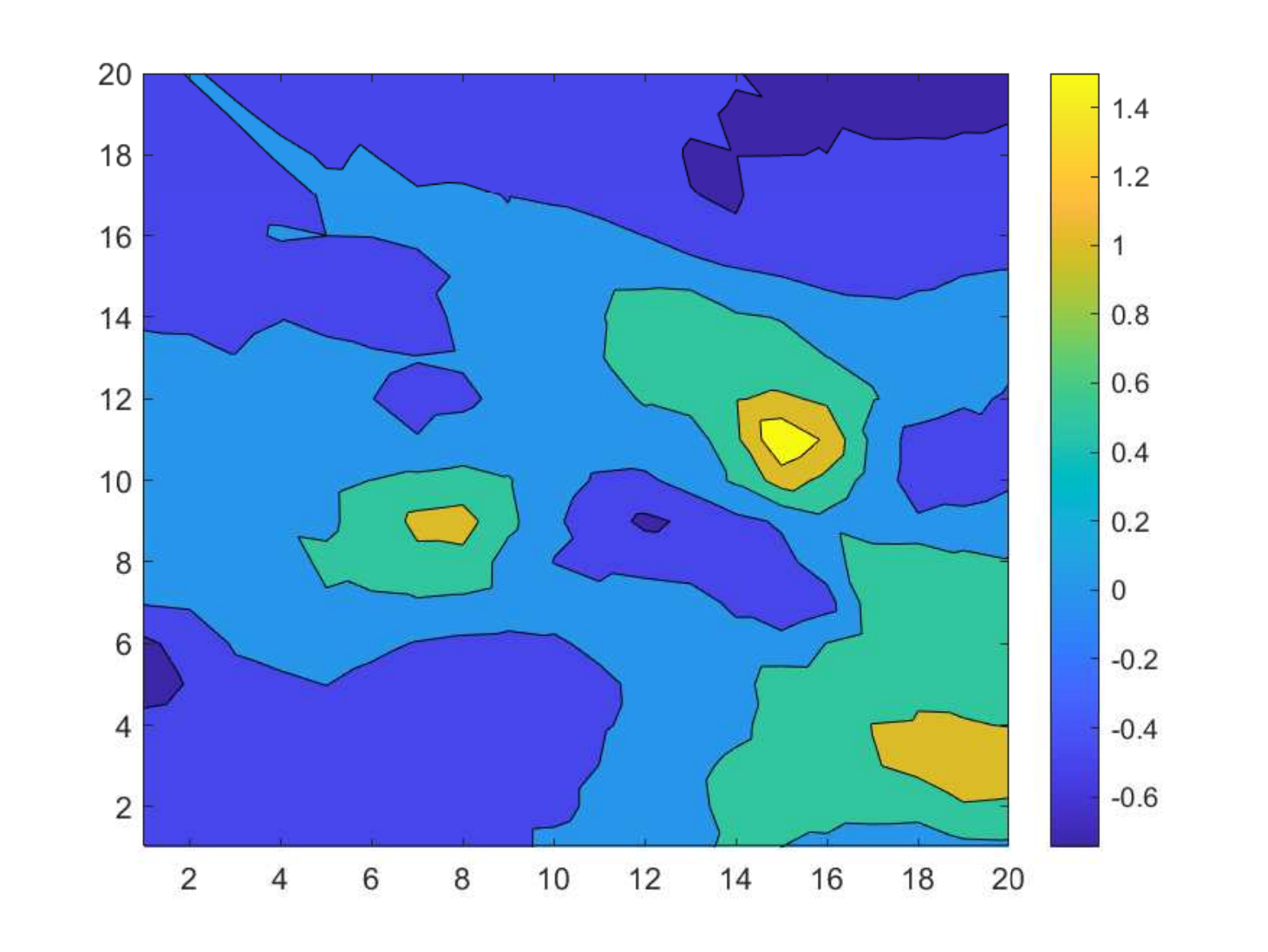}
        \includegraphics[width=2.5cm,height=2.5cm]{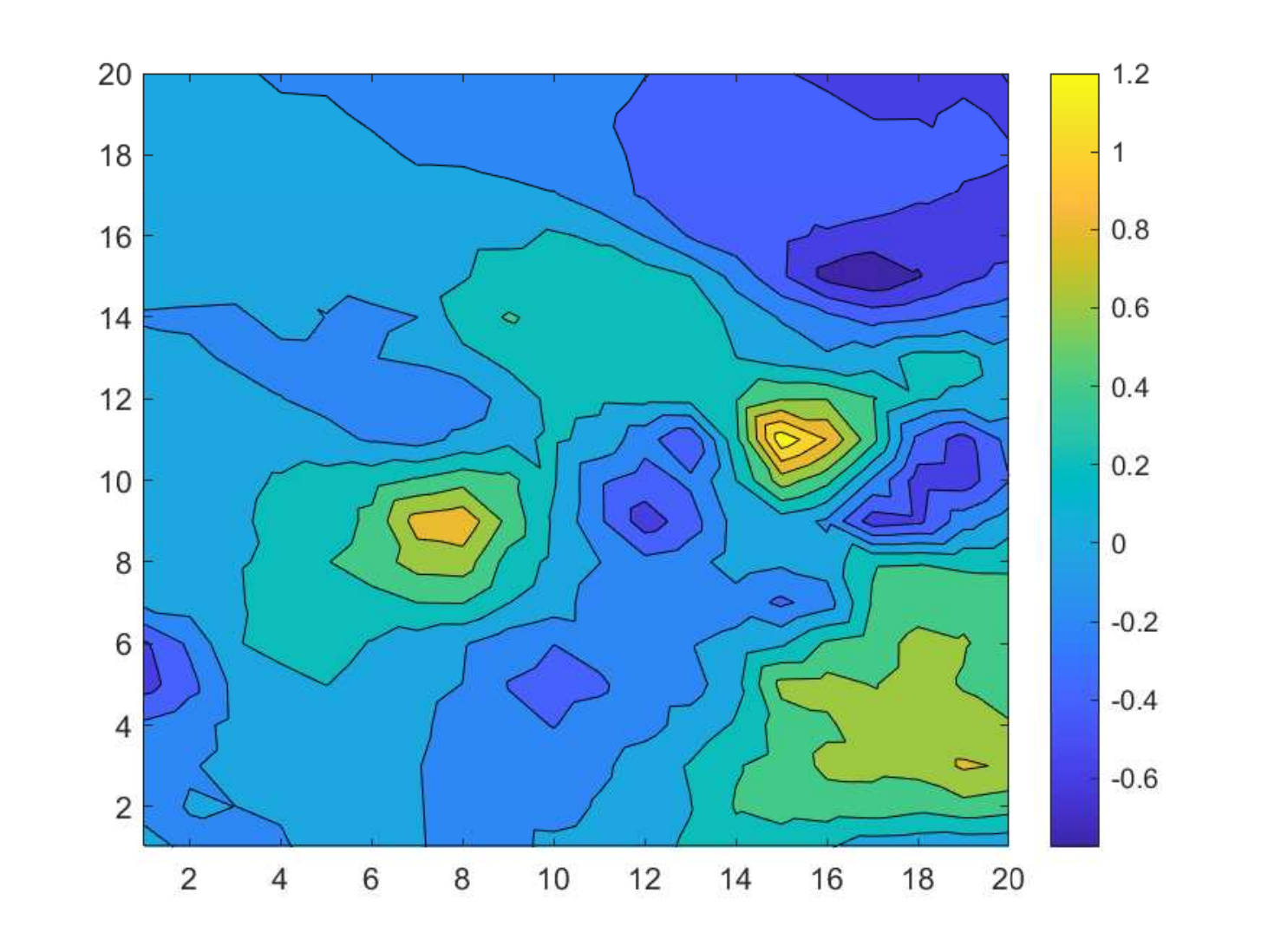}
        \includegraphics[width=2.5cm,height=2.5cm]{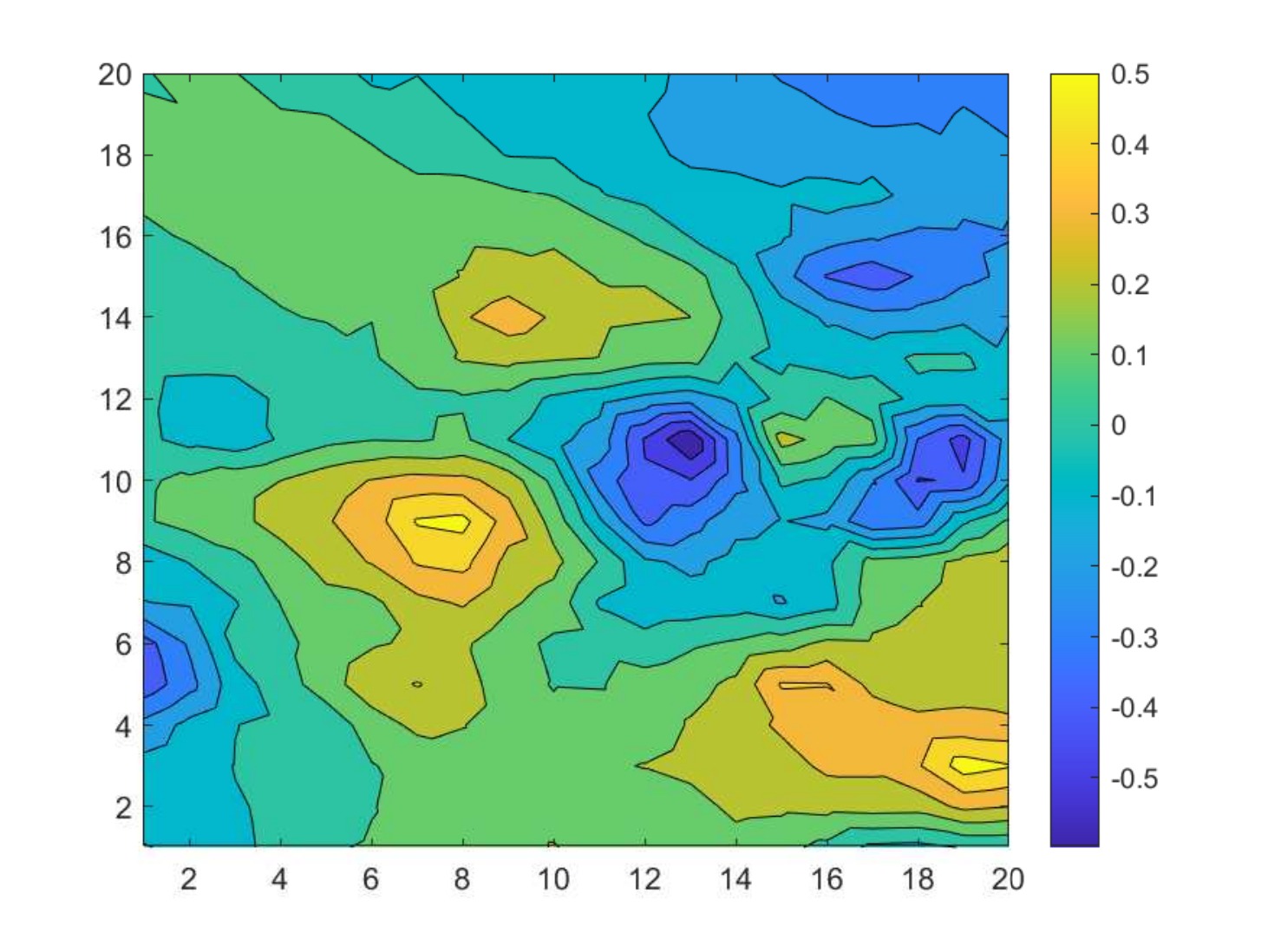}
        \hspace*{0.56cm}
        \includegraphics[width=2.5cm,height=2.5cm]{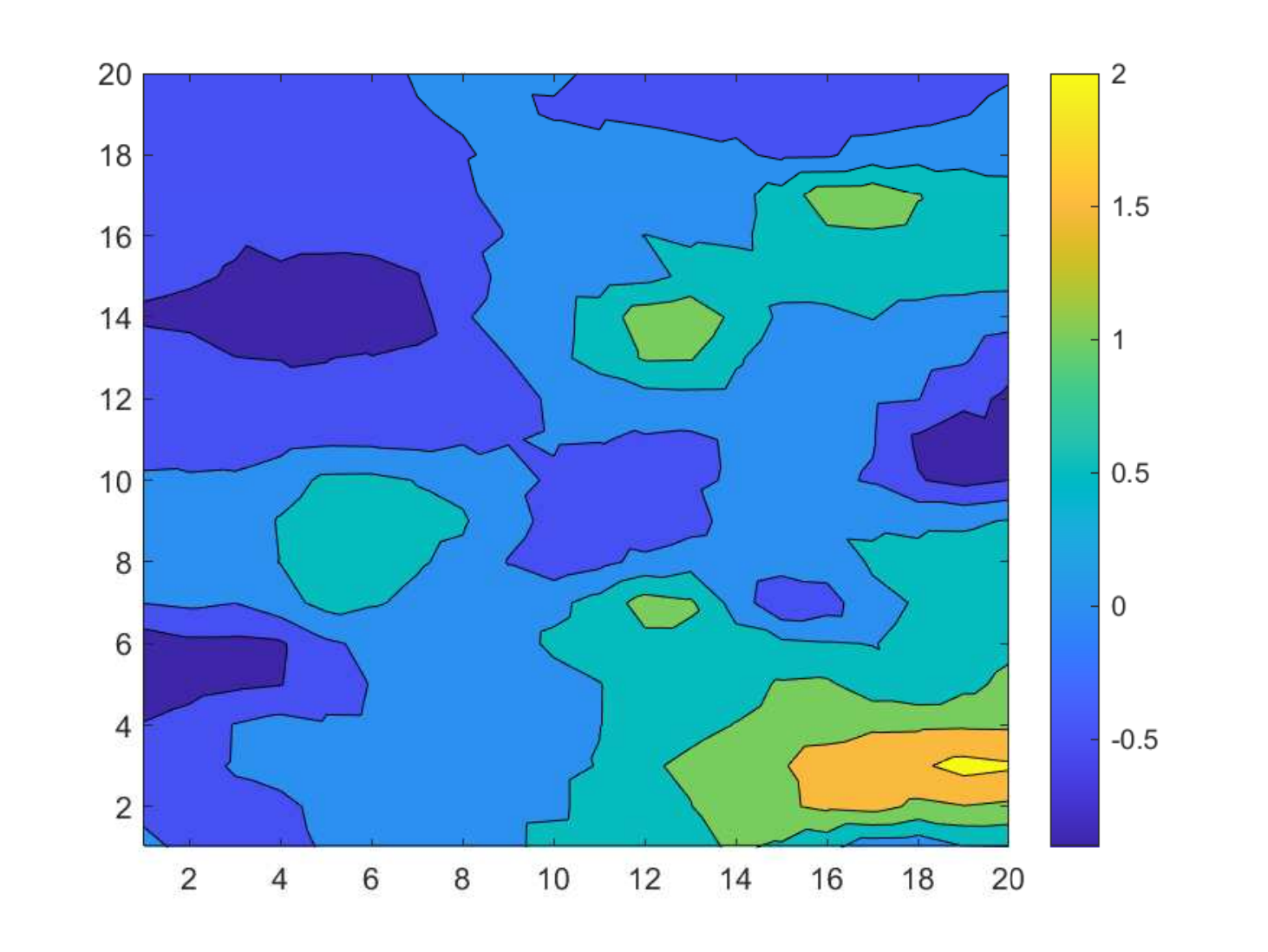}
        \includegraphics[width=2.5cm,height=2.5cm]{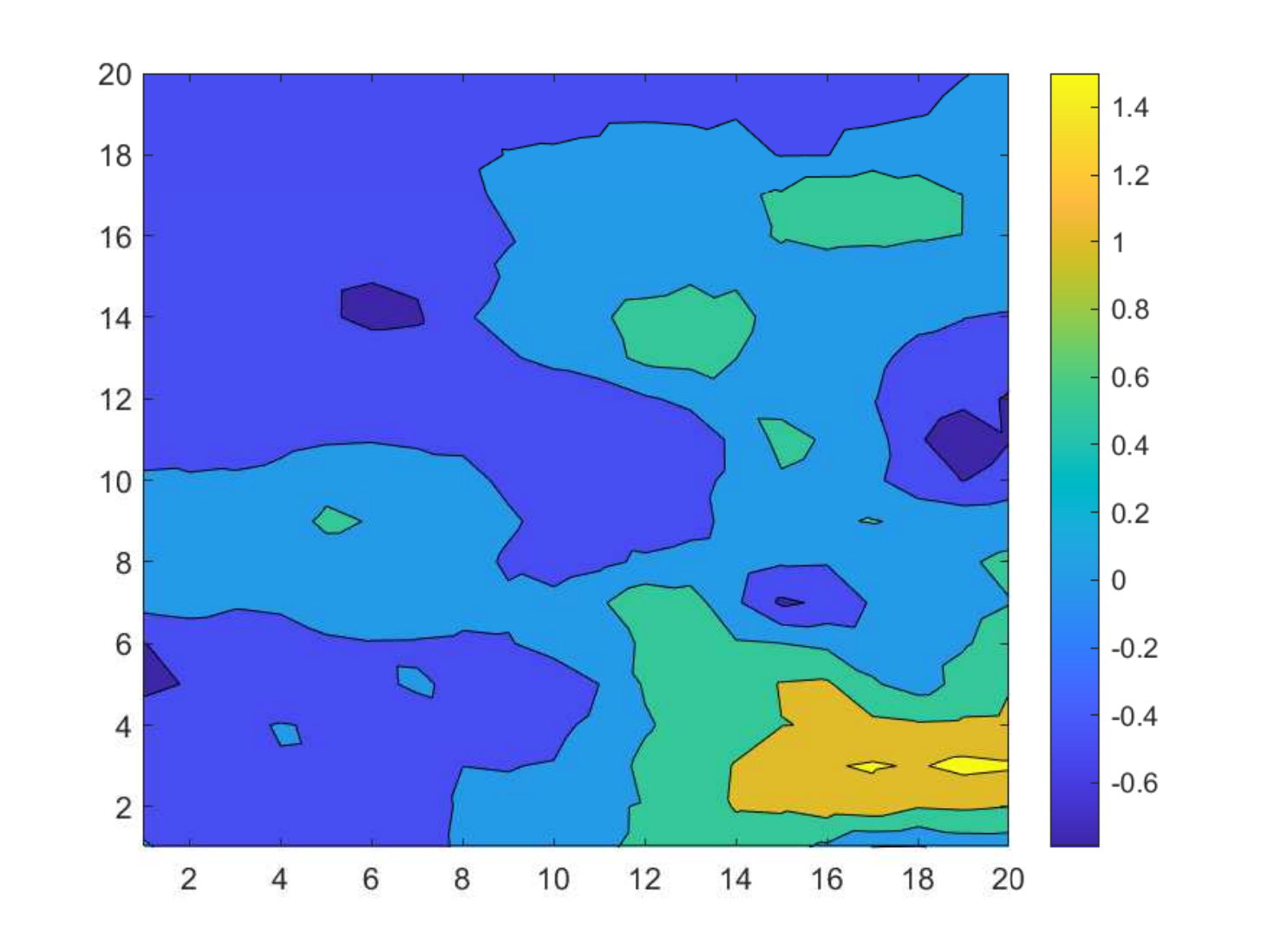}
        \includegraphics[width=2.5cm,height=2.5cm]{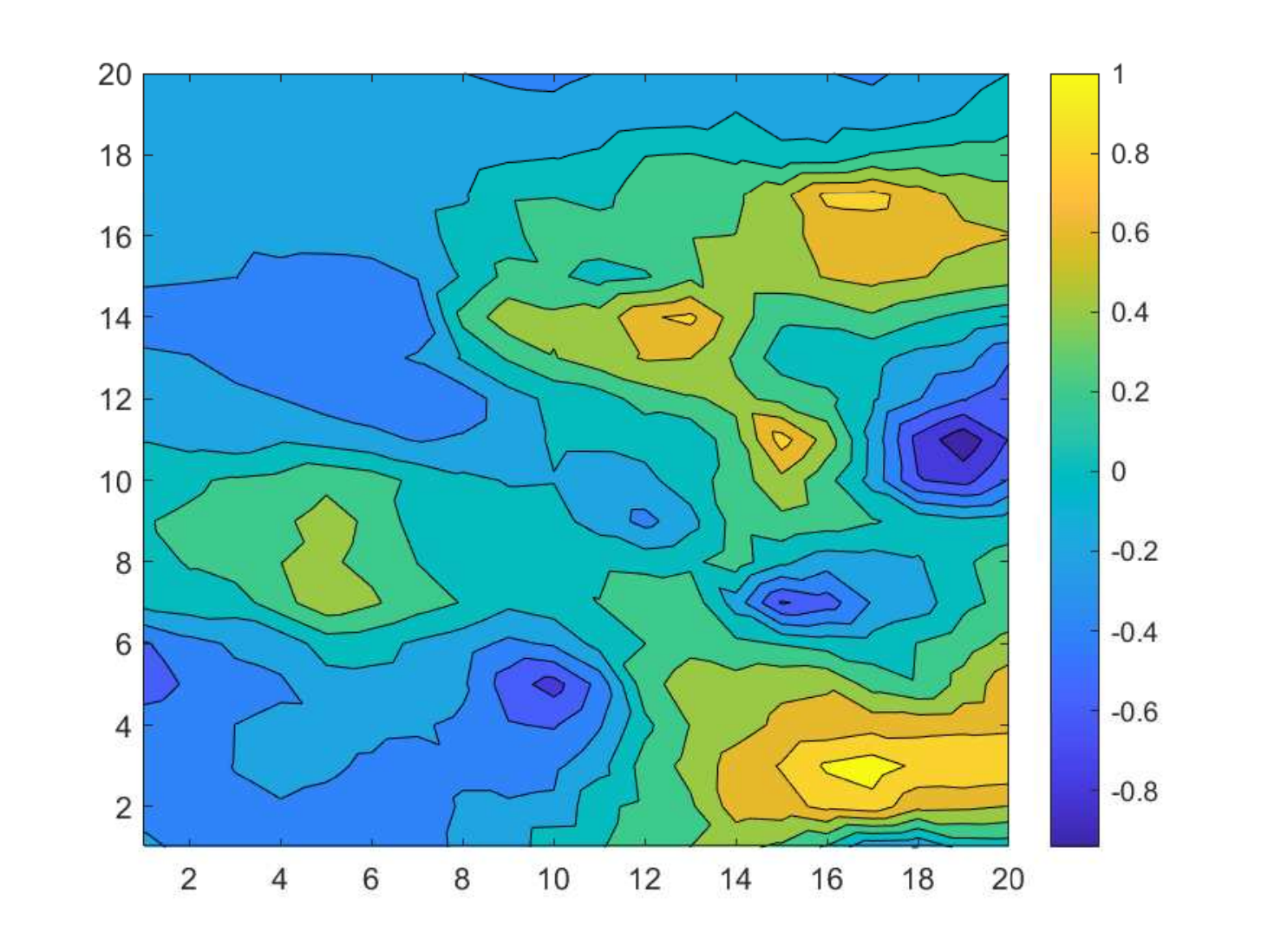}
        \includegraphics[width=2.5cm,height=2.5cm]{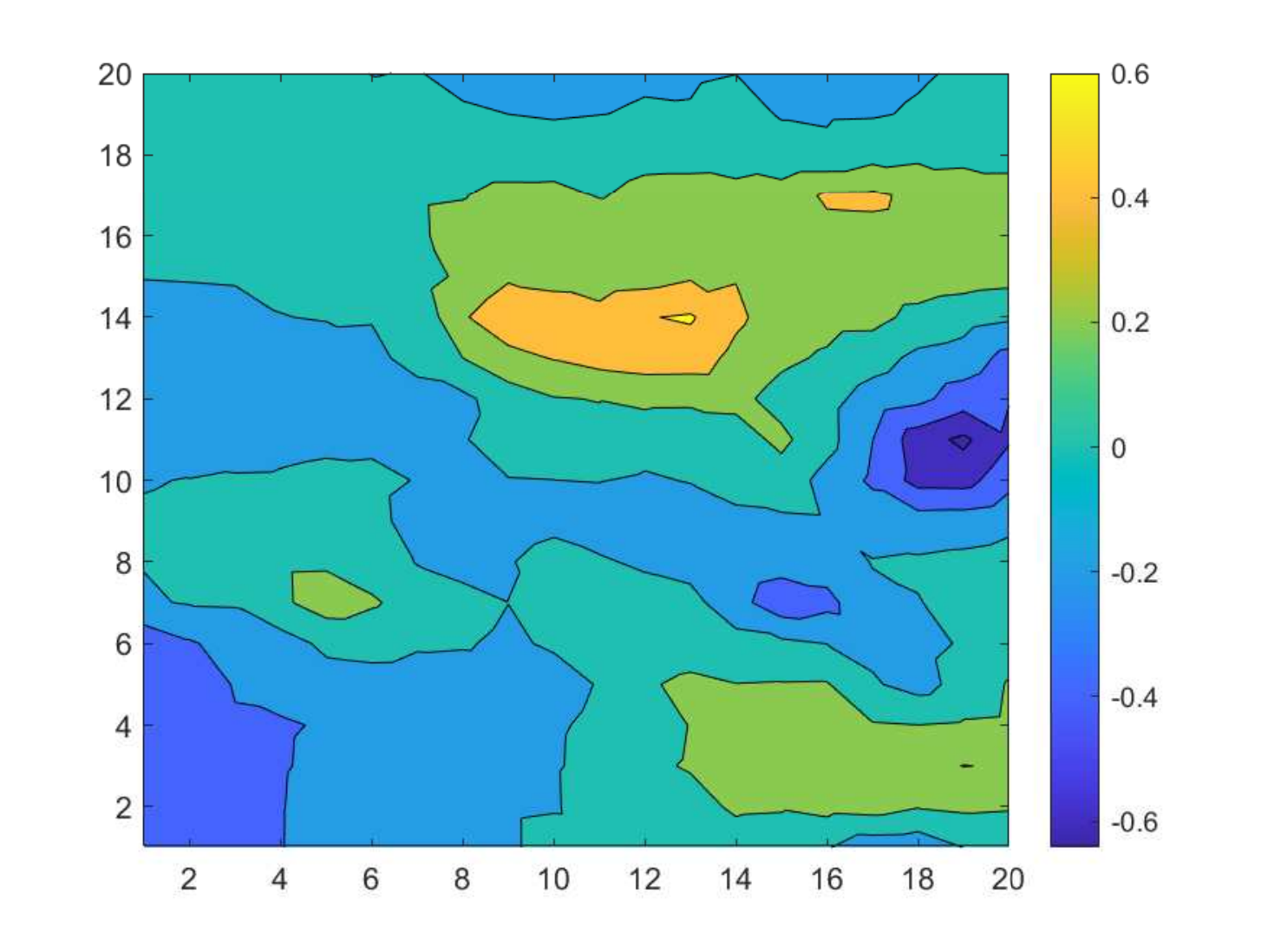}
        \hspace*{0.56cm}
        \includegraphics[width=2.5cm,height=2.5cm]{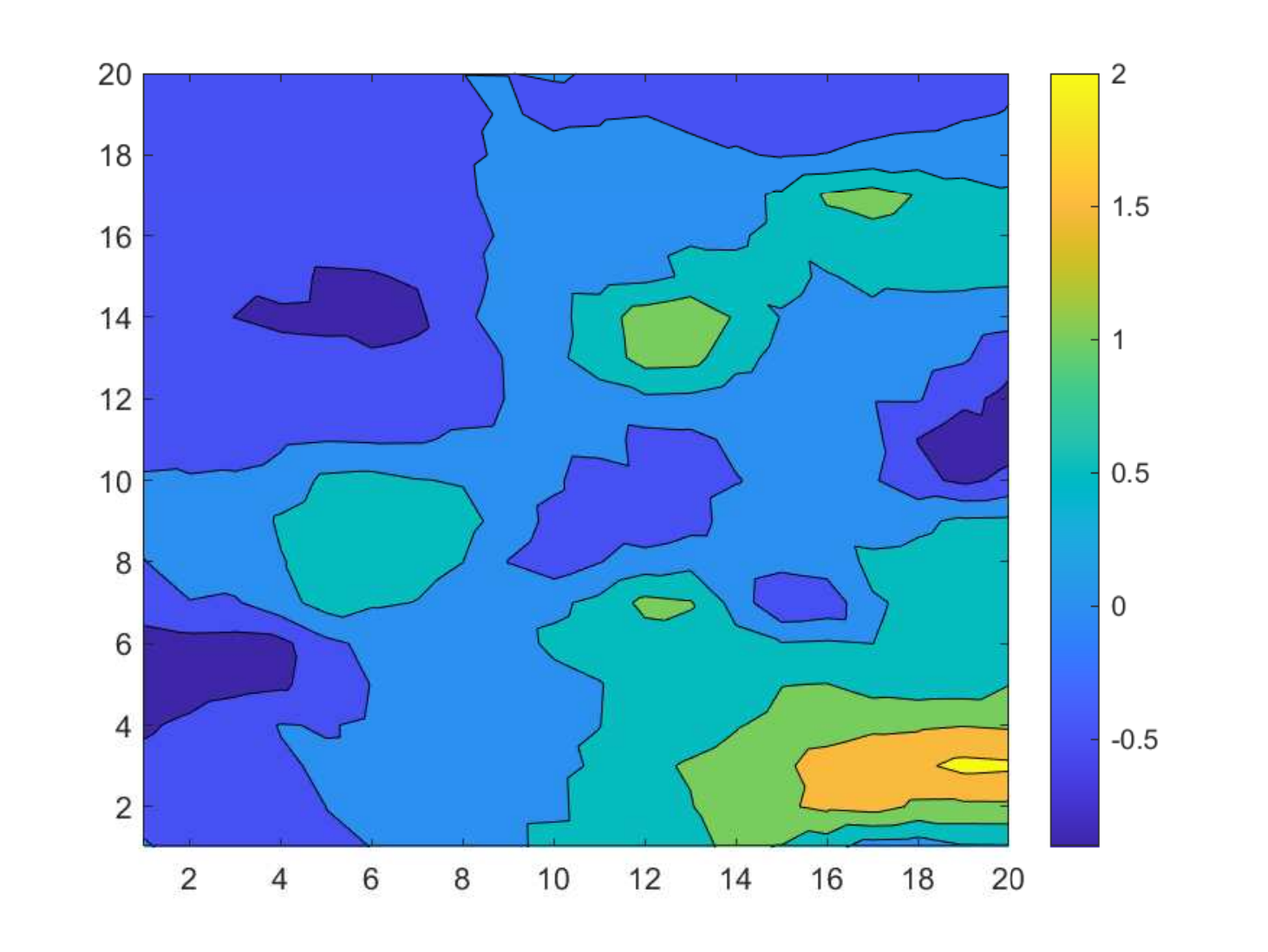}
        \includegraphics[width=2.5cm,height=2.5cm]{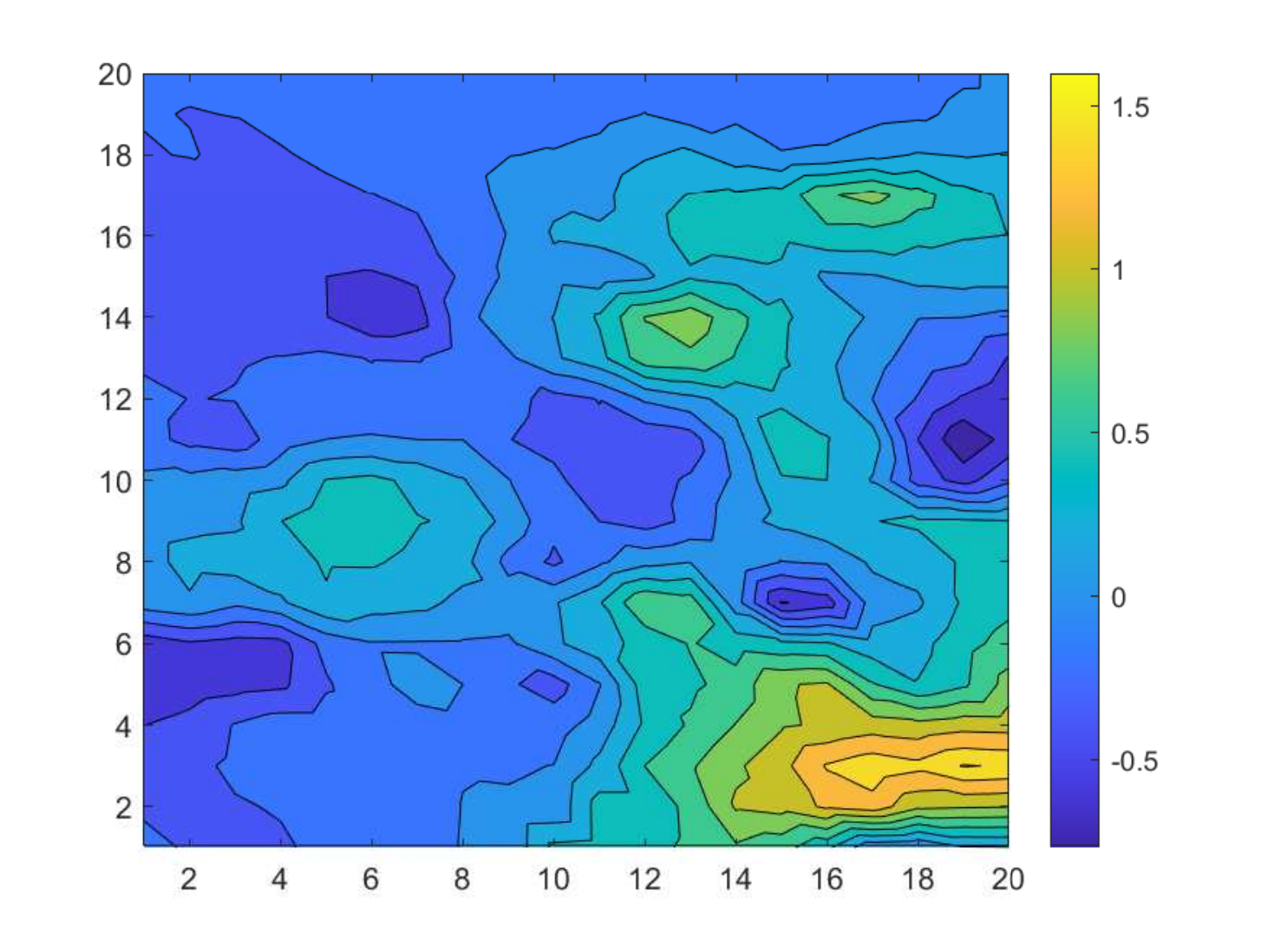}
        \includegraphics[width=2.5cm,height=2.5cm]{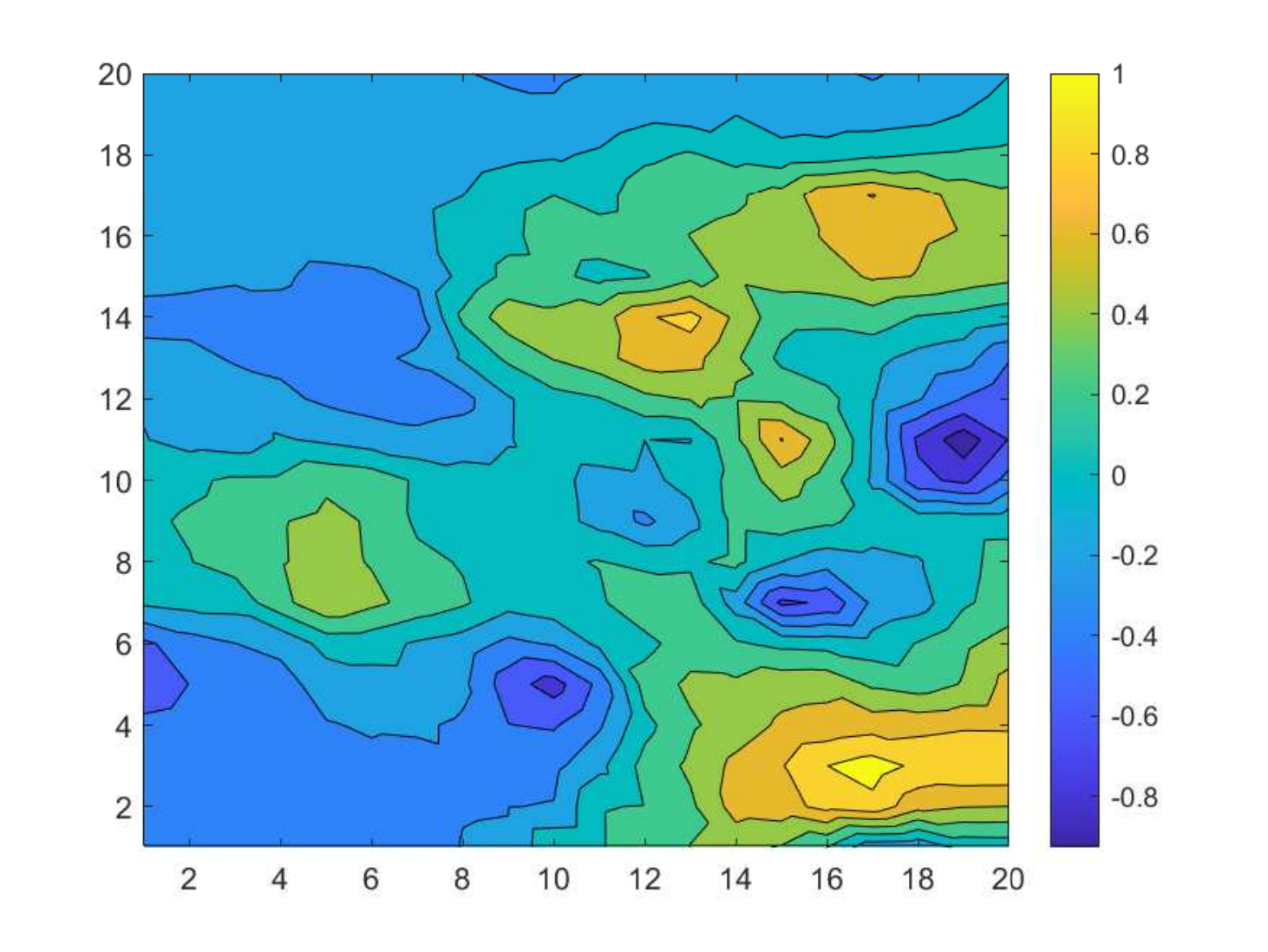}
        \includegraphics[width=2.5cm,height=2.5cm]{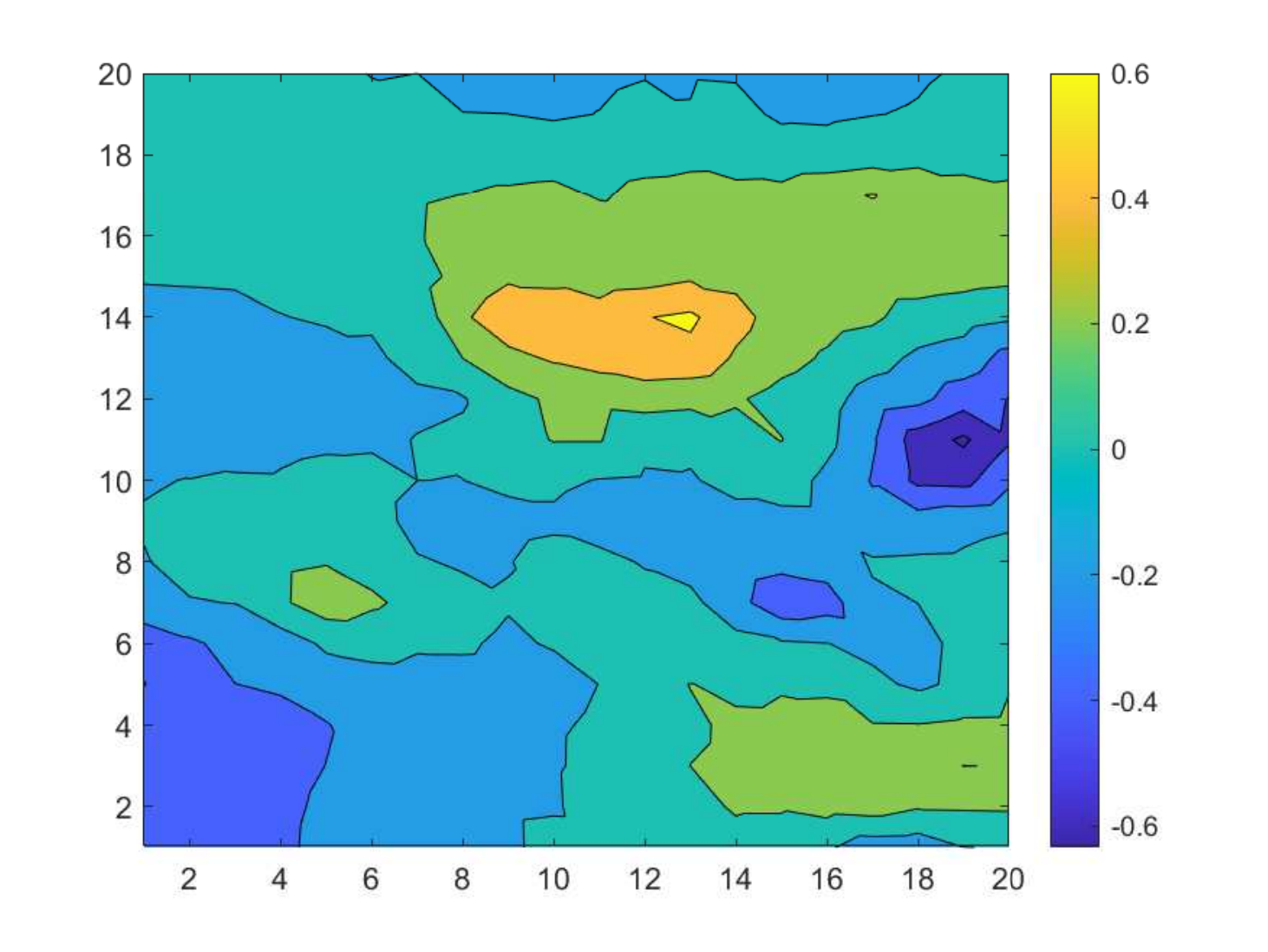}
        \caption{Contour plots of the observed log--intensity field at monthly times $t=108$ (top--row) and $t=216$ (third--row), and  the estimated log--intensity field at $t=108$ (second--row) and   $t=216$ (bottom--row). Both observed and estimated values at times $t=108,$ and  $t=216$  are displayed  through the scales $j=7,8,9,10$ (from left to right), in the Haar wavelet system, from the  temporal and spatial interpolated data over a $20\times 20 $  regular grid}
 \label{figa1}
 \end{figure}
\end{center}

\begin{center}
 \begin{figure}[!h]
 \centering
\includegraphics[width=2.5cm,height=2.5cm]{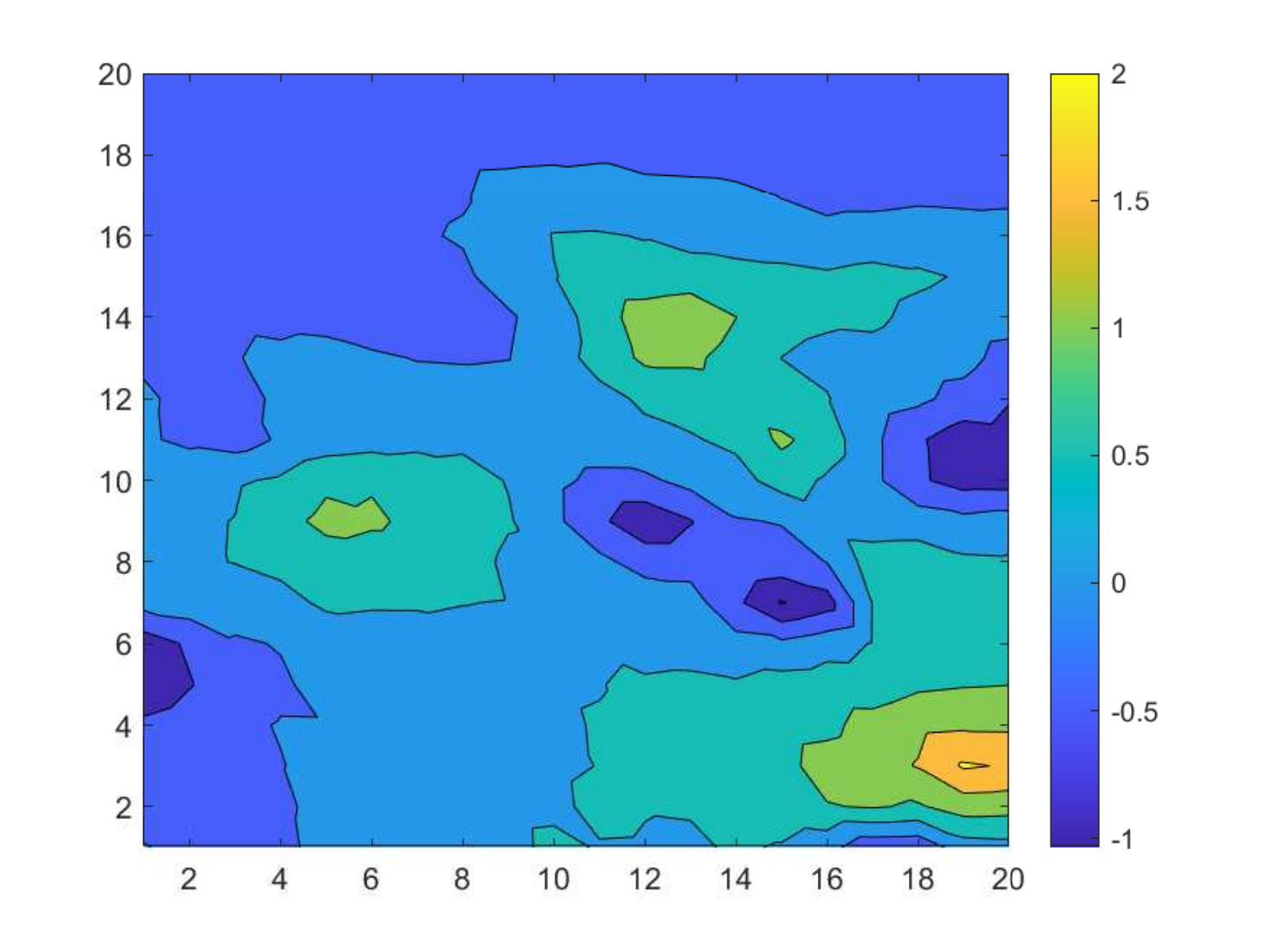}
        \includegraphics[width=2.5cm,height=2.5cm]{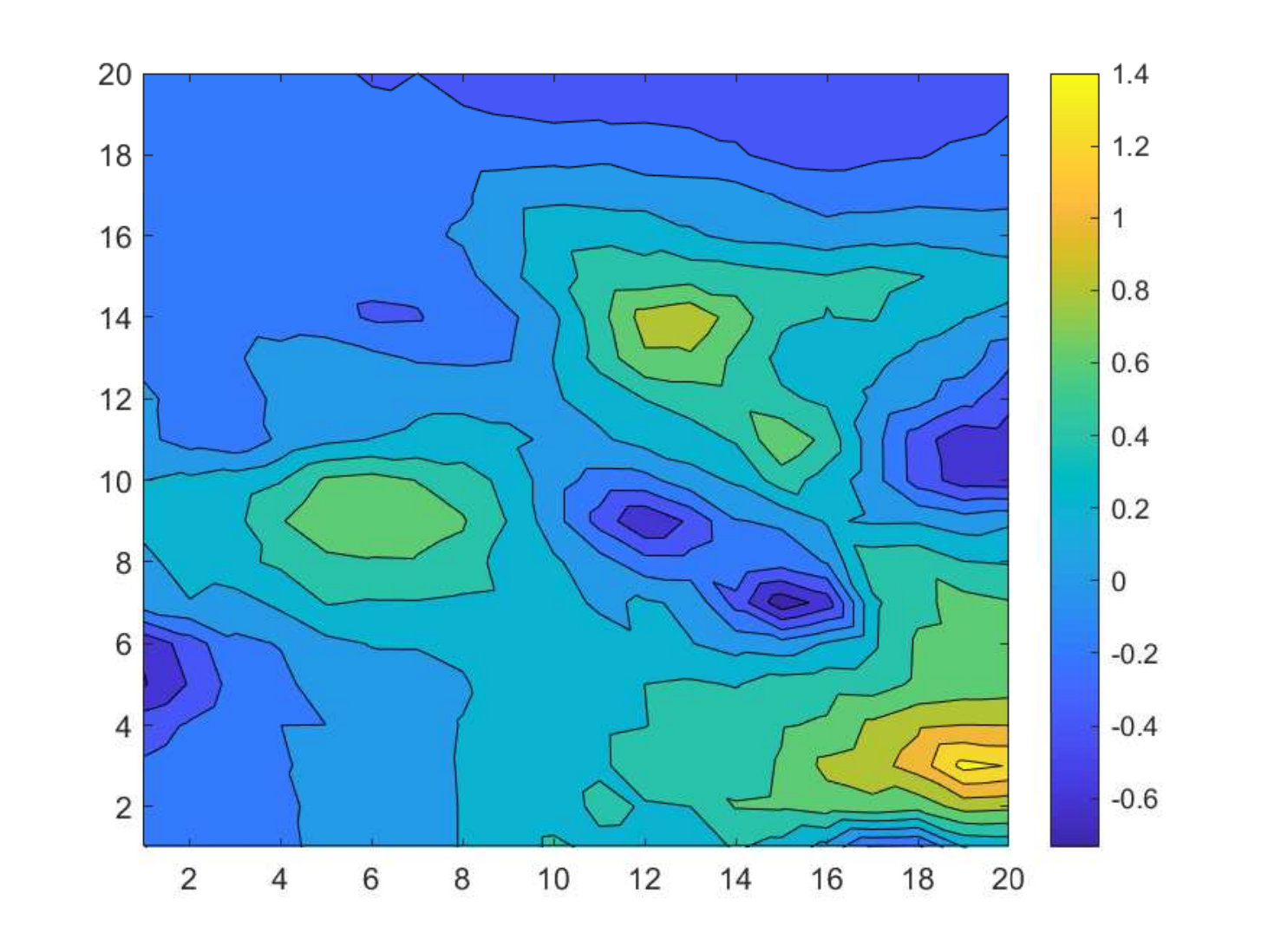}
        \includegraphics[width=2.5cm,height=2.5cm]{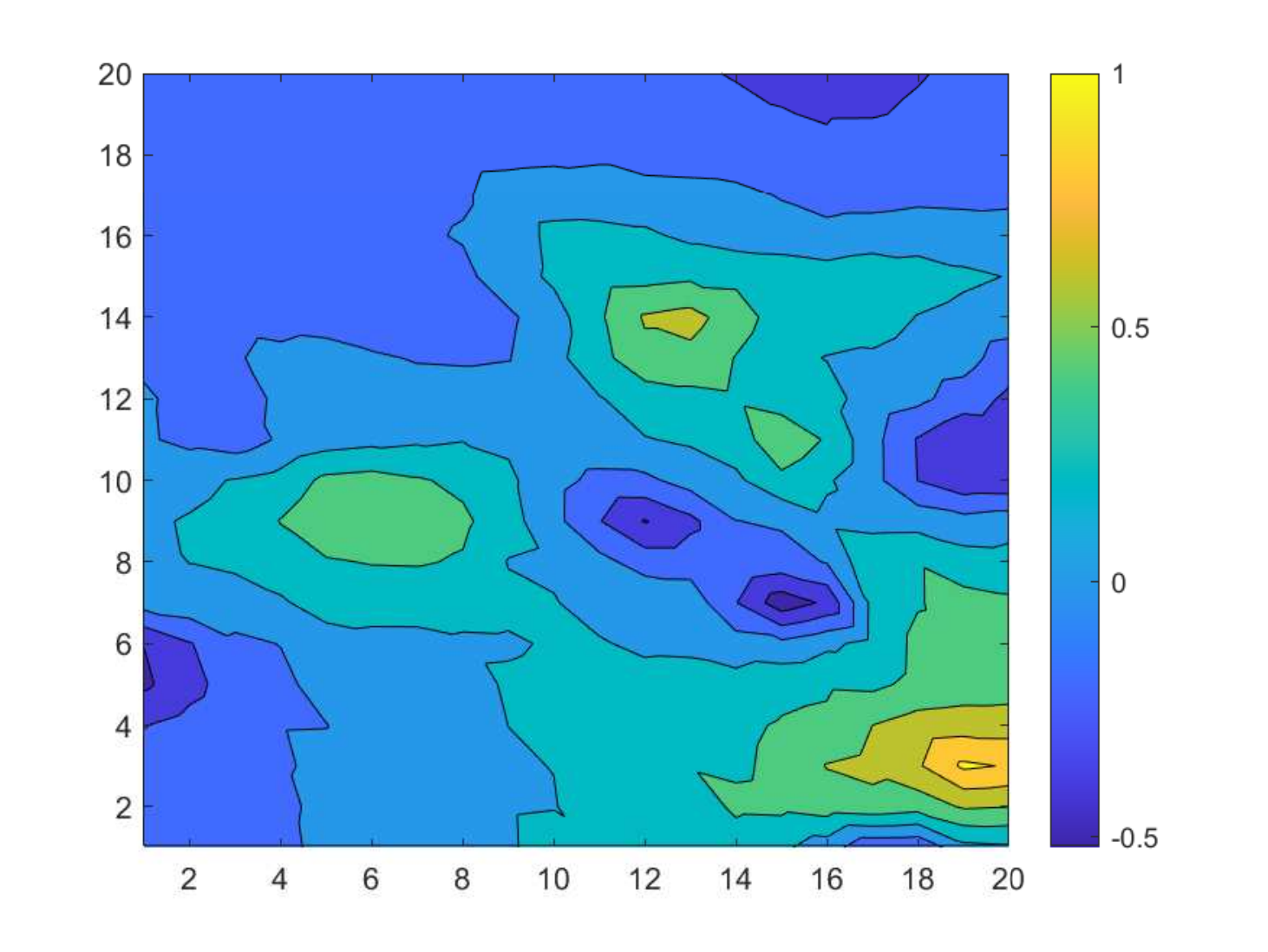}
        \includegraphics[width=2.5cm,height=2.5cm]{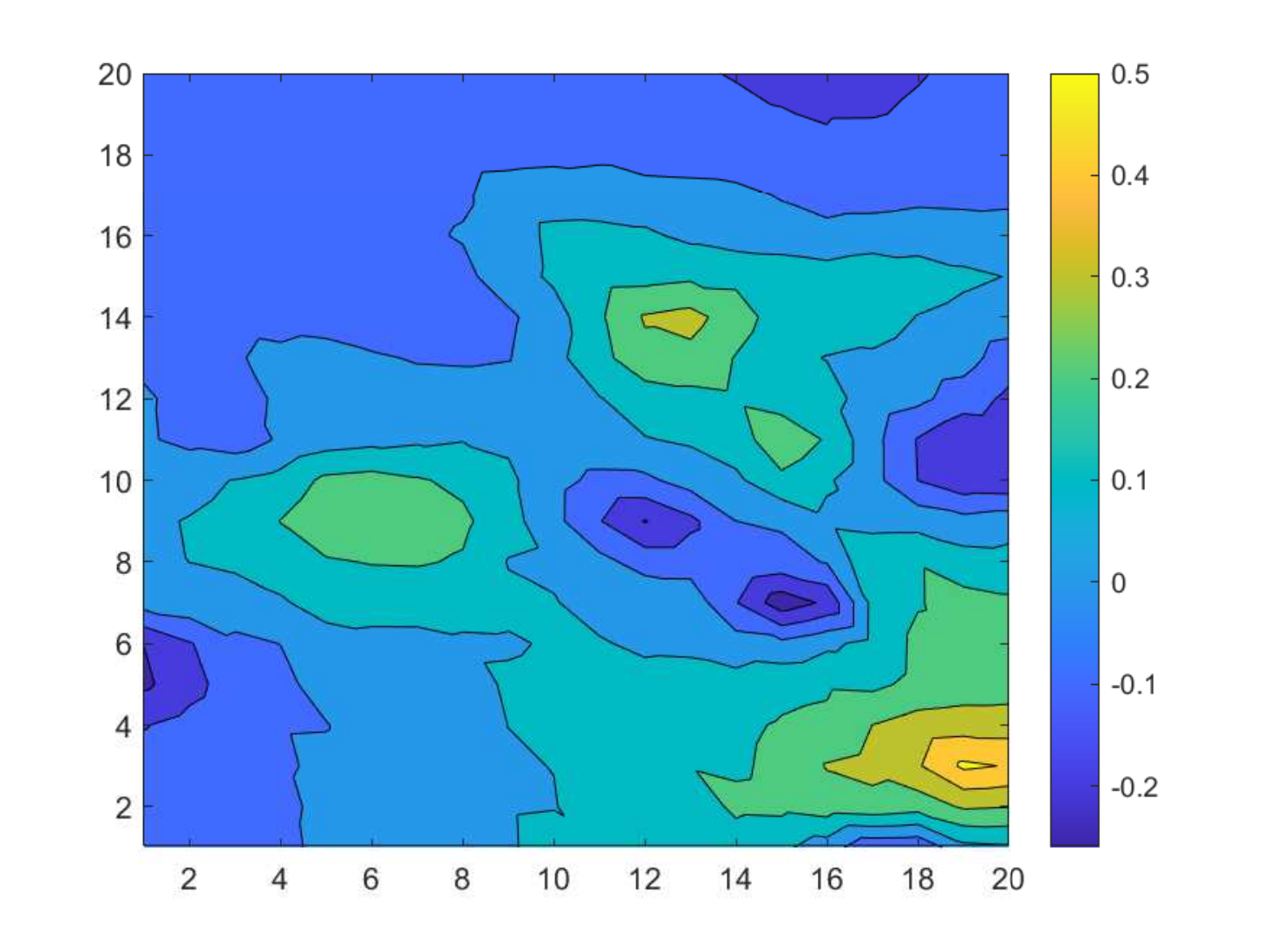}
        \hspace*{0.56cm}
         \includegraphics[width=2.5cm,height=2.5cm]{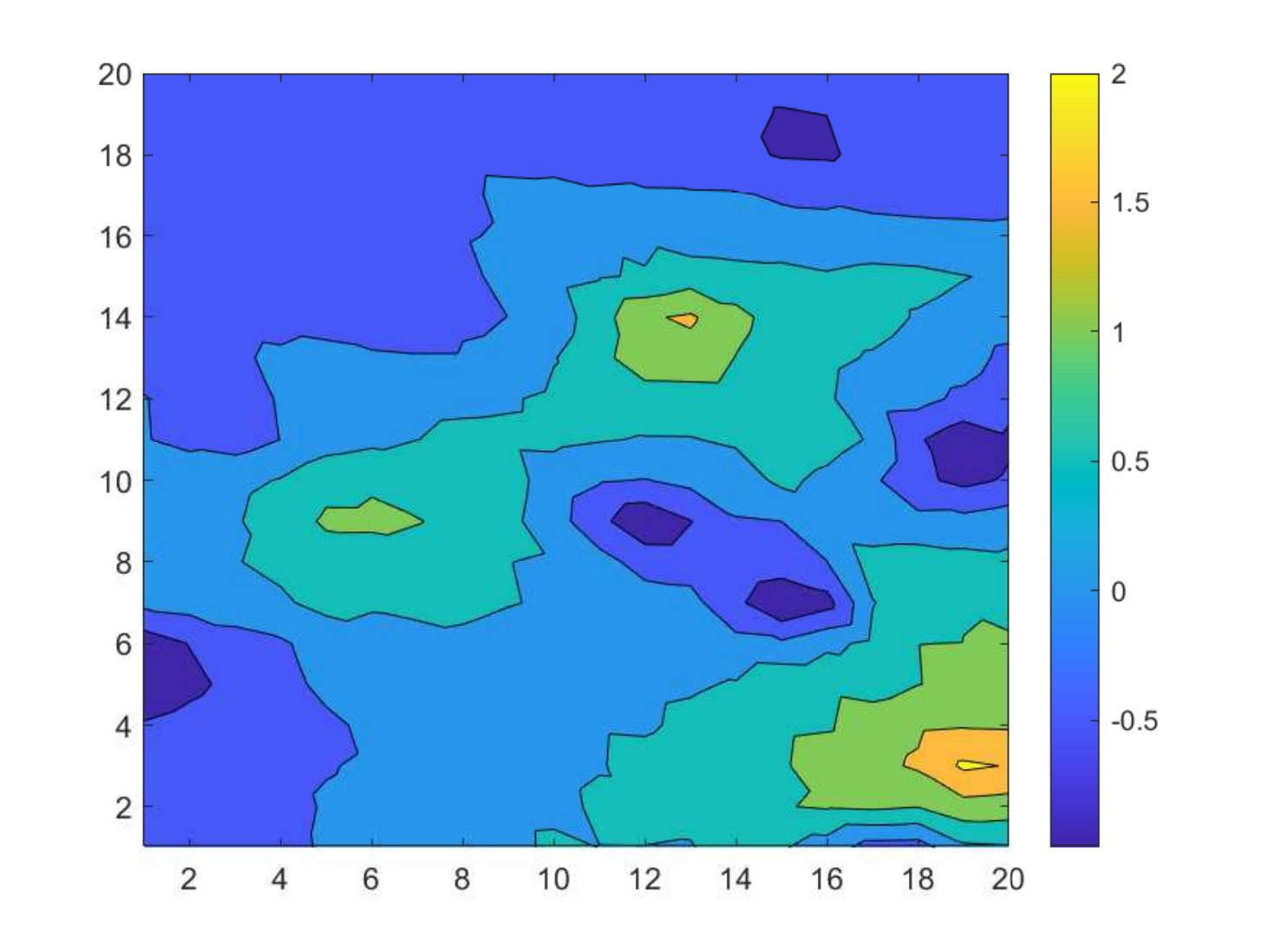}
        \includegraphics[width=2.5cm,height=2.5cm]{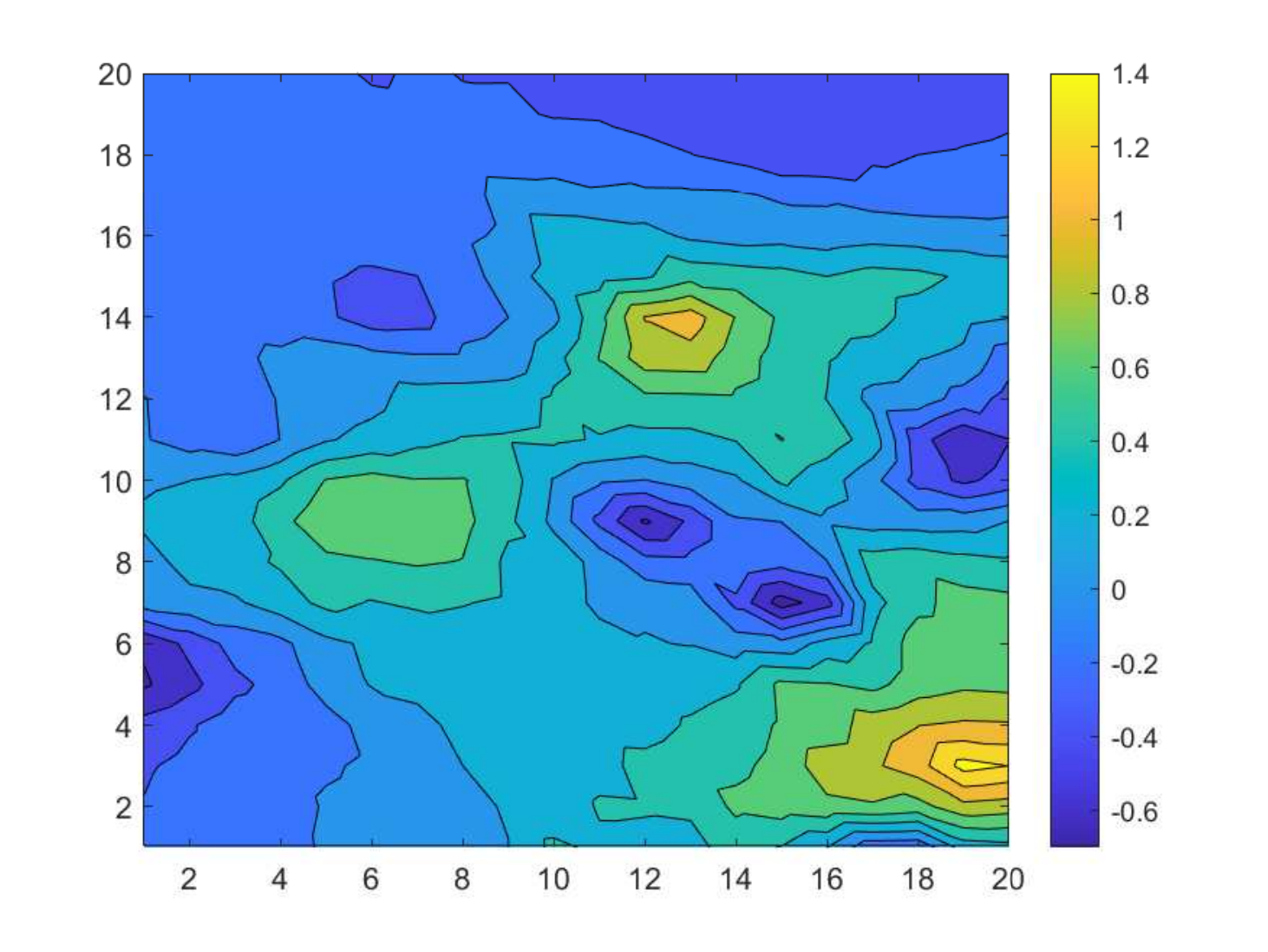}
        \includegraphics[width=2.5cm,height=2.5cm]{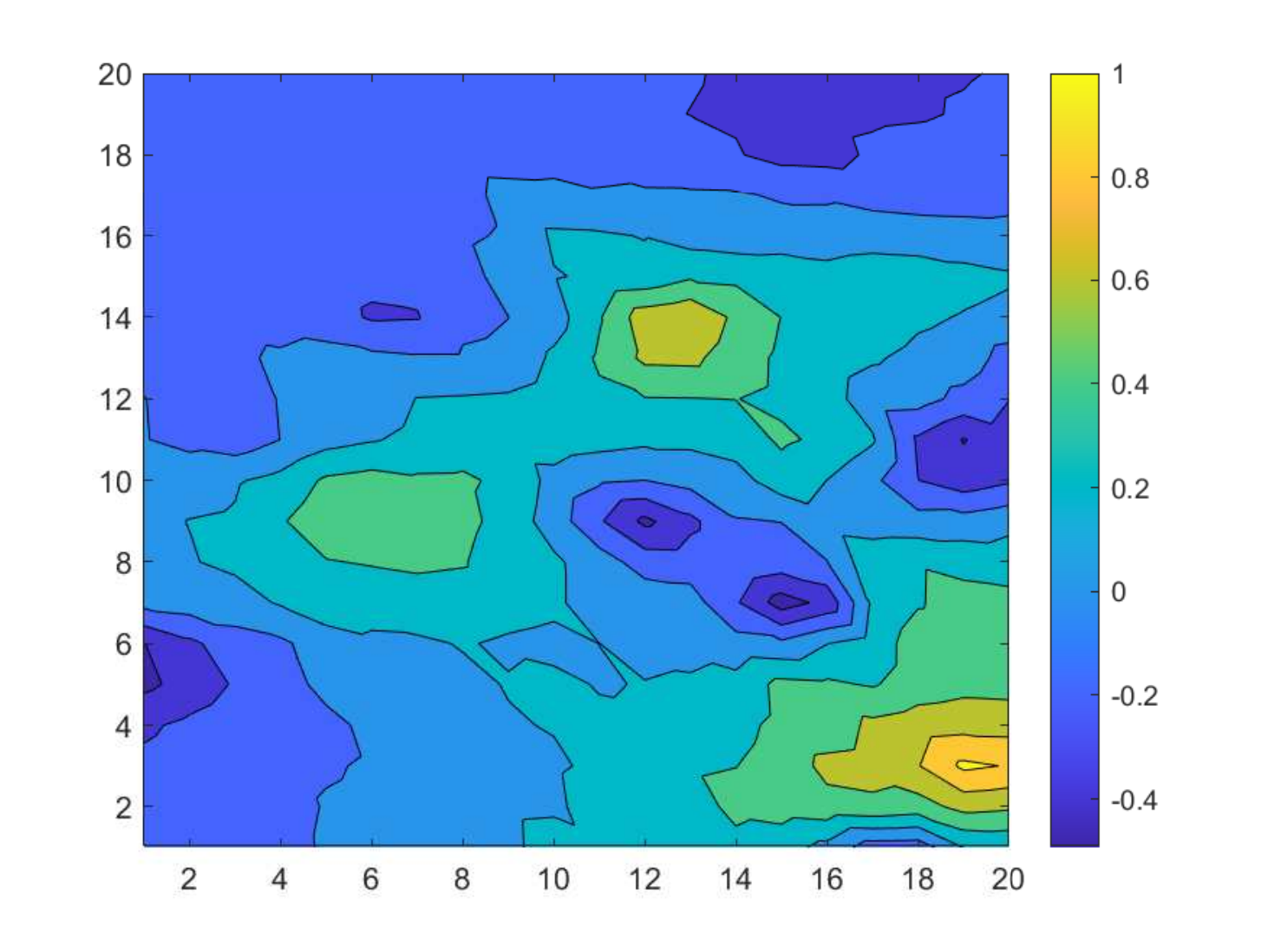}
        \includegraphics[width=2.5cm,height=2.5cm]{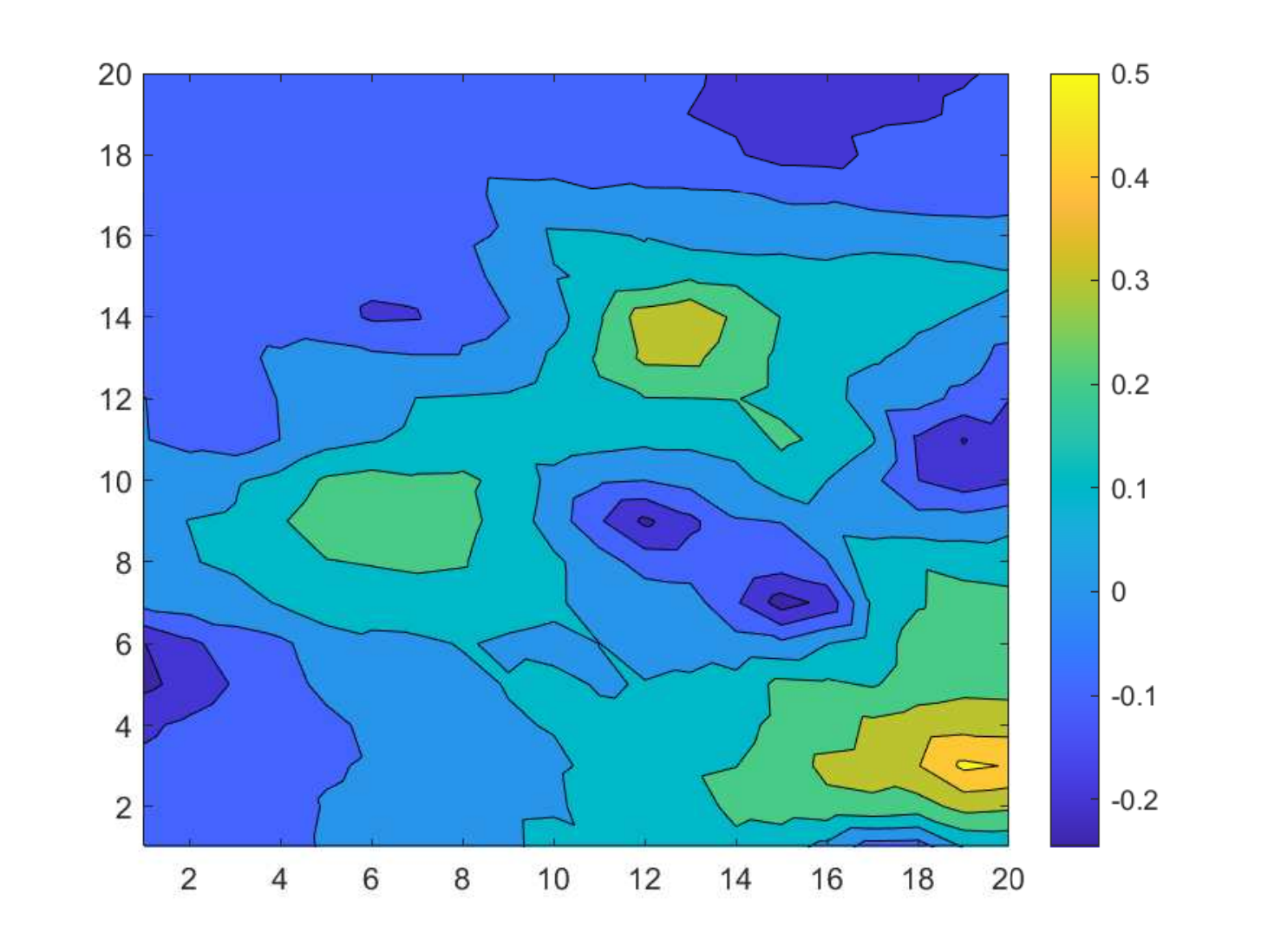}
        \hspace*{0.56cm}
         \includegraphics[width=2.5cm,height=2.5cm]{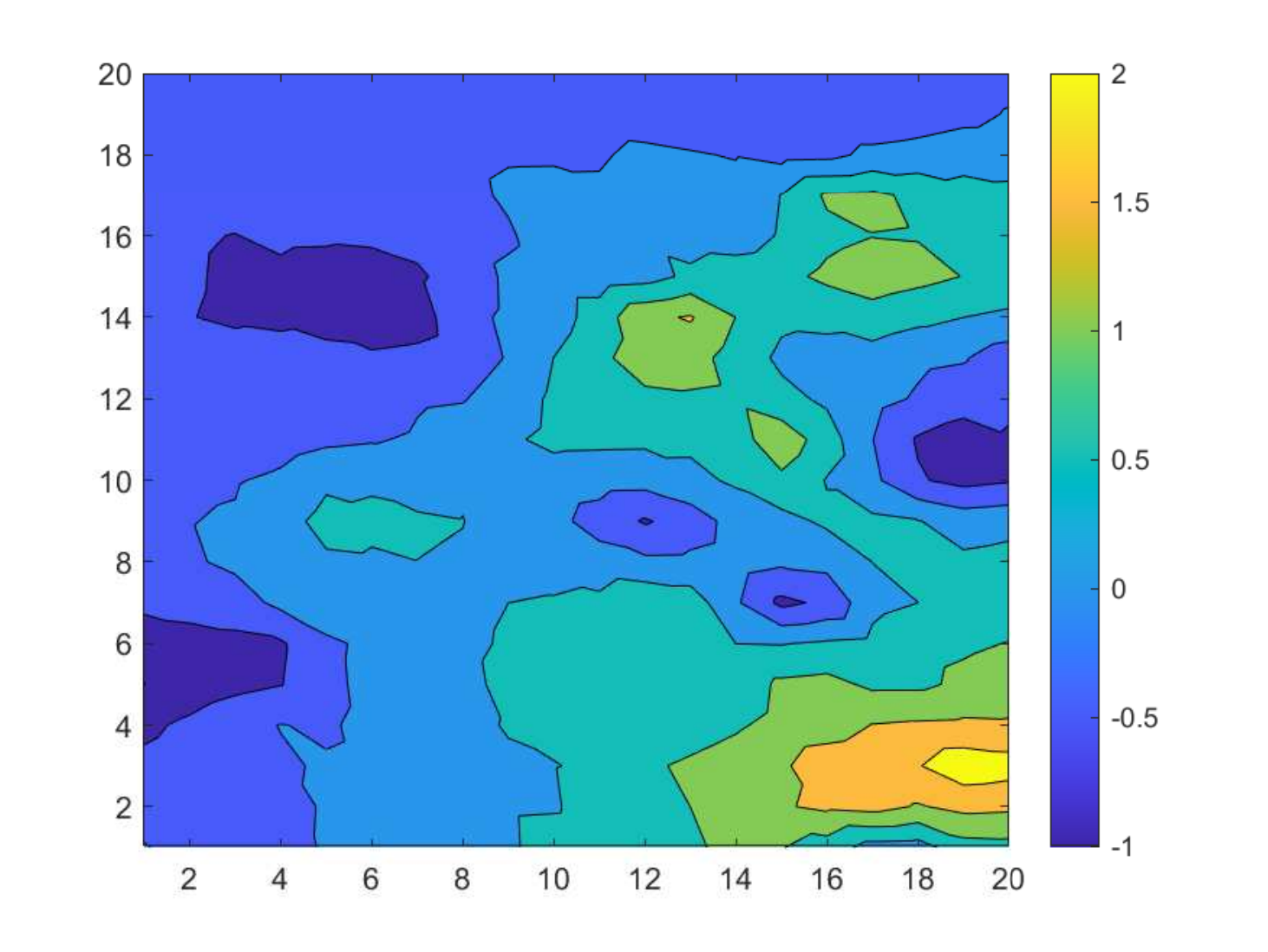}
        \includegraphics[width=2.5cm,height=2.5cm]{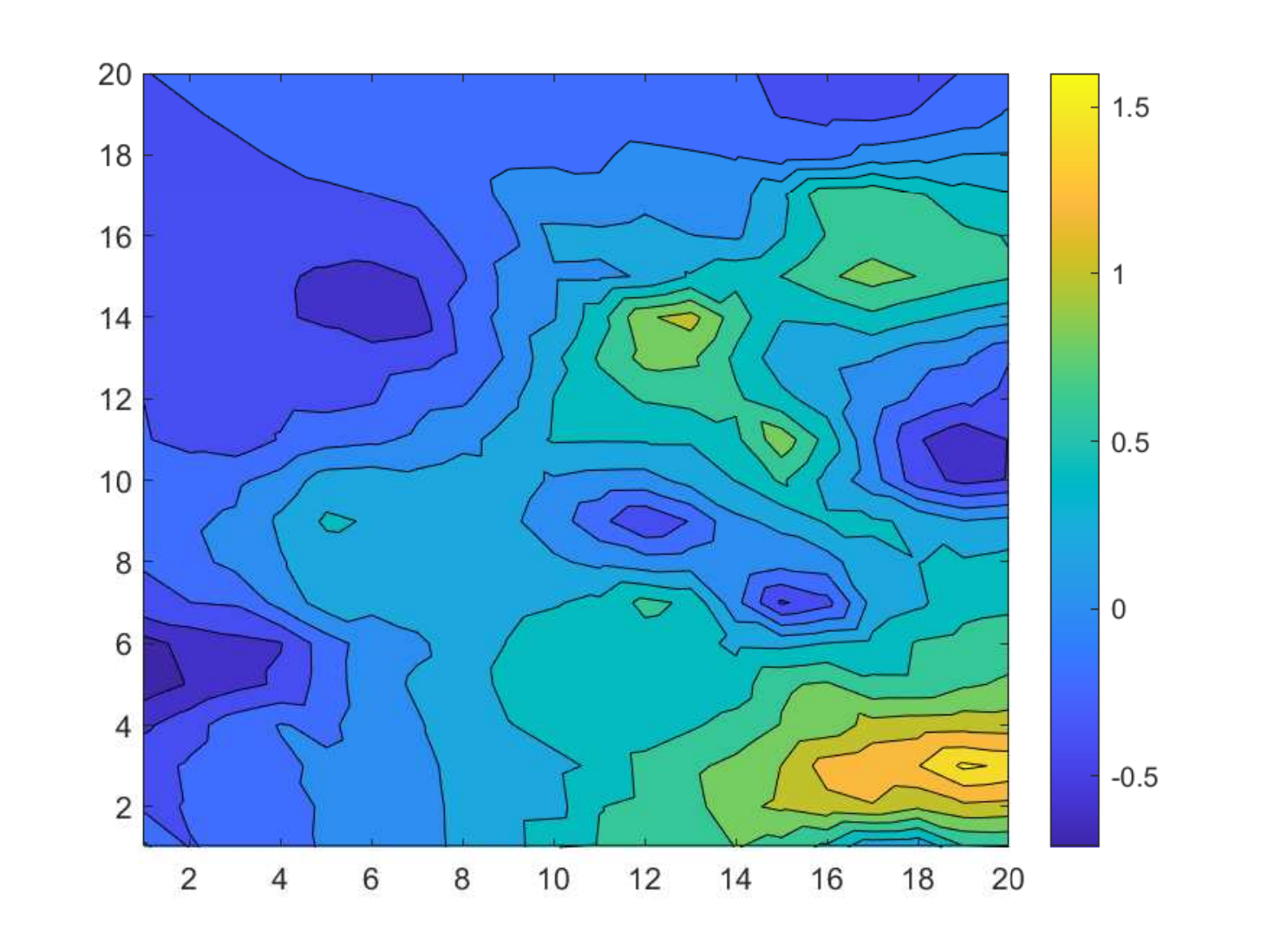}
        \includegraphics[width=2.5cm,height=2.5cm]{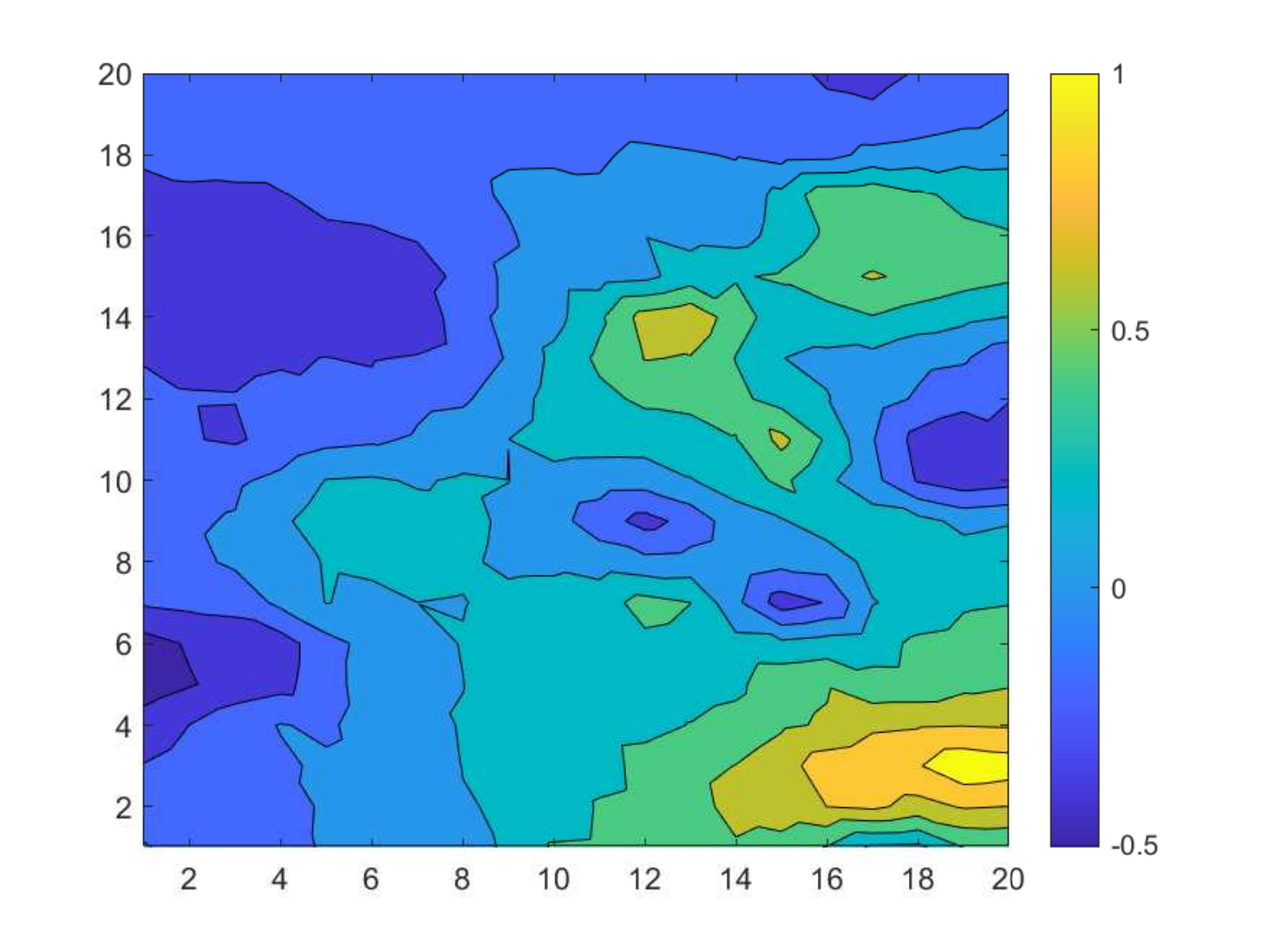}
        \includegraphics[width=2.5cm,height=2.5cm]{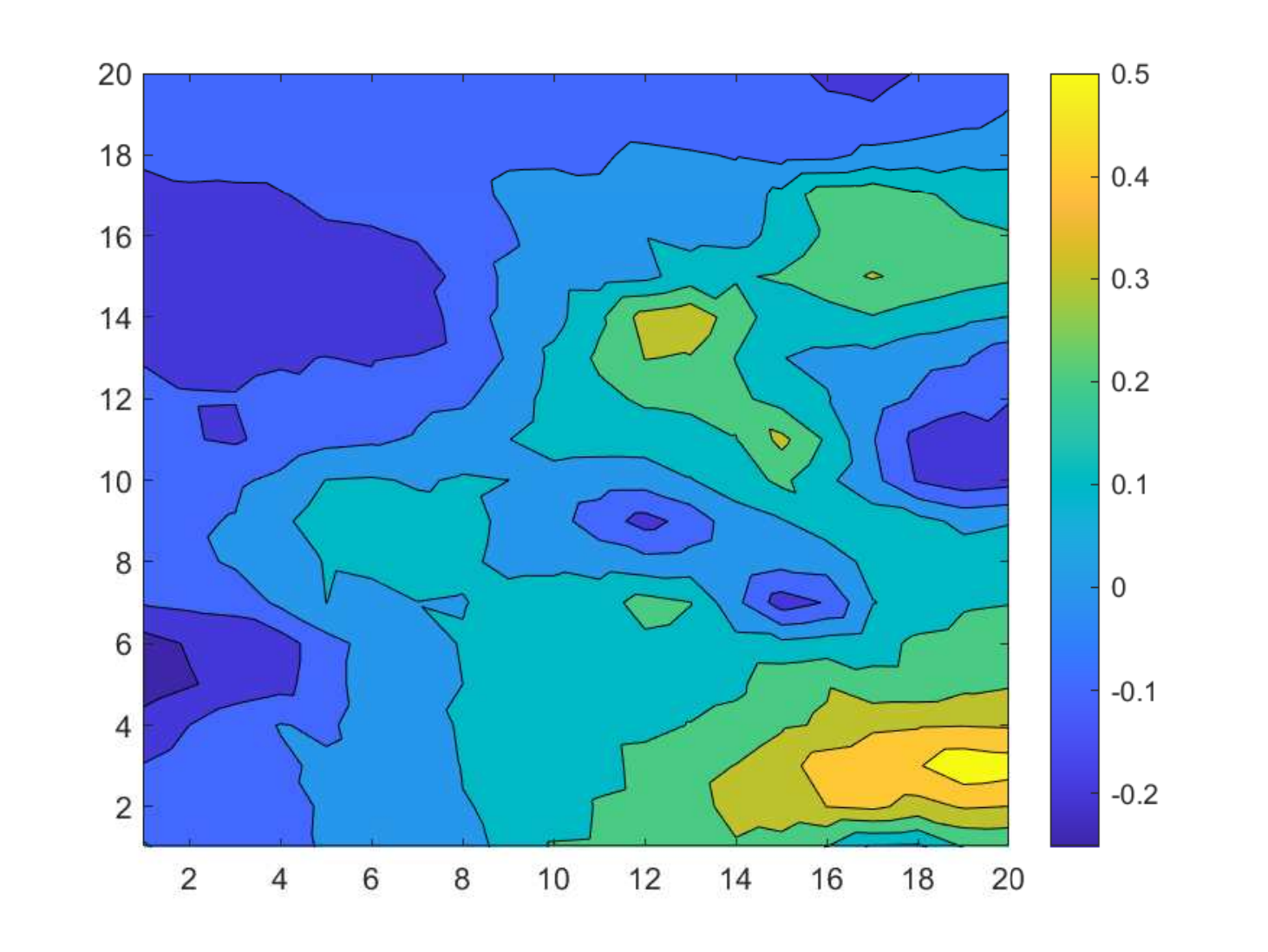}
        \hspace*{0.56cm}
        \includegraphics[width=2.5cm,height=2.5cm]{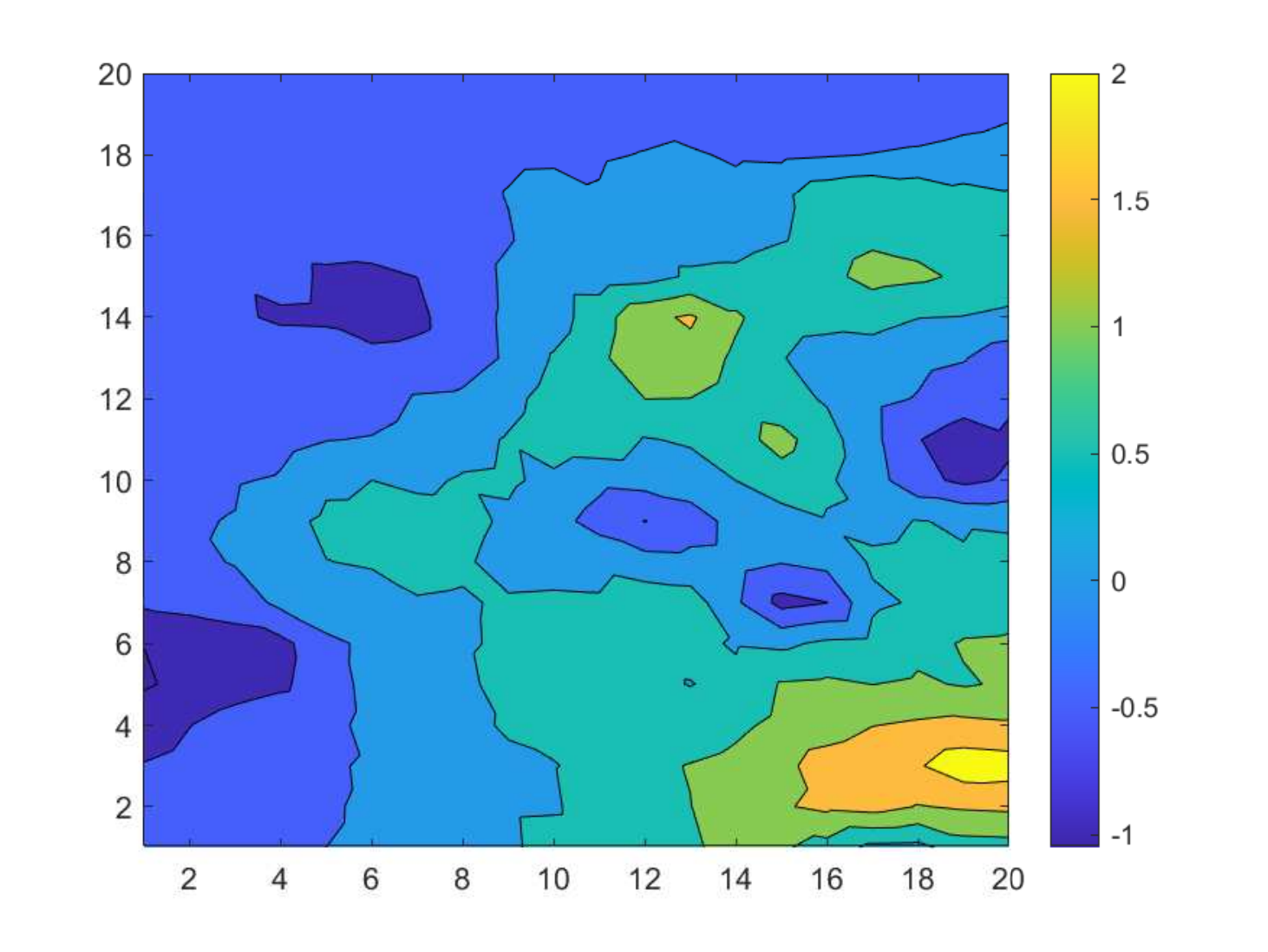}
        \includegraphics[width=2.5cm,height=2.5cm]{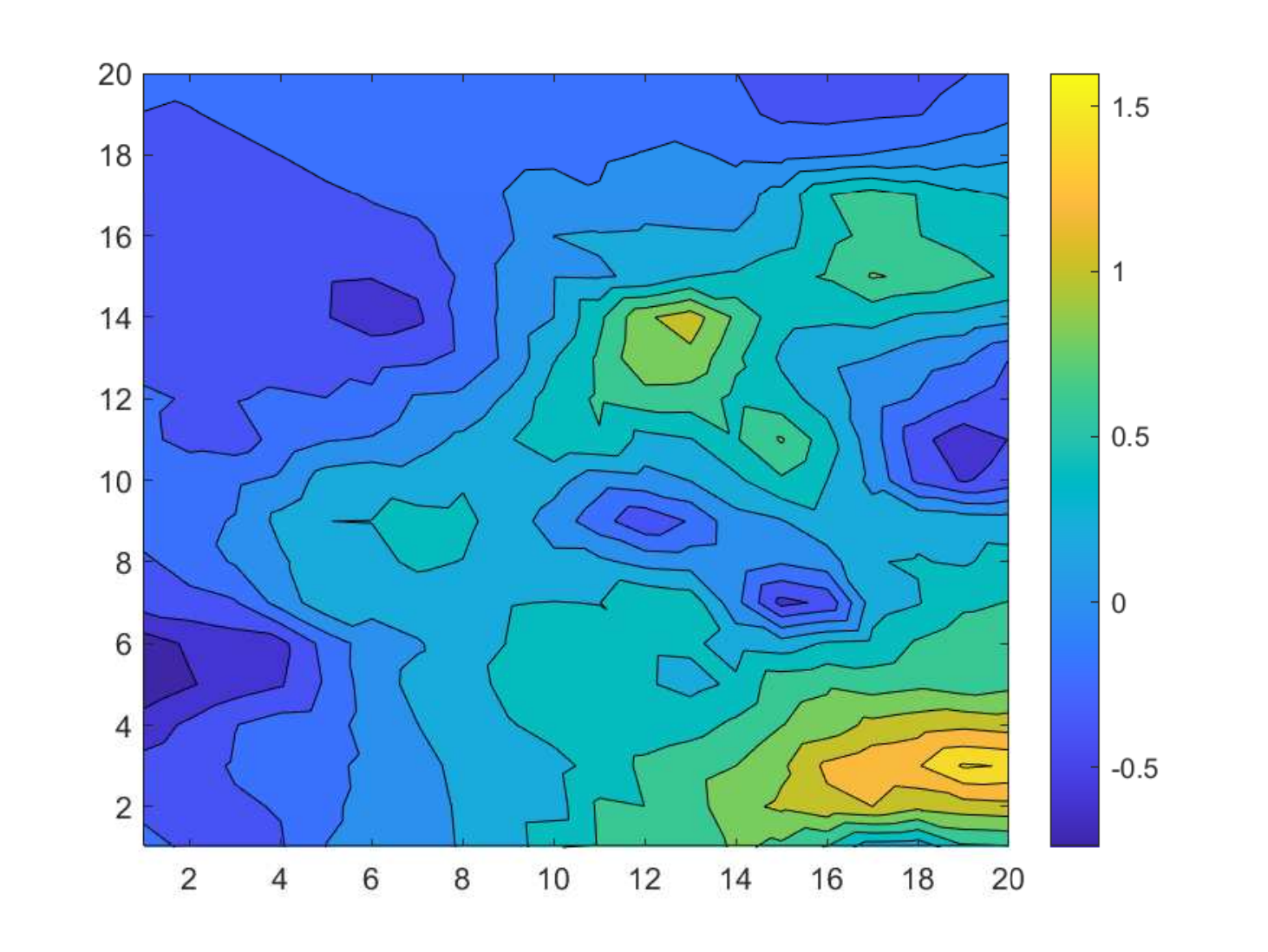}
        \includegraphics[width=2.5cm,height=2.5cm]{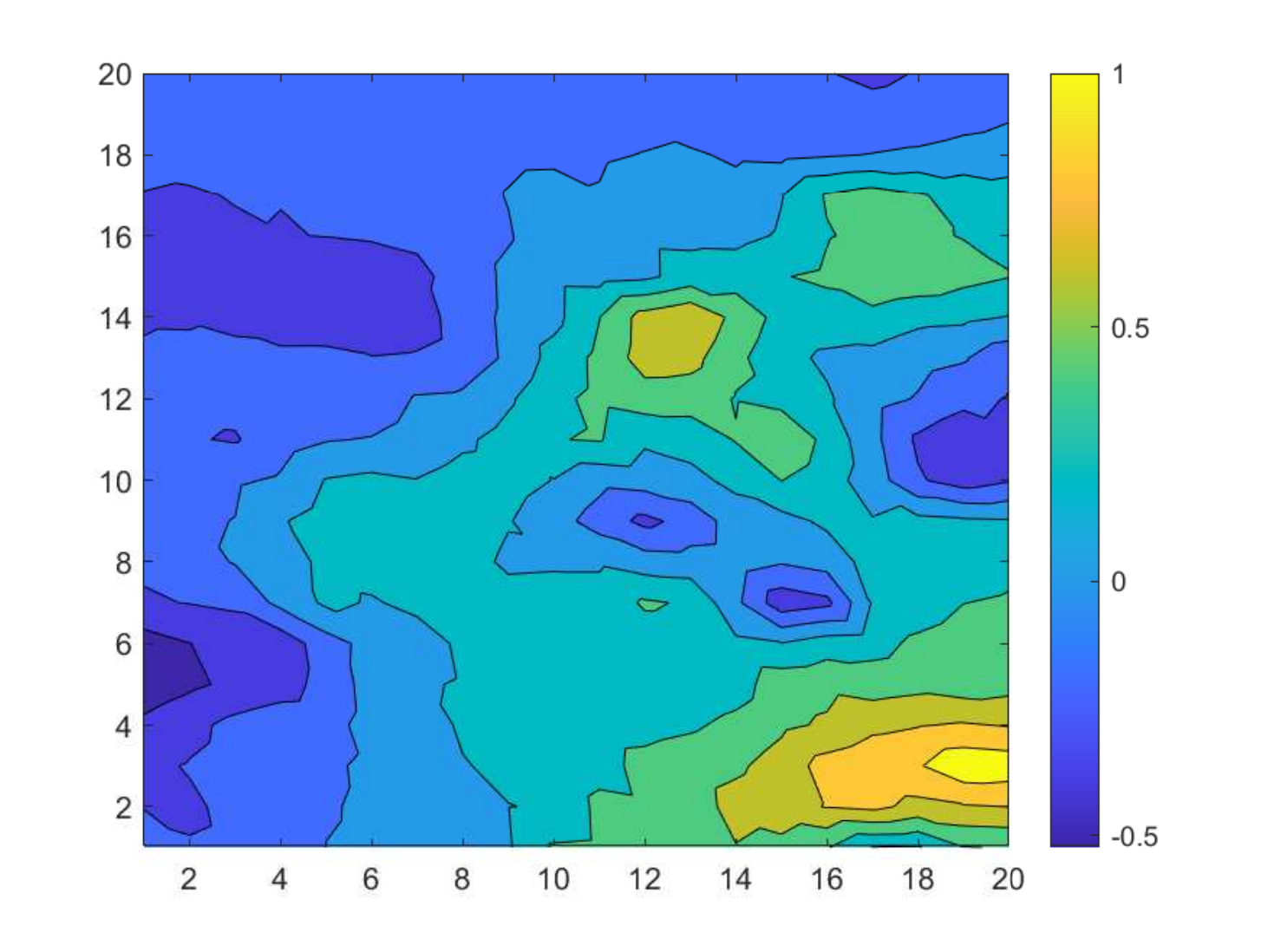}
        \includegraphics[width=2.5cm,height=2.5cm]{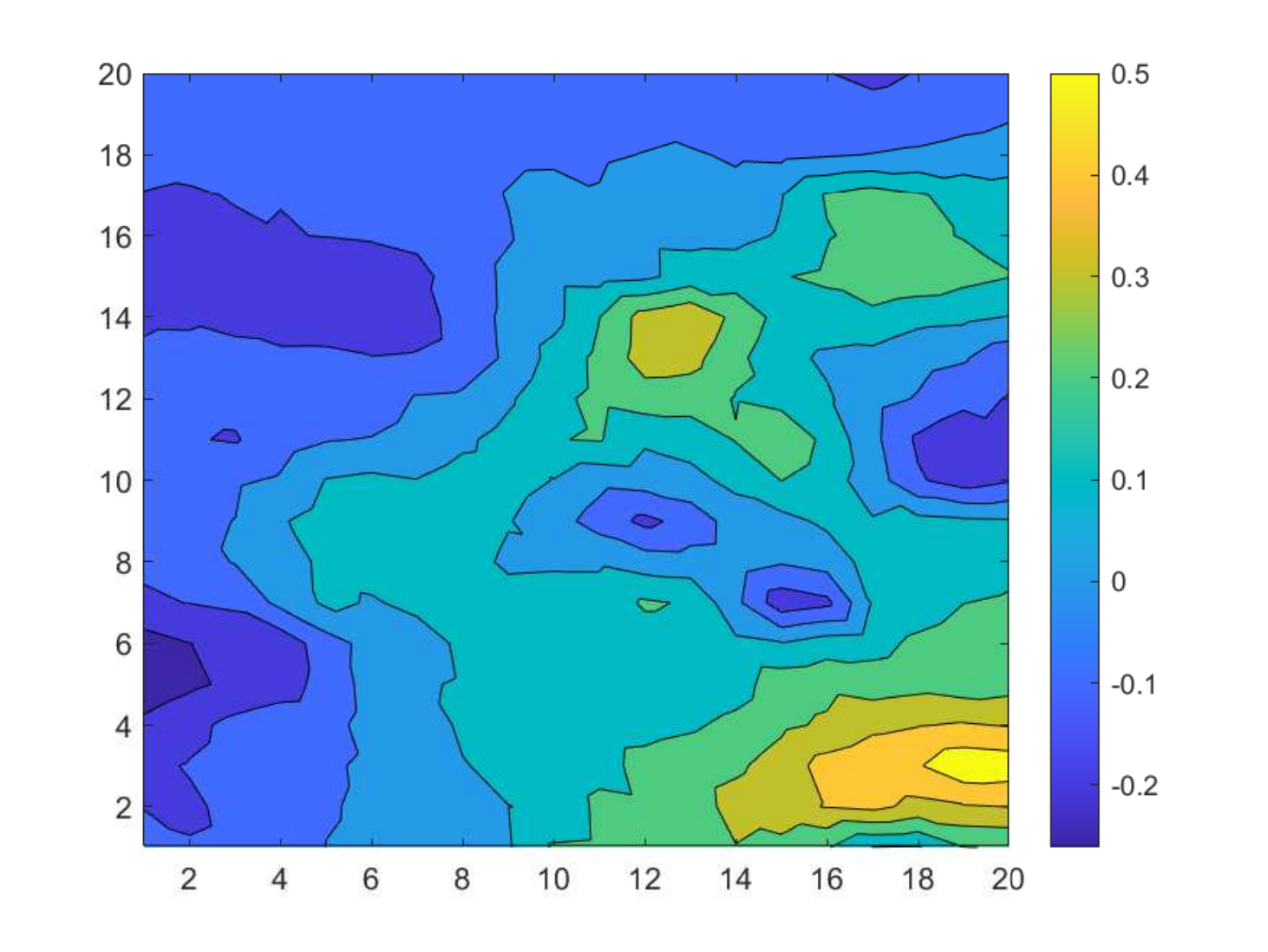}
 \caption{Contour plots of the observed log--intensity field at monthly times $t=108$ (top--row) and $t=216$ (third--row), and  the estimated log--intensity field at $t=108$ (second--row) and   $t=216$ (bottom--row). Both observed and estimated values at times $t=108,$ and  $t=216$  are displayed  through the scales $j=7,8,9,10$ (from  left to  right), in the Haar wavelet system, from the  temporal  interpolated and smoothed data over a $20\times 20 $  regular grid}
 \label{figa2}
 \end{figure}
\end{center}

\clearpage
\subsection{Validation results}

        Our multiscale spatial functional approach is now validated from the data. Specifically,  by leaving aside  the curves observed at the nodes  in a neighbourhood of the province defining the region of interest (the validation functional data set), equations (\ref{ps})--(\ref{SARHpp})  are computed  from the  remaining functional observations, spatially distributed at the neighbourhoods of the rest of the  Spanish provinces   (the  training functional data set). The corresponding multiscale SAR$\ell^{2}$(1) componentwise parameter estimators and predictors are then obtained, from the empirical wavelet reconstruction formula at resolution level 10 (truncated version of equation (\ref{SARHpp})). This process  is repeated  48 times. Thus,   the   cross-validation functional error is calculated as  the mean of the absolute functional  errors computed at each one of the  48 iterations.    The annual pointwise mean of the  computed    cross-validation functional  error can be  found in Table 4 above. The original and estimated  annually averaged  number of deaths at each province,    for each one of the years analysed, are also displayed in Figures \ref{fig3ccmapa} and \ref{fig4mapa}.
\begin{table}
\caption{ALOOCVE. Pointwise annually averaged
  cross-validation  errors.}
\label{LOOCVy}
\begin{tabular}{cccccc}
%\\
 \\
Year& ALOOCVE&Year&ALOOCVE&Year& ALOOCVE\\[5pt]
  1980 & 0.0247 & 1992  & 0.0118 & 2004  & 0.0132 \\
  1981  & 0.0144 & 1993  & 0.0130 & 2005  & 0.0117 \\
  1982  & 0.0112 & 1994  & 0.0163 & 2006  & 0.0135 \\
  1983  & 0.0125 & 1995  & 0.0159 & 2007  & 0.0140 \\
  1984  & 0.0144 & 1996  & 0.0111 & 2008  & 0.0118 \\
  1985  & 0.0122 & 1997  & 0.0099 & 2009  & 0.0113 \\
  1986  & 0.0126 & 1998  & 0.0108 & 2010  & 0.0143 \\
  1987  & 0.0155 & 1999  & 0.0141 & 2011  & 0.0131 \\
  1988  & 0.0161 & 2000  & 0.0167 & 2012  & 0.0122 \\
  1989  & 0.0144 & 2001  & 0.0161 & 2013  & 0.0115\\
  1990  & 0.0125 & 2002  & 0.0143 & 2014  & 0.0145\\
  1991  & 0.0118 & 2003  & 0.0140 & 2015 & 0.0221\\
\end{tabular}
\end{table}

\begin{figure}
\begin{center}
\includegraphics[width=12cm,height=12cm]{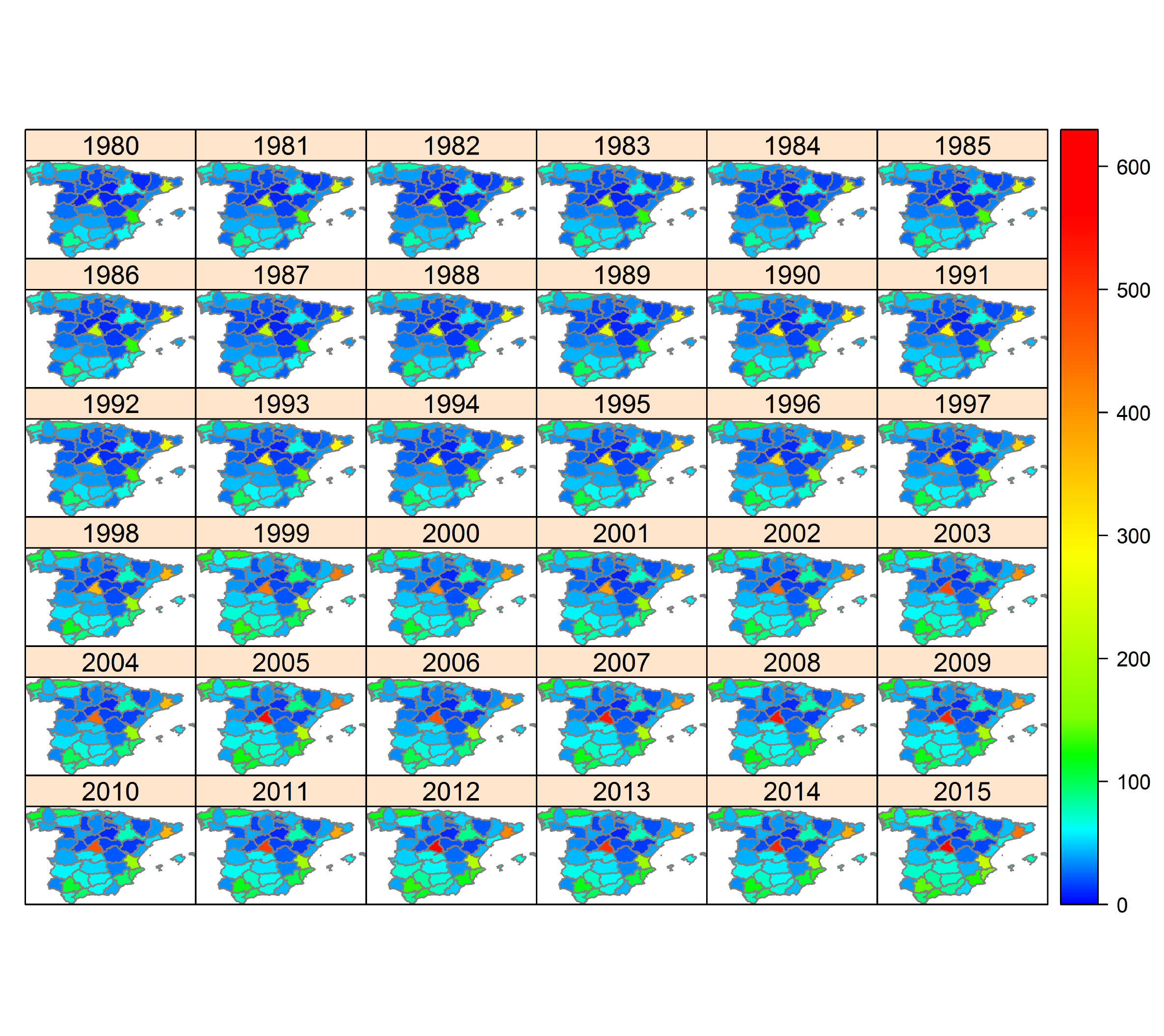}
\end{center}
\caption{Annually averaged   observed    number of respiratory disease deaths   at each one of  the 48 Spanish provinces from January 1980 to December 2015.}
\label{fig3ccmapa}
\end{figure}

\begin{figure}
\begin{center}
\includegraphics[width=12cm,height=12cm]{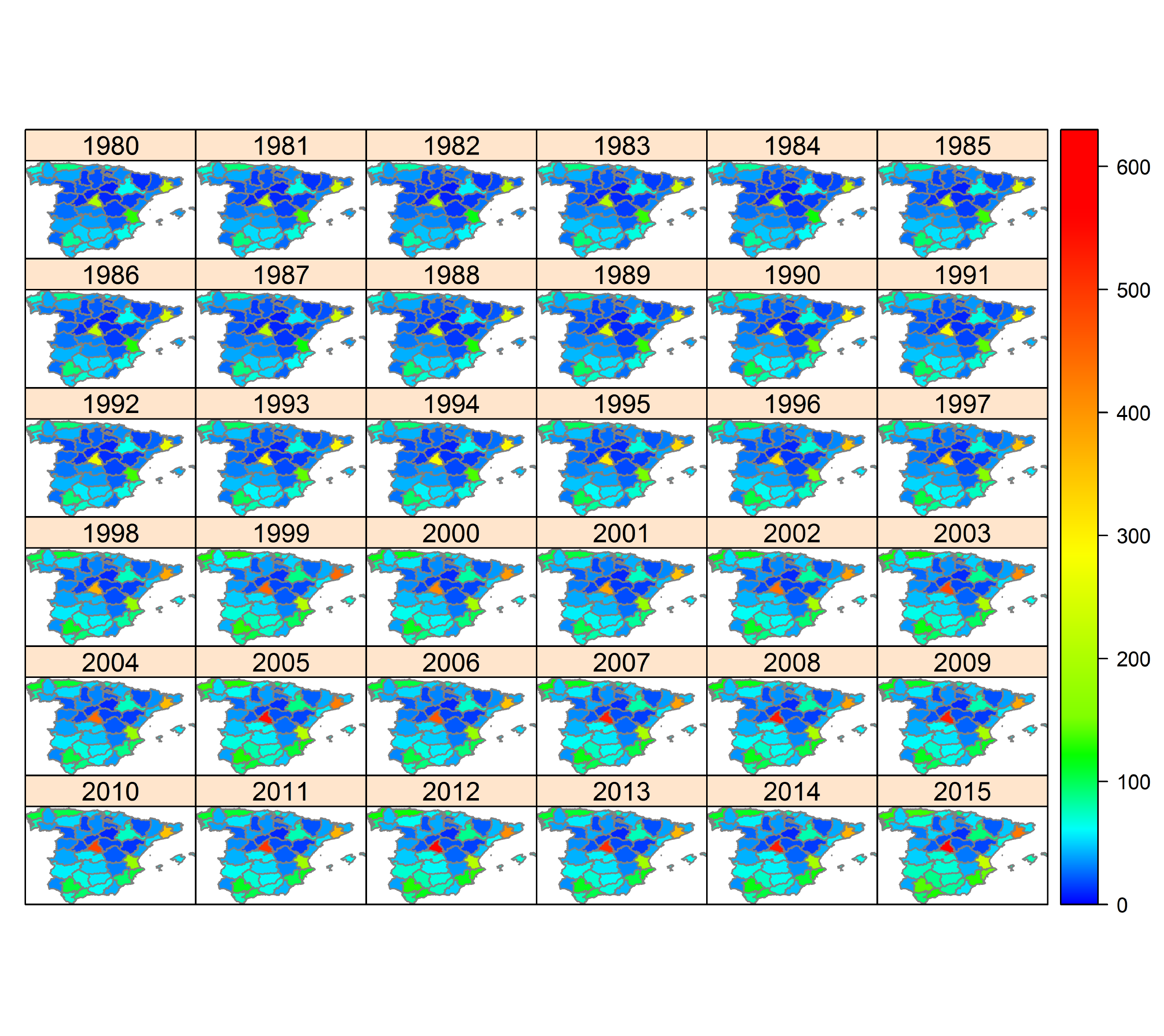}
\end{center}
\caption{Annually averaged  estimates of the number  of  respiratory disease deaths, at each one of the  48 Spanish provinces in the period  1980--2015.}\label{fig4mapa}
\end{figure}

\clearpage
\section{Concluding remarks}
\label{s7}
 The multiscale spatial functional prediction methodology presented allows heterogeneity analysis over different temporal scales of the log--intensity field. It is well--known that FDA preprocessing techniques (e.g., B--spline smoothing) usually hide or eliminate local variation at high resolution levels  (see, e.g., \cite{MullerStandmuler}; \cite{HorKok12};
  \cite{GB2016}, among others).
  The estimation approach adopted in this paper overcomes this limitation, providing a more flexible framework. Thus, a suitable choice of the scale where the    log--intensity field should be analysed can be performed, according to the    aims  of the  study and uncertainties in the counts associated with  the lack of sample information.

   The  infinite--dimensional parametric estimation approach proposed, based on  relative entropy in the spatial  spectral domain, through a multiscale analysis in time, does not require previous information about the parameter probability distribution, as  in the Bayesian framework. Furthermore heavy computational problems, arising in the latter framework (e.g., high--dimensional covariance matrices associated with latent Gaussian variables and  hyperparameters) are avoided with the presented estimation methodology.

  Our approach can be extended to the case where a multiresolution analysis is also performed in space, for  approximation of the  hidden spatial continuous functional log--intensity process driving the counts, as an alternative to the usual spatial B-spline smoothing techniques. The resulting  approach allows   heterogeneity analysis through temporal and spatial scales, providing a multiresolution approximation of  space--time interaction affecting the evolution of  the log--intensity process. This  topic constitutes the subject of a subsequent paper.

%%%%%%%%%%%%%%%%%%%%%%%%%%%%%%%%%%%%%%%%%%%%%%%%%%%%%%%%%%%%%%%%%%%%%%%%%%%%%%%%%%%%%%%%%%%%%%%%%%%%%%%%%%%%%%%%%%%%%%%%%%%%

%%%%%%%%%%%%%%%%%%%%%%%%%%%%%%%%%%%%%%%%%%%%%%%%%%%%%%%%%%%%%%%%%%%%%%%%%%%%%%%%%%%%%%%%%%%%%%%%%%%%%%%%%%%%%%%%%%%%%%%%%%%%

\noindent {\bfseries Acknowledgements}

\medskip

This work  was supported  by MCIU/AEI/ERDF, UE grant
PGC2018-099549-B-I00 (M.D. Ruiz--Medina, M.P. Fr\'{\i}as, A. Torres-Signes),  PID2019-107392RB-100 (J. Mateu), and by grant A-FQM-345-UGR18  (M.D. Ruiz--Medina, M.P. Fr\'{\i}as, A. Torres-Signes) cofinanced by ERDF Operational Programme 2014-2020 and the Economy and Knowledge Council of the Regional Government of Andalusia, Spain.

%%%%%%%%%%%%%%%%%%%%%%%%%%%%%%%%%%%%%%%%%%%%%%%%%%%%%%%%%%

\end{document}